\documentclass[12pt,letterpaper]{JHEP3}

\usepackage{amscd,amsmath,amssymb,amsfonts,xspace,mathrsfs,amsthm}
\usepackage{color}
\let\normalcolor\relax
\usepackage{bbold}
\usepackage{latexsym}
\usepackage{graphicx}
\usepackage{dsfont}
\usepackage{color}
\usepackage{longtable}

 \newtheorem{theorem}{Theorem}
  \newtheorem{proposition}{Proposition}
  \newtheorem{lemma}{Lemma}
  \newtheorem{conj}{Conjecture}

\numberwithin{equation}{section}

\def\lfig#1#2#3#4#5{
\begin{figure}[t]
 \centerline{\includegraphics[width=#3]{#2}}
 \vspace{#5}
  \caption{#1 \label{#4}}
 \end{figure}
}

\def\det{\,{\rm det}\, }
\def\diag{{\rm diag}}
\def\sign{{\rm sgn}}

\def\Ch{{\rm Ch}}

\def\Sym{\,{\rm Sym}\, }

\def\Im{\,{\rm Im}\,}

\def\({\left(}
\def\){\right)}
\def\[{\left[}
\def\]{\right]}
\def\<{\left\langle}
\def\>{\right\rangle}
\def\hf{{1\over 2}}
\def\haf{\textstyle{1\over 2}}

\newcommand{\de}{\mathrm{d}}

\newcommand{\I}{\mathrm{i}}

\newcommand{\cL}{\mathcal{L}}
\newcommand{\cD}{\mathcal{D}}

\def\vrh{\varrho}
\def\vph{\varphi}
\newcommand{\p}{\partial}

\newcommand{\cV}{\mathcal{V}}

\newcommand{\cM}{\mathcal{M}}

\newcommand{\cN}{\mathcal{N}}

\newcommand{\cJ}{\mathcal{J}}

\DeclareSymbolFont{AMSa}{U}{msa}{m}{n}
\DeclareSymbolFont{AMSb}{U}{msb}{m}{n}
\DeclareMathSymbol{\fieldR}{\mathalpha}{AMSb}{"52}

\newcommand{\Fb}{{\mathbb F}}
\newcommand{\Bb}{{\mathbb B}}

\newcommand{\cI}{\mathcal{I}}
\newcommand{\cO}{\mathcal{O}}
\newcommand{\cQ}{\mathcal{Q}}

\newcommand{\pa}{\partial}
\newcommand{\nn}{\nonumber}

\newcommand{\eps}{\epsilon}
\newcommand{\veps}{\varepsilon}
\newcommand{\IT}{\mathds{T}}
\newcommand{\IR}{\mathds{R}}
\newcommand{\IB}{\mathds{B}}
\newcommand{\IC}{\mathds{C}}
\newcommand{\IZ}{\mathds{Z}}

\newcommand{\IH}{\mathds{H}}

\newcommand{\IP}{\mathds{P}}
\newcommand{\IN}{\mathds{N}}

\newcommand{\sgn}{\mbox{\rm sgn}}

\newcommand{\q}{\mbox{q}}

\newcommand{\td}{\tilde d}
\newcommand{\trho}{\tilde\rho}

\def\bea{\begin{eqnarray}}
\def\eea{\end{eqnarray}}
\def\be{\begin{equation}}
\def\ee{\end{equation}}
\def\ba{\begin{align}}
\def\ea{\end{align}}
\def\bse{\begin{subequations}}
\def\ese{\end{subequations}}

\def\ba{\bar a}

\def\btau{\bar \tau}
\def\hmu{\hat\mu}
\def\bw{\bar w}

\newcommand{\cB}{\mathcal{B}}

\def\cij#1{c}
\def\ci#1{c}

\def\XXint#1#2#3{{\setbox0=\hbox{$#1{#2#3}{\int}$}
\vcenter{\hbox{$#2#3$}}\kern-.5\wd0}}

\def\gamD#1{\tilde\gamma}

\def\CY{\mathfrak{Y}}

\DeclareMathOperator{\Erf}{Erf}

\def\cl0{\tilde c_0}

\newcommand{\bpt}{{\scriptstyle\boldsymbol{*}}}
\newcommand{\bfw}{{\boldsymbol w}}

\newcommand{\bfv}{{\boldsymbol v}}

\newcommand{\bfk}{{\boldsymbol k}}

\newcommand{\bfp}{{\boldsymbol p}}

\newcommand{\bfx}{{\boldsymbol x}}
\newcommand{\bfy}{{\boldsymbol y}}

\newcommand{\bfmu}{{\boldsymbol \mu}}
\newcommand{\bfnu}{{\boldsymbol \nu}}
\newcommand{\bfrho}{{\boldsymbol \rho}}

\newcommand{\bfxi}{{\boldsymbol \xi}}
\newcommand{\bfzeta}{{\boldsymbol \zeta}}

\newcommand{\bfDelta}{{\boldsymbol \Delta}}

\def\sE{\mathscr{E}}
\def\Zv{\mathscr{Z}}

\def\Zv{\mathscr{Z}}

\def\Rv{\mathscr{R}}

\def\bOm{\bar\Omega}

\def\whh{\widehat h}

\def\whg{\widehat g}

\def\whPhi{\widehat\Phi}

\def\hrVW{h^{\rm VW,ref}}

\def\sEp{\sE^{(+)}}
\def\sEf{\sE^{(0)}}

\def\Lat{\mathbf{\Lambda}}
\def\Latc{\Lambda_S^v}
\def\Latci#1{\Lambda_{#1}^v}

\def\ptt{\mathscr{N}}
\def\bfptt{\boldsymbol{\mathfrak{p}}}

\def\cl{c^{(\ell)}}

\def\vu{\mathbb{u}}
\def\vv{\mathbb{v}}
\def\vg{\mathbb{g}}
\def\vd{\mathbb{d}}

\newcommand{\tvtheta}{\lefteqn{\smash{\mathop{\vphantom{<}}\limits^{\;\sim}}}\vartheta}
\newcommand{\cvtheta}{\lefteqn{\smash{\mathop{\vphantom{<}}\limits^{\;\circ}}}\vartheta}

\def\gama{\check\gamma}

\def\Vop{\mathbb{V}}

\def\sN{M}
\def\y{y}

\def\nb{b^{\,0}}
\def\nv{v_0}
\def\cv{v_1}
\def\qv{q^v}

\def\cdv{r_1}

\def\under#1#2{\mathop{#1}\limits_{#2}}

\def\bcup{\mathop{\Large\vphantom{A_a}\mbox{$\cup$}}}

\title{Vafa-Witten invariants from modular anomaly}

\author{Sergei Alexandrov
\\
{\it
Laboratoire Charles Coulomb (L2C), Universit\'e de Montpellier,
CNRS, F-34095, Montpellier, France}\\
{\it
Department of High Energy and Elementary Particle Physics,
Saint Petersburg State University,
7/9 Universitetskaya nab., St. Petersburg 199034, Russia}\\

\vspace*{2mm} {\tt e-mail:
\email{sergey.alexandrov@umontpellier.fr}
}

\vspace*{-3mm}

}

\abstract{Recently, a universal formula for a non-holomorphic modular completion of the generating functions
of refined BPS indices in various theories with $N=2$ supersymmetry has been suggested.
It expresses the completion through the holomorphic generating functions of lower ranks.
Here we show that for $U(N)$ Vafa-Witten theory on Hirzebruch and del Pezzo surfaces this formula can be used to extract
the holomorphic functions themselves, thereby providing the Betti numbers of instanton moduli spaces on such surfaces.
As a result, we derive a closed formula for the generating functions and their completions for {\it all} $N$.
Besides, our construction reveals in a simple way instances of fiber-base duality, which can be used to
derive new non-trivial identities for generalized Appell functions. It also suggests the existence of new invariants,
whose meaning however remains obscure.
}

\begin{document}

\section{Introduction}
\label{sec-intro}

The determination of BPS indices is an outstanding problem in both physics and mathematics.
On the physics side, the indices encode the BPS spectrum in supersymmetric gauge and string theories,
which provides an important information about their low energy effective theories and
quantum corrections to physical observables. On the mathematical side, they often turn out to coincide
with various topological invariants of the manifolds that the corresponding physical theory is defined on.
A typical example is given by a (generalized) Donaldson-Thomas (DT) invariant $\Omega(\gamma)$ of a Calabi-Yau (CY)
threefold $\CY$, which on one hand encodes the microscopic entropy of a black hole with charge $\gamma$ in
string theory compactified on $\CY$, and on the other hand, if $\CY$ allows a non-trivial local limit,
counts the number of BPS states of the same charge in a supersymmetric gauge theory resulting from decoupling gravity.

Remarkably, sometimes the BPS indices can be organized into generating functions possessing
some beautiful symmetry properties.
Often mysterious from the mathematical definition of the BPS indices, these properties
can be explained using a proper physical interpretation.
For instance, the DT invariants $\Omega(\gamma)$ supported on an irreducible divisor $\cD\subset\CY$ define a function
transforming as a modular form under $SL(2,\IZ)$ \cite{Maldacena:1997de,Gaiotto:2005gf,deBoer:2006vg}.
Whereas the origin of the modular group is obscure in the Calabi-Yau geometry, it can be traced back to
the S-duality of type IIB string theory or to the torus appearing in the description of the same physical system
as M-theory compactified on $\CY\times T^2$.

Another example is given by Euler numbers of the moduli spaces of semi-stable sheaves on a complex surface $S$,
which have a physical interpretation as moduli spaces of instantons in the topologically twisted $\cN=4$
super-Yang-Mills (SYM), known as Vafa-Witten (VW) theory, defined on this surface \cite{Vafa:1994tf}.
The generating functions of these numbers, called also Vafa-Witten invariants,
appear as modular forms. This fact can be understood as a consequence of S-duality of the supersymmetric gauge theory.
In fact, this example is closely related to the previous one since for non-compact CY
threefolds given by the canonical bundle over $S$, the DT invariants supported on the divisor $\cD=N[S]$
coincide with the VW invariants of this surface for gauge group $U(N)$
\cite{Minahan:1998vr,Alim:2010cf,Gholampour:2017bxh}.

The modular symmetry is so restrictive that it can be used to find the actual values of the BPS indices
defining modular generating functions. For instance, the Rademacher expansion allows to compute
all Fourier coefficients of a modular form of negative weight in terms of just its first few coefficients
(the so-called polar terms) \cite{Rademacher:1938,Rademacher:1939,Niebur:1974}.
However, in many interesting cases the modular properties of the generating functions are not so simple and
these functions acquire a modular anomaly.
In the above examples this is the case when the divisor $\cD$ is {\it reducible} and when the surface $S$ has $b_2^+(S)=1$.
Remarkably, this anomaly is typically of a very special type
which implies that the generating functions are examples of {\it mixed mock} modular forms
or their higher depth generalizations \cite{Vafa:1994tf,Alexandrov:2016tnf,Alexandrov:2018lgp}.
Although in some simple cases a generalization of the Rademacher expansion can still be elaborated
\cite{Bringmann:2011,Bringmann:2010sd,Ferrari:2017msn,Bringmann:2018cov}, in general this appears to be out of reach.

In this paper we propose an alternative method to find BPS indices which is similar
to the one used to solve the topological string in \cite{Huang:2015sta,Gu:2017ccq}.
The idea is to trade the modular anomaly for a holomorphic anomaly, and then to fix
the holomorphic ambiguity in the resulting solution using restrictions from modularity and regularity.
The first step implies finding a modular completion of the original generating function,
i.e. its non-holomorphic modification which transforms as a true (vector valued) modular form.
Note that this completion is highly important by itself because typically the modular symmetry is more fundamental
than holomorphicity and physical quantities are always expected to be expressed through the completed modular functions.

Recently, using the string theory interpretation,
a general formula has been found for the modular completion of the generating functions of DT invariants $\Omega(\gamma)$
assigned to a divisor $\cD$ in arbitrary compact CY
and evaluated in the so-called attractor chamber of the moduli space \cite{Alexandrov:2018lgp}.
It expresses the completion $\whh_{\cD,\mu}(\tau,\bar\tau)$, where $\mu\in\Lambda/\Lambda^*$ (with $\Lambda=H_4(\CY,\IZ)$)
keeps track of the residual flux after spectral flow, as an expansion
in terms of products of the holomorphic functions $h_{\cD_i,\mu_i}(\tau)$ such that $\sum \cD_i=\cD$.
Later in \cite{Alexandrov:2019rth}, this solution was generalized to include
a complex refinement parameter $y=e^{2\pi\I z}$, in which case the formula for the completion even simplifies,
as well as extended to the case of non-compact CYs. Due to the relation mentioned above,
this provided a modular completion for the generating functions $\hrVW_{N,\mu}(\tau,z)$
of (refined) VW invariants of $S$ with $b_2^+(S)=1$
and $b_1(S)=0$\footnote{The second condition is needed to ensure that $S$ is rigid inside $\CY$.}
for gauge group $U(N)$ of {\it any} rank $N$, evaluated in the so called ``canonical" chamber of the moduli space,
which corresponds to the attractor chamber of the CY geometry.
From now on we restrict ourselves to this case and to avoid cluttering, we drop the label ``VW,ref"
so that $h_{N,\mu}$ will denote the generating functions of refined VW invariants defined below in \eqref{defhVWref}.

The construction of \cite{Alexandrov:2019rth} gives $\whh_{N,\mu}$ in terms of $h_{N_i,\mu}$, $N_i\le N$, and
ensures that it transforms as a vector valued Jacobi form with a certain multiplier system and with weight and index given by
\be
\label{indexconjVW}
w_S = -\frac12\, b_2(S),
\qquad
m_S(N) = -\frac16\, K_S^2 (N^3-N)-2N ,
\ee
where $K_S=-c_1(S)$ is the canonical class of $S$.
At this point, the holomorphic functions $h_{N,\mu}$ remain undetermined and represent the unknown part of the completion.
However, it is clear that the transformation properties of $\whh_{N,\mu}$ impose on them severe restrictions.
In this paper we show that they can actually be uniquely fixed up to a holomorphic modular function.
Furthermore, taking into account that $h_{N,\mu}$ must have a simple pole at $z=0$ allows to fix this ambiguity as well,
up to a finite number of parameters: the choice of a null vector $\nv\in\Lambda_S= H^2(S,\IZ)$, i.e. satisfying $\nv^2=0$,
and $b_2(S)-2$ integer parameters $\kappa_I$.
All these parameters are easily fixed by comparing, for example, $h_{2,\mu}$ with the known results in the literature.

As a result, we arrive at explicit representations for both $h_{N,\mu}$ and $\whh_{N,\mu}$
in the canonical chamber in terms of various theta series and certain modular forms.
In this paper we concentrate on the case of Hirzebruch and del Pezzo surfaces for which
the generating functions are found to be\footnote{In the main text
we mainly work in terms of the normalized functions $g_{N,\mu}=h_{N,\mu} \, h_{1,0}^{-N}$
and $\phi_{N,\mu}=\vph_{N,\mu} \, h_{1,0}^{-N}$.}
\be
h_{N,\mu}(\tau,z)=\sum_{n=1}^\infty\frac{1}{2^{n-1}}
\sum_{\sum_{i=1}^n \gama_i=\gama}
\Phi_n(\{\gama_i\})
\, \q^{\hf Q_n(\{\gama_i\})}
\, \y^{c_1(S)\cdot\sum\limits_{i=1}^n \ptt_i q_i }\prod_{i=1}^n \vph_{N_i,\mu_i}(\tau,z),
\label{mainres}
\ee
where $\gama=(N,\mu-\frac{N}{2}\, c_1(S))$, $\gama_i=(N_i,q_i)$ with $q_i$ decomposed as in \eqref{quant-q},
$Q_n$ is the quadratic form \eqref{defQlr} and $\ptt_i=\sum_{j<i}N_j-\sum_{j>i}N_j$. The kernel $\Phi_n$
is specified in Theorem \ref{th-main}, Eq.~\eqref{kerg}, and $\vph_{N,\mu}$ is a vector valued Jacobi form given by
\be
\vph_{N,\mu}=H^S_{N,\mu}(\tau,z;\nv)=\delta^{(N)}_{\nv\cdot\mu}\,
\frac{\I (-1)^{N-1} \eta(\tau)^{2N-3}}
{\theta_1(\tau,2Nz)\, \prod_{m=1}^{N-1} \theta_1(\tau,2mz)^2}
\prod_{\alpha=3}^{b_2}B_{N,\mu^\alpha}(\tau,z),
\ee
where $\delta^{(n)}_x$ is the Kronecker delta \eqref{defdelta},
$\nv$ is the relevant null vector, $H^S_{N,\mu}$ is the generating function of the so called stack invariants
evaluated in the chamber of the moduli space with $J=\nv$, $B_{N,\ell}$
are the ``blow-up functions" \eqref{defBNk} and $\mu^\alpha$, $\alpha\ge 3$,
are the components of the first Chern class of the sheaf along $m-1$ exceptional divisors of the del Pezzo surface $\Bb_m$.
The formula for the completion $\whh_{N,\mu}$ has exactly the same form, but with the kernel replaced by $\whPhi_n$ \eqref{kerhg}.

Furthermore, the existence of solutions with other parameters, in particular, generated by other null vectors $\nv'\ne \nv$,
leads to interesting consequences.
First, it turns out that certain pairs of null vectors give rise to the same generating functions.
This can be seen as a manifestation of the fiber-base duality \cite{Katz:1997eq,Mitev:2014jza}.
For the generating functions of refined VW invariants, we find this phenomenon for $\Fb_0$, $\Fb_2$ and $\Bb_m$ with $m\ge 4$.
The equality of the two sets of generating functions implies certain identities between
Jacobi theta functions, Dedekind function and generalized Appell functions introduced in \cite{Manschot:2014cca}.
Whereas for $N=2$ they reduce to the periodicity property of the classical Appell--Lerch function, for higher $N$
they appear to be new and very non-trivial.

Second, expanding the alternative solutions in Fourier series, one may extract new rational numbers.
It is an interesting question whether they can be interpreted as some BPS indices or topological invariants.
Here we do not try to answer it and restrict ourselves just to noticing their existence.

The outline of the paper is the following. In the next section we define the generating functions
of refined VW invariants and describe the expression for their modular completion found in \cite{Alexandrov:2019rth}.
Then in $\S$\ref{sec-indef} we find the holomorphic generating functions up to a modular ambiguity,
which is then fixed in $\S$\ref{sec-holom} by studying the behavior near $z=0$.
In $\S$\ref{sec-ident} we discuss the generating functions based on different null vectors,
reveal the fiber-base duality and derive its consequences for the generalized Appell functions.
Finally, $\S$\ref{sec-concl} presents our conclusions.
A few appendices review relevant information about indefinite theta series, Hirzebruch and del Pezzo surfaces,
and contain details of some calculations.

\section{Vafa-Witten invariants and modular completion}
\label{sec-compl}

\subsection{Generating functions of refined VW invariants}
\label{subsec-VWinv}

As was shown in the seminal paper \cite{Vafa:1994tf}, twisting $\cN=4$ super Yang-Mills gives rise to a topological
theory which can be defined on any smooth compact 4-dimensional manifold $S$. We assume that $S$
is a smooth almost Fano surface with $b_2^+(S)=1$ and $b_1(S)=0$ so that it can appear as the base of
a smooth elliptic fibration $\CY\to S$ with a single section and the total space
being a CY threefold where the divisor $S$ is rigid. These restrictions are needed
to borrow the results explained below, which have been obtained originally for compact CYs.
The non-compact CY, providing a bridge to VW theory, arises by taking the so-called local limit,
which plays the central role in geometric engineering of supersymmetric gauge theories \cite{Katz:1996fh,Katz:1997eq}.
This limit is obtained by zooming in on the region near a singularity in the moduli space, and is realized
mathematically by sending to infinity the K\"ahler modulus associated
to the elliptic fiber \cite{Alexandrov:2017mgi}.

After twisting, the path integral localizes on solutions of hermitian Yang-Mills equations\footnote{For general complex surfaces
this statement is not true and there are additional contributions from the so-called monopole branch \cite{Dijkgraaf:1997ce}.
However, Fano surfaces are K\"ahler manifolds with a positive anti-canonical class
which ensures the absence of the monopole contributions \cite{Vafa:1994tf}.}
for the field strength $F$, which means that $F^{(0,2)}=F^{(2,0)}=0$ and $\int_S F\wedge J$
is proportional to the identity matrix, where $J$ is
the K\"ahler form on $S$. For gauge group $U(N)$, the solutions span a moduli space $\cM_{N,\mu,n,J}$
classified by $\mu=-c_1(F)\in \Lambda_S\equiv H^2(S,\IZ)$ and $n=\int_S c_2(F)\in \IZ$.
In fact, the parameter $\mu$ can be restricted to $\Lambda_S/N\Lambda_S$
because the moduli space does not change upon tensoring $F$ with a line bundle $\cL$
which leads to $\mu\to \mu-Nc_1(\cL)$, but leaves invariant the Bogomolov discriminant
\be
\Delta(F):= \frac{1}{N} \left( n - \frac{N-1}{2N} \mu^2 \right),
\label{Bogom}
\ee
where $\mu^2\equiv \int_S\mu^2$.

According to the results of \cite{Vafa:1994tf}, the partition function of this theory
is expressed through Euler numbers of the moduli spaces $\cM_{N,\mu,n,J}$.
We however will be interested in the refined invariants defined by the Betti numbers of the moduli spaces
\be
\Omega_J(\gamma,y)=\frac{\sum_{p=0}^{2d_{\IC}(\cM_{\gamma,J})} y^{p-d_{\IC}(\cM_{\gamma,J})}\, b_p(\cM_{\gamma,J})}{y-y^{-1}}\, ,
\label{intOm}
\ee
where $y=e^{2\pi\I z}$ is the refinement parameter, $d_{\IC}(\cM)$ is the complex dimension of $\cM$,
and we introduced $\gamma=(N,\mu,-n+\frac12\, \mu^2)$.
As usual (see, e.g. \cite{Manschot:2010xp,Manschot:2017xcr}),
for discussion of modularity it is more convenient to work in terms of their rational counterparts given by
\be
\bOm_J(\gamma,y) =  \sum_{m|\gamma} \frac{1}{m}\, \Omega_J(\gamma/m, - (-y)^{m}).
\label{defcref}
\ee
Clearly, the dependence on $J$ is only piecewise constant and moreover was found to be absent when $b_2^+(S)>1$.
For $b_2^+(S)=1$ and $b_2(S)>1$, which is our case of interest, it is present, but is captured
by the standard wall-crossing formulas \cite{ks,Joyce:2008pc,Joyce:2009xv},
familiar in the context of supersymmetric gauge theories, $N=2$ supergravity and DT invariants.
We will be interested in one particular chamber with $J=-K_S$, which is called canonical and corresponds
to the attractor chamber for DT invariants where the results of \cite{Alexandrov:2018lgp,Alexandrov:2019rth} are applied.
Hence, we define
\be
\label{defhVWref}
h_{N,\mu}(\tau,z) =
\sum_{n\geq 0}
\bOm_{-K_S}(\gamma,y)\,
\q^{N\(\Delta(F) - \tfrac{\chi(S)}{24}\)},
\ee
where we used the standard notation $\q=e^{2\pi\I\tau}$.
Note that due to \eqref{intOm} this function has single poles at $z=0$ and $z=\hf$ with the residues
given by the generating function of the unrefined VW invariants.

For $N=1$, the generating function is known for any $S$ \cite{Gottsche:1990} and when
$b_1(S)=0$ is given by
\be
\label{h10anySref}
h_{1,0}(\tau,z) = \frac{\I}{\theta_1(\tau,2z)\, \eta(\tau)^{b_2(S)-1}},
\ee
where $\theta_1(\tau,z)$ is the Jacobi theta function and $\eta(\tau)$ is the Dedekind eta function,
whose definitions are recalled in Appendix \ref{ap-thJ}.
For $N>1$, the situation is more complicated, although many explicit expressions are already available in the literature.
For instance, up to $N=3$ they exist for $S=\IP^2$ \cite{Yoshioka:1994,Manschot:2010nc,Manschot:2011ym},
Hirzebruch surfaces \cite{Yoshioka:1995,Manschot:2011dj,Manschot:2011ym} and $S=\IB_9$ \cite{yoshioka1999euler,Klemm:2012sx},
even in generic chamber of the moduli space. (Ref. \cite{Manschot:2011ym} gives also $h_{4,0}$ in a specific chamber).
For other del Pezzo surfaces, $h_{2,\mu}$ can be found in \cite{Haghighat:2012bm}.
In principle, a general procedure is known \cite{Manschot:2014cca} which allows to compute
$h_{N,\mu}$ using the blow-up formula \cite{Yoshioka:1996,0961.14022,Li:1998nv} and wall-crossing.
But it is complicated by the fact that one should pass through the so-called stack invariants, which are polynomial
combinations of $\bOm_J(\gamma,y)$ having simpler transformation properties under wall-crossing.
Recently, in \cite{Beaujard:2020sgs} another general method to compute VW invariants
has been proposed based on the relation to quivers and the flow tree index introduced in \cite{Alexandrov:2018iao}.
That paper also provided many explicit expressions, including the VW invariants for higher ranks.
However, to the best of our knowledge, up to now there was no closed formula for the generating functions
of arbitrary rank $N$ for any relevant $S$. The goal of this paper is to fill this gap using the constraints imposed by modularity.

\subsection{Modular completion}

S-duality of $U(N)$ $\cN=4$ super-Yang-Mills suggests that the partition function of VW theory transforms as
a modular form so that one can expect that the generating functions $h_{N,\mu}$ behave as Jacobi forms under
\be
\tau\to \frac{a\tau+b}{c\tau+d}\, ,
\qquad
z\to \frac{z}{c\tau+d}\, ,
\qquad
\(\begin{array}{cc}
a & b \\ c & d
\end{array}\)\in SL(2,\IZ)
\label{transf-tz}
\ee
(see Appendix \ref{ap-Jacobi} for the definition of Jacobi form).
And indeed, for $N=1$ the function \eqref{h10anySref} is a Jacobi form of weight $-\frac12 b_2(S)$ and index $-2$.
However, this expectation turns out to be naive for $h_{N,\mu}$ with $N\ge 2$ and $b_2^+(S)=1$
in which case a modular anomaly has been found \cite{Vafa:1994tf}.
This does not imply the failure of S-duality yet. In fact, it is supposed to be more fundamental than holomorphicity, and
therefore one expects that the partition function is expressed through a non-holomorphic modular completion $\whh_{N,\mu}$
which does transform as a true Jacobi form.
As was shown explicitly in \cite{Dabholkar:2020fde},
the non-holomorphic contributions to the path integral are generated by Q-exact terms due to
boundaries of the moduli space, similarly to the holomorphic anomaly in
the topological string theory \cite{Bershadsky:1993ta}.
Thus, the determination of the completion is an important problem both
for the purpose of finding the physical partition function
and as a characterization of the modular anomaly of the original generating function.

Until recently, only very limited results existed in that respect, not going beyond $N=2$
\cite{Vafa:1994tf,Bringmann:2010sd,Manschot:2011dj} and $N=3$ for $\IP^2$ \cite{Manschot:2017xcr}.
The breakthrough came from the analysis of D-instantons in Calabi-Yau compactifications of type II string theory
(see \cite{Alexandrov:2011va} for a review).
S-duality of type IIB string theory implies that the hypermultiplet moduli space
of the compactified theory carries an isometric action of $SL(2,\IZ)$.
In particular, it must be consistent with instanton corrections coming from D3-branes wrapping a divisor in CY.
Since, on one hand, these instanton contributions are weighted by DT invariants $\Omega(\gamma)$,
and on the other hand, their description is known to all orders \cite{Alexandrov:2008gh,Alexandrov:2009zh},
this can be used to derive a restriction on $\Omega(\gamma)$, which is realized as a constraint
on the transformation property of their generating function $h_{\cD,\mu}$ evaluated at the attractor point.
The constraint turns out to depend on the properties of the divisor $\cD$ wrapped by D3-brane.
Whereas for an irreducible divisor $h_{\cD,\mu}$ must be modular \cite{Alexandrov:2012au},
for reducible $\cD$, i.e. decomposable as a sum of effective divisors $\cD=\sum\cD_i$,
it was shown to have a modular anomaly \cite{Alexandrov:2016tnf}.
However, using an expansion of $\Omega(\gamma)$ in terms of their values at the attractor point \cite{Alexandrov:2018iao}
derived from the attractor flow conjecture \cite{Denef:2001xn}, it was possible to find a non-holomorphic combination of
the holomorphic generating functions that must transform as a vector valued modular form
to ensure the isometric action of $SL(2,\IZ)$ on the hypermultiplet moduli space \cite{Alexandrov:2018lgp}.
This combination is nothing else but the modular completion $\whh_{\cD,\mu}$.

Furthermore, in \cite{Alexandrov:2019rth} this construction was generalized to the refined case with a non-trivial
refinement parameter $y$. Although this refined construction remains conjectural
since it does not have a rigorous justification from a well-established S-duality, it passed several
non-trivial consistency checks.  Finally, choosing the CY to be an elliptic fibration over $S$ and taking a local limit,
one arrives at the following expression for the modular completion of the generating function of
refined VW invariants\footnote{This expression is the specification of Eq. (3.12) in \cite{Alexandrov:2019rth}
to the case of interest, i.e. the local limit of elliptically fibered CY.
Most of relevant data was computed in $\S$4.3 of that paper.
Comparing to \cite{Alexandrov:2019rth}, we also extracted the power of $y$ from the coefficient and called the new one $\Rv_n$.}
\be
\whh_{N,\mu}(\tau,z)= \sum_{n=1}^\infty\frac{1}{2^{n-1}}
\sum_{\sum_{i=1}^n \gama_i=\gama}
\Rv_n(\{\gama_i\},\tau_2,\beta)
\, \q^{\hf Q_n(\{\gama_i\})}
\, \y^{\sum_{i<j} \gamma_{ij} }
\prod_{i=1}^n h_{N_i,\mu_i}(\tau,z).
\label{exp-whhr}
\ee

Let us explain various notations appearing in this formula.
First, we introduced charges $\gama=( N, q)$ where $q\in \Lambda_S+\frac{N}{2} K_S$.
To take into account the spectral flow invariance discussed above \eqref{Bogom},
we further decompose $q$ into the part spanning $N\Lambda_S$ and the residue class $\mu\in\Lambda_S/N\Lambda_S$.
In a basis $D_\alpha$, $\alpha=1,\dots, b_2(S)$, of $H_2(S,\IZ)$ this decomposition is given by\footnote{In \cite{Alexandrov:2019rth}
the last term was slightly different with $N$ replaced by $N^2$.
However, the two decompositions are related by a shift of $\mu$ in \eqref{quant-q}, as discussed around Eq. (4.35) of that paper.
Here we choose to work in the conventions accepted in VW theory to facilitate comparison to the literature.
For the same reason we changed also the sign of the refinement parameter in \eqref{exp-whhr}.
}
\be
\label{quant-q}
q_\alpha =  N \, C_{\alpha\beta} \epsilon^{\beta} + \mu_\alpha- \frac{N}{2}\, C_{\alpha\beta} c_1^\beta,
\qquad
\eps^\alpha\in \IZ,
\ee
where $C_{\alpha\beta}=D_\alpha\cap D_\beta$ is the intersection matrix on $S$
and $c_1^\alpha$ are the components of the first Chern class.
Then the sum in \eqref{exp-whhr} goes over all decompositions of the charge $\gama$ such that
\be
\sum_{i=1}^n N_i=N,
\qquad
\sum_{i=1}^n q_{i,\alpha}=\mu_\alpha- \frac{N}{2}\, C_{\alpha\beta} c_1^\beta,
\label{sumgam}
\ee
where $q_{i,\alpha}$ are quantized as in \eqref{quant-q} with $N$ replaced by $N_i$.

Next, we defined the combination
\be
\label{gam12}
\gamma_{ij}=
c_1^{\alpha} (N_i q_{j,\alpha} -N_j q_{i,\alpha})
\ee
and the quadratic form $Q_n$ given by
\be
Q_n(\{\gama_i\})= \frac{1}{N}\, q^2-\sum_{i=1}^n \frac{1}{N_i} q_{i}^2
=-\sum_{i<j}\frac{(N_i q_j - N_j q_i)^2}{NN_iN_j}\, ,
\label{defQlr}
\ee
where $q^2=C^{\alpha\beta}q_\alpha q_\beta$ and $C^{\alpha\beta}$ is the inverse of $C_{\alpha\beta}$.
Note that they are both invariant under an overall shift of $\eps_i^\alpha$.
The same is true for the coefficients $\Rv_n$ so that the r.h.s. of \eqref{exp-whhr}
is invariant under shifts of $\eps^\alpha$,
which explains why it is possible to put it zero in \eqref{sumgam}.

\lfig{An example of Schr\"oder tree contributing to $\Rv_8$. Near each vertex we showed the corresponding factor
using the shorthand notation $\gamma_{i+j}=\gamma_i+\gamma_j$.}
{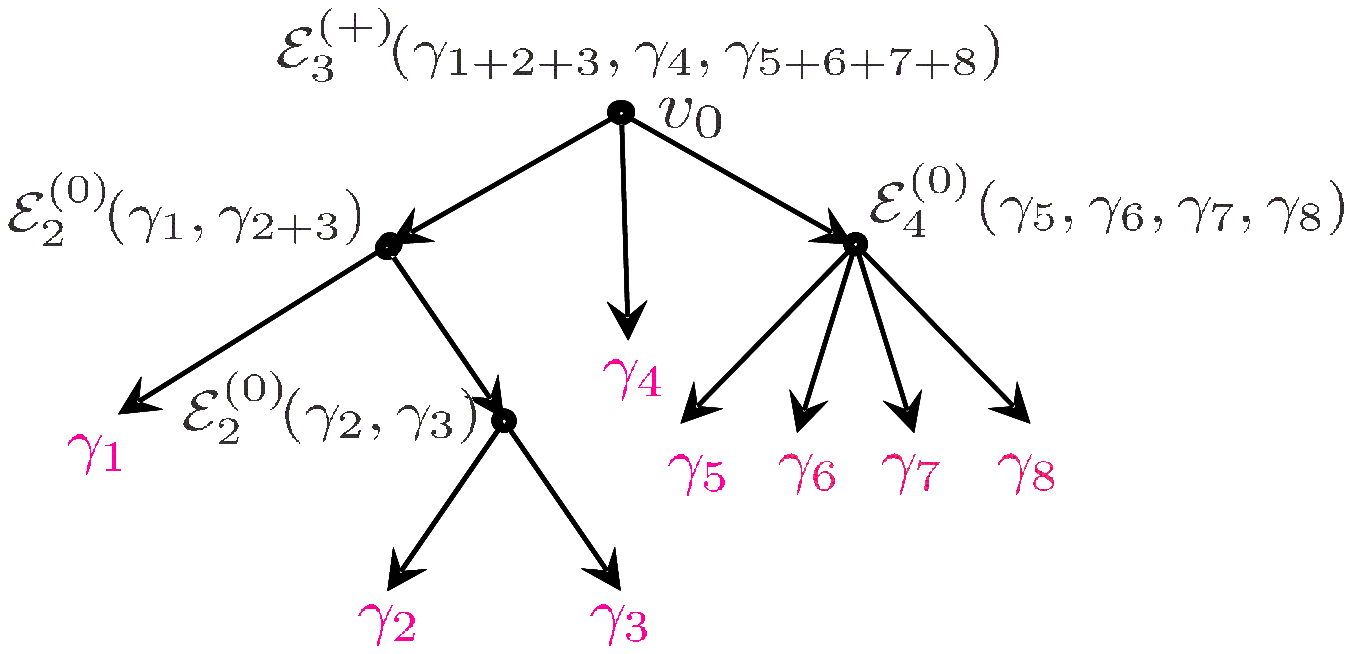}{9.75cm}{fig-Rtree}{-1.2cm}

The non-holomorphicity of the completion is due to the coefficients $\Rv_n$. They depend
on the imaginary parts of both $\tau$ and $z$, defined as $\tau=\tau_1+\I\tau_2$ and $z=\alpha-\tau\beta$,
and are given by
\be
\Rv_n(\{\gama_i\},\tau_2,\beta)= \Sym\left\{\sum_{T\in\IT_n^{\rm S}}(-1)^{n_T-1}
\sEp_{v_0}\prod_{v\in V_T\setminus{\{v_0\}}}\sEf_{v}\right\},
\label{solRnr}
\ee
where $\Sym$ denotes symmetrization (with weight $1/n!$) with respect to the charges $\gama_i$.
Here the sum goes over so-called Schr\"oder trees with $n$ leaves (see Figure \ref{fig-Rtree}), i.e. rooted planar
trees such that all vertices $v\in V_T$ (the set of vertices of $T$ excluding the leaves) have $k_v\geq 2$ children,
$n_T$ is the number of elements in $V_T$, and $v_0$ labels the root vertex.
The vertices of $T$ are labelled by charges so that the leaves carry charges $\gama_i$, whereas the charges assigned to other vertices
are given recursively by
the sum of charges of their children, $\gama_v\in\sum_{v'\in\Ch(v)}\gama_{v'}$.
Finally, to define the functions $\sEf_v$ and $\sEp_v$, let us consider
a set of functions $\sE_n$ depending on $n$ charges, $\tau_2$ and $\beta$, whose explicit expressions will be given shortly.
Given this set, we take
\be
\begin{split}
\sEf_n(\{\gama_i\})=&\,\under{\lim}{\tau_2\to\infty}\sE_n\(\{\gama_i\},\tau_2, -\frac{\Im z}{\tau_2}\),
\\
\sEp_n(\{\gama_i\},\tau_2,\beta)=&\,\sE_n(\{\gama_i\},\tau_2,\beta)-\sEf_n(\{\gama_i\}),
\end{split}
\label{redef-cErf}
\ee
so that $\sEf_n$ does not depend on $\tau_2$ (and $\beta$),
whereas the second term $\sEp_n$ turns out to be exponentially suppressed as $\tau_2\to\infty$ keeping
the charges $\gama_i$ fixed. Then, given a Schr\"oder tree $T$,
we set $\sE_{v}\equiv \sE_{k_v}(\{\gama_{v'}\})$ (and similarly for $\sEf_{v}, \sEp_{v}$)
where $v'\in \Ch(v)$ runs over the $k_v$ children of the vertex $v$.

It remains to provide the functions $\sE_n$. They are given by
\be
\sE_n(\{\gama_i\},\tau_2,\beta)= \Phi^E_{n-1}(\{ \bfv_{\ell}\};\bfx),
\label{Erefsim}
\ee
where $\Phi^E_{n-1}$ are (boosted) generalized error functions described in Appendix \ref{ap-generr},
which depend on $nb_2$-dimensional vectors with the following components
\be
\begin{split}
\bfv_{\ell,i}^\alpha =&\, \(\sN_\ell\delta_{i>\ell}-(N-\sN_\ell)\delta_{i\le \ell}\)c_1^\alpha,
\qquad \ell=1,\dots, n-1,
\\
\bfx_i^\alpha=&\, \sqrt{2\tau_2}\(\tfrac{1}{N_i}\,C^{\alpha\beta}q_{i,\beta}+\beta \ptt_i c_1^\alpha\),
\qquad i=1,\dots, n,
\end{split}
\label{vectors}
\ee
where $\sN_\ell=\sum_{k=1}^\ell N_k$ and $\ptt_i=\sum_{j<i}N_j-\sum_{j>i}N_j$.
In particular, with respect to the bilinear form\footnote{We use different multiplication symbols to distinguish between
bilinear forms on different spaces: $\cdot$ denotes contraction of $b_2$-dimensional vectors using $C_{\alpha\beta}$
(or its inverse), whereas $\bpt$ is used for $nb_2$-dimensional vectors.}
\be
\bfx\bpt\bfy=\sum_{i=1}^n N_i \, x_i\cdot y_i,
\label{biform}
\ee
these vectors satisfy
\be
\bfv_\ell\,\bpt \bfx =\sqrt{2\tau_2}\(\Gamma_\ell+N\sN_\ell (N-\sN_\ell) c_1^2\,\beta\),
\qquad
\Gamma_{\ell}=\sum_{i=1}^\ell\sum_{j=\ell+1}^n \gamma_{ij}.
\label{defGam}
\ee
Evaluating the limit in \eqref{redef-cErf}, one obtains \cite{Alexandrov:2019rth}
\be
\sEf_n(\{\gama_i\})=
e_{|\cI|}\prod_{\ell\in \Zv_{n-1}\setminus \cI}\sgn(\Gamma_{\ell}),
\quad \mbox{where }\cI\subseteq \Zv_{n-1}:\
\left\{ \begin{array}{ll}
\Gamma_\ell=0 & \mbox{ for } \ell\in \cI,
\\
\Gamma_\ell\ne 0 & \mbox{ for } \ell\notin \cI.
\end{array}\right.
\label{rel-gEf-zero}
\ee
Here $\Zv_{n}=\{1,\dots,n\}$, $|\cI|$ is the cardinality of the set, and
$e_{m-1}$ is the $m$-th Taylor coefficient of ${\rm arctanh}$, namely
\be
e_m=\left\{ \begin{array}{cc}
0 & \mbox{ if }m \mbox{ is odd},
\\
\frac{1}{m+1}& \mbox{ if }m \mbox{ is even}.
\end{array}\right.
\label{valek}
\ee
Thus, the configurations of charges leading to the vanishing arguments of sign functions should be treated separately.

The analysis of \cite{Alexandrov:2019rth} suggests that the completion defined in \eqref{exp-whhr}
transforms as a vector-valued Jacobi form (see \eqref{Jacobi}) of
weight $(-\frac12 b_2(S),0)$, index $m_S(N)$ given in \eqref{indexconjVW}, and the following multiplier system
\be
M^{(N)}_{\mu\nu}(T)=e^{-\pi\I \tfrac{N-1}{N}\, \mu^2 -\tfrac{\pi\I}{12}N \chi(S)}\,\delta_{\mu\nu},
\qquad
M^{(N)}_{\mu\nu}(S)=
\frac{\I(-1)^{N-1}}{(\I N)^{b_2(S)/2}}\,
e^{- \frac{2\pi\I}{N}\mu\cdot \nu}.
\ee

\section{Generating functions and indefinite theta series}
\label{sec-indef}

In this section we show how modularity fixes the generating functions $h_{N,\mu}$ up to holomorphic {\it modular} functions.
Our analysis is restricted to the Hirzebruch surfaces $\Fb_m$ with $0\leq m\leq 2$
and the del Pezzo surfaces $\Bb_m$ with $1\leq m \leq 8$, for which the relevant geometric data are reviewed
in Appendices \ref{sec-Fm} and \ref{sec-Bm}. In particular, we do not consider $\IB_0=\IP^2$ for which $b_2(S)=1$
since, as we will see shortly, it requires special attention.
In all cases of interest the lattice $\Lambda_S=H^2(S,\IZ)$ is unimodular
of signature $(1,b_2(S)-1)$, i.e. the corresponding quadratic form always has only one positive eigenvalue.

The form of the modular completion \eqref{exp-whhr} suggests that it is convenient to define
\be
g_{N,\mu}=h_{N,\mu} \, h_{1,0}^{-N},
\qquad
\whg_{N,\mu}=\whh_{N,\mu} \, h_{1,0}^{-N},
\label{defgmu}
\ee
where $h_{1,0}$ is given explicitly in \eqref{h10anySref}.
It follows that $\whg_{N,\mu}$ transforms as a vector valued Jacobi form of weight $\frac12(N-1)b_2(S)$ and
index $ -\frac16\, (N^3-N)K_S^2$.

\subsection{$N=2$}

We start with the simplest case $N=2$ to demonstrate the mechanism fixing the generating functions $g_{N,\mu}$ in detail.
We will be able to fix $g_{2,\mu}$ up to a holomorphic vector valued Jacobi form.
In the next subsection, we extend this derivation to arbitrary rank, whereas the holomorphic modular ambiguity will be fixed in
$\S$\ref{sec-holom}.

For $N=2$, the modular completion \eqref{exp-whhr} reads
\be
\begin{split}
\whg_{2,\mu}=&\, g_{2,\mu}+\hf \sum_{q_1+q_2=\mu+K_S}
\(E_1\(\sqrt{\frac{\tau_2}{c_1^2}}\(\gamma_{12}+2 c_1^2\beta\)\)-\sgn(\gamma_{12})\)
\q^{-\frac14 (q_2-q_1)^2}\, \y^{\gamma_{12} }
\\
=&\, g_{2,\mu}+\hf \sum_{k\in \Lambda_S+\hf\mu}
\(E_1\(2\sqrt{\frac{\tau_2}{c_1^2}}\(c_1\cdot k+c_1^2\beta\)\)-\sgn(c_1\cdot k)\)
\q^{- k^2}\, \y^{ 2c_1\cdot k},
\end{split}
\label{compl2}
\ee
where $E_1(u) = \Erf(\sqrt{\pi}u)$ and we took into account that $q_i\in \Lambda_S+\hf\, K_S$.
The second term is a theta series which belongs to the class
of theta series described in Appendix \ref{ap-theta}.
Comparing with the definition \eqref{Vignerasth}, we read off its data
\be
\begin{split}
&
\Lat=\Lambda_S \mbox{ with bilinear form } 2C_{\alpha\beta},
\\
\bfmu^\alpha&=\hf\, C^{\alpha\beta}\mu_\beta, \qquad \bfp^\alpha=0, \qquad \bfptt^\alpha=c_1^\alpha,
\end{split}
\label{datalat2}
\ee
and the kernel
\be
\Phi(x)=\hf\( \Phi_1^E(c_1;x)-\sgn(c_1\cdot x)\).
\ee
This theta series is not modular due to the sign function in the kernel.
However, it can be made modular by adjusting the holomorphic term $g_{2,\mu}$.
Indeed, this term can cancel the troubling sign function because it is holomorphic.
But this is not enough since a theta series with the kernel given by a single sign or a single error function
is divergent. As follows from Theorem \ref{th-converg} in Appendix \ref{ap-converge},
to make it convergent, another sign function
should be added to the kernel, say $\sgn(\nv\cdot x)$. It ensures the convergence provided
\be
\nv\cdot c_1>0,
\label{c0c1}
\ee
and does not spoil modularity only if the vector $\nv$ is null and belongs to the lattice.
Thus, we arrive at the following ansatz
\be
g_{2,\mu}=\phi_{2,\mu}+\hf\sum_{k\in \Lambda_S+\hf\mu}
\Bigl(\sgn(c_1\cdot k)-\sgn(\nv\cdot (k+\beta c_1))\Bigr)\,
\q^{- k^2}\, \y^{ 2c_1\cdot k},
\label{resg2}
\ee
where $\phi_{2,\mu}(\tau,z)$ is a holomorphic vector valued Jacobi form of weight $\hf b_2(S)$ and index $-K_S^2$,
which remains undetermined at this point.
Substituting this ansatz into \eqref{compl2}, we do get a function with correct transformation properties
\be
\whg_{2,\mu}=\phi_{2,\mu}+\hf\,\vartheta_{\mu/2}(\tau,z;0,c_1,\whPhi_2),
\ee
where the second term is the theta series \eqref{Vignerasth} defined by the lattice \eqref{datalat2}
and the kernel given by
\be
\whPhi_2(x)=\ \Phi_1^E(c_1;x)-\sgn(\nv\cdot x).
\ee

Note that this solution does not work for $S=\IP^2$ because in this case $b_2=1$ and there are no null vectors in one dimension.
This is why this case should be treated separately and we leave it for a future work.

Instead, for Hirzebruch and del Pezzo surfaces a null lattice vector always exists.
Moreover, it is not unique. Therefore, the natural question is whether its choice affects the construction?
We will return to this question below in $\S$\ref{sec-ident}, and for the moment continue with a generic choice of $\nv$.

\subsection{Arbitrary rank}
\label{subsec-anyN}

The construction done for $N=2$ in the previous subsection can be repeated for any $N$.
There are however two features which lead to more complicated final expressions.
First, at each rank one has to introduce a new holomorphic function $\phi_{N,\mu}$,
transforming as a vector valued Jacobi form, which then propagates to all higher ranks.
Second, starting from $N=3$, one must take into account the contributions to functions $\sEf_n$
from charge configurations giving rise to $\Gamma_\ell=0$ and appearing in \eqref{rel-gEf-zero} with weights $e_m$.
For instance, accepting the convention $\sgn(0)=0$, one has
\be
\sEf_3=\sgn(\Gamma_1)\,\sgn(\Gamma_2) +\frac13\, \delta_{\Gamma_1=\Gamma_2=0}.
\ee

To formulate the final result of the analysis, it is convenient to set $\phi_{1,0}=1$ and introduce
a class of theta functions
\be
\vartheta^{(\vec N)}_{\mu,\vec\mu}(\tau,z;\Phi)
=\!\!\!\sum_{\sum_{i=1}^n q_i=\mu-\frac{N}{2}\, c_1}\!\!\!
\Phi(\{\gama_i\})
\, \q^{\hf Q_n(\{\gama_i\})}
\, \y^{c_1\cdot\sum\limits_{i=1}^n \ptt_i q_i },
\label{newtheta}
\ee
where vectors denote collections of $n$ components, $N=\sum_{i=1}^n N_i$,
and in contrast to \eqref{exp-whhr}, the sum is performed keeping the residue classes
$\mu_i\in\Lambda_S/N_i\Lambda_S$ fixed.
These theta series will always appear in the following combinations
\be
\Theta_{N,\mu}(\tau,z;\{\Phi_n\})=\sum_{n=1}^\infty\frac{1}{2^{n-1}}
\sum_{\sum_{i=1}^n N_i=N}
\sum_{\vec\mu}\vartheta^{(\vec N)}_{\mu,\vec\mu}(\tau,z;\Phi_n)
\prod_{i=1}^n\phi_{N_i,\mu_i}(\tau,z).
\label{combTh}
\ee
Then, requiring the proper modular transformations of the completion \eqref{exp-whhr},
one arrives at the following

\begin{theorem}
\label{th-main}
The normalized generating functions and their modular completions are expressed through the combinations  \eqref{combTh}
\be
g_{N,\mu}=\Theta_{N,\mu}(\tau,z;\{\Phi_n\}),
\qquad
\whg_{N,\mu}(\tau,z)=\Theta_{N,\mu}(\tau,z;\{\whPhi_n\}),
\label{whgN}
\ee
where the kernels are given by
\bea
\Phi_n(\{\gama_i\})&=&\(\sum_{\cJ\subseteq \cI} e_{|\cJ|}\prod_{k\in \cI\setminus\cJ}\Bigl(-\sgn\(\nv\cdot b_k\)\Bigr)\)
\prod_{k\in \Zv_{n-1}\setminus \cI}\Bigl(\sgn(\Gamma_k)-\sgn\(\nv\cdot b_k\)\Bigr),
\label{kerg}
\\
\whPhi_n(\{\gama_i\})&=&\sum_{\cJ\subseteq \Zv_{n-1}} \Phi_{|\cJ|}^E(\{ \bfv_{\ell}\}_{\ell\in \cJ};\bfx)
\prod_{k\in \Zv_{n-1}\setminus \cJ}\Bigl(-\sgn\(\nv\cdot b_k\)\Bigr).
\label{kerhg}
\eea
Here $\Gamma_k$ was defined in \eqref{defGam}, $e_m$ in \eqref{valek},
$\cI\subseteq \Zv_{n-1}$ is the subset of indices for which $\Gamma_k=0$ as in \eqref{rel-gEf-zero},
$\Phi_n^E$ are the generalized error functions \eqref{generrPhiME},
the vectors $\bfv_{\ell}$ and $\bfx$ are from \eqref{vectors},
and
\be
b_k=N_k q_{k+1}-N_{k+1} q_k+\beta N_k N_{k+1}(N_k+N_{k+1})c_1.
\label{defbk}
\ee
\end{theorem}

The proof of this theorem is a bit long and technical, and we relegate it to Appendix \ref{ap_theorem}.
Although the expressions for $g_{N,\mu}$ and their completions might seem to be complicated,
the complications are mainly due to the two features mentioned above:
the sum over partitions of $N$ is needed to account for the contributions of the holomorphic modular ambiguities $\phi_{N_i,\mu_i}$
appearing at each rank, and the first factor in \eqref{kerg} takes into account the charge configurations
giving rise to vanishing arguments of sign functions. If none of them is vanishing, this factor is absent
and the kernel $\Phi_n$ reduces to the standard kernel for indefinite theta series (c.f. \eqref{kerconverge})
\be
\Phi_n(\{\gama_i\})=
\prod_{k=1}^{n-1}\Bigl(\sgn(\Gamma_k)-\sgn\(\nv\cdot b_k\)\Bigr),
\label{kerg-no0}
\ee
where $\nv\cdot b_k$ can be equally written as (c.f. \eqref{defGam} and \eqref{vectors})
\be
\sqrt{2\tau_2}\, \nv\cdot b_k=\bfw_{k,k+1}\bpt\bfx,
\qquad
\bfw_{k\ell,i}^\alpha= \(N_k\delta_{i\ell}-N_\ell\delta_{ik}\) \nv^\alpha,
\label{nullvec}
\ee
and the vectors on the r.h.s. are contracted using the bilinear form \eqref{biform}.
Furthermore, the kernel $\whPhi_n$ \eqref{kerhg} defining the completion is just the one obtained from \eqref{kerg-no0}
by applying the recipe to construct modular completions of indefinite theta series, explained in Appendix \ref{ap-generr}:
expand the product and replace each monomial by the (boosted) generalized error function with parameters determined by
arguments of the sign functions entering the monomial.
However, since the vectors $\bfw_{k\ell}$ are null, the rank of some generalized error functions can be
reduced by the property \eqref{Phinull}, which finally gives \eqref{kerhg}.

\section{Holomorphic modular ambiguity}
\label{sec-holom}

The results of the previous section reduce the unknown part of the generating functions $g_{N,\mu}$
to the {\it holomorphic modular} functions $\phi_{N,\mu}$. In this section we show how this holomorphic modular ambiguity
can be fixed by requiring the proper behavior in the unrefined limit $\y\to 1$.
Unfortunately, we do not have proofs for all our statements and some of them are left as conjectures.

\subsection{The unrefined limit}
\label{subsec-smooth}

Let us recall that the generating functions $h_{N,\mu}$ \eqref{defhVWref} by construction have single poles at $z=0$ and $z=\hf$.
This implies that the normalized functions $g_{N,\mu}$ \eqref{defgmu} have zeros of order $n-1$ at these points.
This condition must be imposed on the theta series representation \eqref{whgN} derived in the previous section
and can be viewed as a restriction on the functions $\phi_{N,\mu}$. As we will se now, it turns out to be so restrictive
that fixes these functions almost uniquely.

But first we should make this condition more explicit.
To this end, we reveal the behavior of the theta functions \eqref{newtheta} in the unrefined limit $\y\to 1$.
These theta functions are very close to the ones defining the so-called {\it tree index} and extensively studied in \cite{Alexandrov:2018iao},
where it was shown that certain sign identities may ensure the required vanishing property upon proper choice of the kernel.
Inspired by these findings, we suggest the following

\begin{conj}
\label{conj-smooth}
Let $\nb_{k\ell}=\nv\cdot(N_k q_{\ell}-N_{\ell} q_k)$. Provided all $\nb_{k\ell}$
are non-vanishing, the function
\be
F_n(\{\gama_i\},y)=\Sym\left\{\y^{\sum_{i<j} \gamma_{ij} }\, \Phi_n(\{\gama_i\})\right\},
\ee
where $\Phi_n$ is defined in \eqref{kerg}, has zero of order $n-1$ at $y=1$.
\end{conj}

We have checked this conjecture on Mathematica up to order $n=7$ which provides already a good level of confidence.
The condition that all $\nb_{k\ell}$ are non-vanishing is essential since it is easy to find examples where it is broken
and one does not get the zero of the correct order.
Some of them will be considered below.

This conjecture ensures that the theta function $\vartheta^{(\vec N)}_{\mu,\vec\mu}(\tau,z;\Phi_n)$
can be split into two parts
\be
\vartheta^{(\vec N)}_{\mu,\vec\mu}(\tau,z;\Phi_n)=\cvtheta^{(\vec N)}_{\mu,\vec\mu}(\tau,z)+\tvtheta^{(\vec N)}_{\mu,\vec\mu}(\tau,z),
\label{septheta}
\ee
where the first term takes into account only charge configurations spoiling the condition of Conjecture \ref{conj-smooth},
which it is natural to call ``zero modes", whereas the second term is given by the sum over all other charges.
It is clear that $\tvtheta^{(\vec N)}_{\mu,\vec\mu}$ has zero of order $n-1$ at $z=0$ and
only $\cvtheta^{(\vec N)}_{\mu,\vec\mu}$, the contribution of zero modes, requires a special attention,
It can be written as
\be
\cvtheta^{(\vec N)}_{\mu,\vec\mu}(\tau,z)=
\!\!\!\sum_{\sum_{i=1}^n q_i=\mu-\frac{N}{2}\, c_1 \atop \exists k,\ell \ : \ \nb_{k\ell}=0}\!\!\!
F_n(\{\gama_i\},y)
\, \q^{\hf Q_n(\{\gama_i\})},
\label{defctheta}
\ee
i.e. as the sum of the terms in the theta function that give rise to $\nb_{k\ell}=0$ for some $k$ and $\ell$.
The number of vanishing $\nb_{k\ell}$ will be called ``the order of the zero mode".
Our goal is to understand the behavior of \eqref{defctheta} in the unrefined limit.
But first, one should adapt our lattice for this purpose.

\subsection{Lattice factorization}
\label{subsec-lat}

Let us split the lattice of charges, which one sums over in \eqref{newtheta}, into the two parts
contributing respectively to the two terms in \eqref{septheta}.
For this purpose, it is convenient to start from the lattice $\Lambda_S$ and factorize it into two sublattices:
the first is a two-dimensional lattice $\Latc$ spanned by the integer linear combinations
of\footnote{We assume that $\nv$ is chosen
to be primitive, i.e. $\gcd(\nv^\alpha)=1$. The first Chern class however $c_1(S)$ is not always primitive,
as is the case for $\Fb_0$ and $\Fb_2$. Therefore, to obtain the second basis vector $\cv$, we divide $c_1$ by the $\gcd$
of its components. \label{foot-c1}}
$\nv$ and $\cv=c_1/\gcd(c_1^\alpha)$
and the second is the orthogonal complement to $\Latc$, which will be denoted by $\Lambda_S^\perp$.
The idea behind this factorization is that the kernel $\Phi_n$ is independent of the components along $\Lambda_S^\perp$
so that we reduce our problem to a two-dimensional one. (Of course, if $\Lambda_S$ is two-dimensional as for
Hirzebruch surfaces, $\Lambda_S^\perp$ does not arise.)
Furthermore, working in the basis $\{\nv,\cv\}$ allows to identify the part of the lattice contributing to the zero modes.
Indeed, since $\nv$ is null, the condition $\nb_{k\ell}=0$ can be rewritten as
\be
N_k q_{\ell}^{\cv}-N_{\ell} q_k^{\cv}=0,
\ee
where $q^{\cv}$ denotes the $\cv$-component of the vector $q$.
As a result, the lattice defining $\cvtheta^{(\vec N)}_{\mu,\vec\mu}$ is simply obtained by fixing this component
of the charge vectors.

The problem however is that generically $\Lambda_S$ does {\it not} coincide with $\Latc\oplus\Lambda_S^\perp$.
Whereas, the original lattice is unimodular, the determinant of $\Latc$
is equal to $r^2=(\nv\cdot \cv)^2$, which follows from the Gram matrix of the basis vectors
\be
\(\begin{array}{cc}
0  & \nv\cdot \cv
\\
\nv\cdot \cv & \cv^2
\end{array}\).
\ee
This implies that for $r_S\equiv r\sqrt{\det\Lambda_S^\perp}>1$
not all elements of $\Lambda_S$ can be obtained by linear combinations with integer coefficients of
the elements of $\Latc$ and $\Lambda_S^\perp$. In such case one has to introduce
the so called glue vectors \cite{CSbook}. They can be viewed as elements of the dual lattice and are given by the sum
of representatives of the two cosets,
$\vg_a=\vv_a+\vd_a$ where $\vv_a\in (\Latc)^*/\Latc$ and $\vd_a\in (\Lambda_S^\perp)^*/\Lambda_S^\perp$
with $a=0,\dots, r_S-1$, including the trivial vector $\vg_0=0$.
In terms of these vectors, the original lattice is given by
\be
\Lambda_S=\bcup\limits_{a=0}^{r_S-1}\Bigl[(\Latc+\vv_a)\oplus(\Lambda_S^\perp+\vd_a)\Bigr].
\ee
As a result, the charge vectors appearing in the definitions of the theta functions \eqref{newtheta} and \eqref{defctheta}
can be represented as a sum of two orthogonal components
\be
\label{factor-q}
\begin{split}
q_i =&\, \qv_i+q_i^\perp,
\\
\qv_i=&\, N_i \(m^0 \nv + \(m^1-\tfrac{\gcd(c_1^\alpha)}{2}\) \cv\) + \vrh_i+\vv_{a_i},
\qquad
m^0,
m^1\in \IZ,
\quad
\vrh_i\in \Latc/N_i\Latc,
\\
q_i^\perp=&\, N_i\lambda_i+\rho_i+\vd_{a_i},
\qquad\qquad\qquad\qquad\qquad\qquad\qquad
\lambda_i\in \Lambda_S^\perp,
\quad
\rho_i\in \Lambda_S^\perp/N_i\Lambda_S^\perp.
\end{split}
\ee
Due to the orthogonality properties, as already mentioned above, the kernel $\Phi_n$ is independent of $q_i^\perp$.
Taking into account that the quadratic form also factorizes, all theta functions can be written in a factorized form.
For instance,
\be
\vartheta^{(\vec N)}_{\mu,\vec\mu}(\tau,z)=\sum_{a_i=0}^{r_S-1}\vartheta^{\, v(\vec N)}_{\vrh,\vec\vrh+\vec \vv_a}(\tau,z)
\, \vartheta^{\, \perp(\vec N)}_{\rho,\vec\rho+\vec\vd_a}(\tau),
\label{factheta}
\ee
where $\vrh$ and $\rho$ denote projections of $\mu\in\Lambda_S/N\Lambda_S$ onto $\Latc\otimes \IR$
and $\Lambda_S^\perp\otimes \IR$, respectively, and we defined
\be
\vartheta^{\, v(\vec N)}_{\vrh,\vec\vrh+\vec \vv_a}(\tau,z)=
\!\!\!\sum_{\sum_{i=1}^n \qv_i=\vrh-\frac{N}{2}\, c_1}\!\!\!
F_n(\{\gama_i\},y)
\, \q^{\hf Q_n(\{\qv_i\})},
\qquad
\vartheta^{\, \perp(\vec N)}_{\rho,\vec\rho+\vec\vd_a}(\tau)=
\!\!\!\sum_{\sum_{i=1}^n q_i^\perp=\rho}\!\!\!
\, \q^{\hf Q_n(\{q_i^\perp\})}.
\label{twothetas}
\ee
Note that the second theta series is independent of the refinement parameter,
and the summation condition in the first one can be written more explicitly as
\be
\sum_{i=1}^n N_i m^\eps_i=\hat\vrh^\eps, \qquad \eps=0,1,
\quad m^\eps_i\in\IZ,
\ee
where $\hat\vrh= \hat\vrh^0\nv+\hat\vrh^1 \cv\equiv \vrh-\sum_i (\vrh_i+\vv_{a_i})$.
A similar factorized formula can be written for $\cvtheta^{(\vec N)}_{\mu,\vec\mu}(\tau,z)$.
It has the {\it same} second factor, whereas the sum in the first is further restricted by the condition
that for some $i,j$ one has
\be
m_i^1-m_j^1=\frac{\Bigl(N_i (\vrh_j+\vv_{a_j})-N_j (\vrh_i+\vv_{a_i})\Bigr)\cdot \nv}{r\,N_i N_j}\, .
\label{condzero}
\ee

\subsection{Determination of the ambiguity}
\label{subsec-amb}

After the preliminary work done in the previous subsections, we are ready to approach the problem
of finding the holomorphic modular ambiguities $\phi_{N,\mu}$.
We do this explicitly for ranks 2 and 3. The results obtained in these cases will be suggestive enough to guess the general
answer, which we leave as a conjecture.

\subsubsection{$N=2$}

In this case, we can start with the explicit result for the normalized generating function $g_{2,\mu}$ given in \eqref{resg2}.
It is clear that in the limit $z\to 0$, each term in the sum with $\nv\cdot k\ne 0$ cancels the corresponding term with $-k$
so that this part of the theta series (denoted by $\tvtheta^{(\vec N)}_{\mu,\vec\mu}$ in \eqref{septheta}) vanishes,
in agreement with Conjecture \ref{conj-smooth}.
Moreover, since $2c_1\cdot k\in\IZ$, this is true for $z=1/2$ as well.

After factorization of the lattice, the remaining terms become
\be
\cvtheta^{(1,1)}_{\mu}=\sum_{a=0}^{r_S-1}  \delta^{(2)}_{\vrh^1+2\vv_a^1}
\!\!\!\!
\sum_{m^0\in \IZ+\hf\vrh^0 +\vv_a^0}\Bigl[\sgn(m^0) -\sgn(\beta)\Bigr]
\,  y^{2\cdv m^0 }
\!\!\sum_{k^\perp\in \Lambda_S^\perp+\hf\rho+\vd_a}\!\!
\, \q^{-(k^\perp)^2},
\label{cvthetaN2}
\ee
where $\cdv=r\gcd(c_1^\alpha)=\nv\cdot c_1$ and
\be
\delta^{(n)}_x=\left\{ \begin{array}{l} 1 \quad \mbox{if } x=0 \mbox{ mod }n, \\ 0 \quad \mbox{otherwise.} \end{array} \right.
\label{defdelta}
\ee
The last factor is nothing else but the theta function
$\vartheta^{\, \perp(\vec N)}_{\rho,\vec\rho+\vec\vd_a}$ introduced in \eqref{twothetas} and specified to the case $\vec N=(1,1)$,
whereas the first factor, together with the delta symbol, corresponds to $\cvtheta^{\, v(\vec N)}_{\vrh,\vec\vrh+\vec \vv_a}$.
The sum over $m_0$ gives rise to a geometric progression which leads to
\be
\cvtheta^{\, v(1,1)}_{\vrh,\vv_a}=\delta^{(2)}_{\vrh^1+2\vv_a^1}\(
\frac{2\,y^{2\cdv\{\hf\vrh^0+\vv_a^0\}}}{1-y^{2\cdv}}-\delta^{(2)}_{\vrh^0+2\vv_a^0}\) ,
\ee
where the brackets $\{\,\cdot\, \}$ denote the fractional part.
This result shows that, instead of vanishing, the function \eqref{cvthetaN2} has poles at $z=0$ and $z=\hf$.
Hence, the two leading terms in the Laurent expansion around these poles must be cancelled by the holomorphic modular ambiguity $\phi_{2,\mu}$.
Thus, we arrive at the additional condition that, besides being a Jacobi form of weight $\hf b_2$
and index $-c_1^2$, this function must also have poles at $z=0$ and $z=\hf$ with the leading behavior
(near $z=0$) given by
\be
\phi_{2,\mu}(\tau,z)\sim \sum_{a=0}^{r_S-1}\delta^{(2)}_{\vrh^1+2\vv_a^1}\,
\(\frac{1}{4\pi\I \cdv z}+\Delta(\vrh^0+2\,\vv_a^0)\)\vartheta^{\, \perp}_{\hf\rho+\vd_a}\!(\tau),
\label{resN2}
\ee
where
\be
\Delta(x)=\hf\(\delta^{(2)}_{x}-1\)+\left\{\frac{x}{2}\right\}
\label{FunDx}
\ee
and we used the notation \eqref{defthetaL2} for the theta series defined by $\Lambda_S^\perp$
(replacing in the upper index the lattice by $\perp$ to avoid cluttering).

The r.h.s. of \eqref{resN2} crucially depends on the lattice $\Lambda_S$ and on the choice of the null vector $\nv$.
It also involves the glue vectors which are not unique and should drop out from the final result.
Due to all these complications, we are not able to proceed in full generality anymore.
Instead, in Appendices \ref{sec-Fm} and \ref{sec-Bm}, we analyze the condition \eqref{resN2}
case by case: for all possible null vectors for Hirzebruch surfaces
and several null vectors for del Pezzo surfaces.
In all these cases we find that it has a solution possessing
the required modular properties.
Furthermore, all found solutions (Eqs. \eqref{phi2F0}, \eqref{phi2F1}, \eqref{phi2F1p}, \eqref{phi2F2}, \eqref{phi2Bm1},
\eqref{phi2Bm2} and \eqref{phi2Bm3}) can be summarized by a single equation.
To this end, let $\nv'\in\Lambda_S$ be a null vector such that $\nv\cdot\nv'=1$.\footnote{For the Hirzebruch surface $\Fb_1$
such $\nv'$ does not exist. But for $b_2=2$, which is the case for $\Fb_1$, this vector is not required to define
the solution \eqref{genphi2}.}
Then the orthogonal complement to the two vectors, $\nv$ and $\nv'$, in $\Lambda_S$
is a unimodular and negative definite sublattice.
We denote $E_I$, $I=1,\dots, b_2-2$, its orthonormal basis.
In terms of these data, the solution is given by
\be
\phi_{2,\mu}(\tau,z)=\delta^{(2)}_{\nv\cdot\mu}\,
\frac{2\I\kappa\,\eta(\tau)^3}{\cdv\theta_1(\tau,4 \kappa z)} \prod_{I=1}^{b_2-2}
\theta_{\mu\cdot E_I}^{(2)}(\tau,\kappa_I z),
\label{genphi2}
\ee
where $\theta_{\ell}^{(2)}(\tau,z)$ is the vector valued Jacobi form of weight 1/2 and index $-1$ defined in \eqref{deftheta2z}
and $\kappa,\kappa_I$ are integer\footnote{They must be integer to ensure the existence of zero not only for $z=0$, but also for $z=\hf$.}
parameters restricted by the condition
\be
8\kappa^2-\sum_{I=1}^{b_2-2}\kappa_I^2=c_1^2.
\label{kappac}
\ee
We conjecture that this solution continues to hold for {\it all} null vectors of del Pezzo surfaces satisfying $\cdv>0$.

The solution \eqref{genphi2} still has some ambiguity: the parameters $\kappa_I$ and the choice of the second null vector
$\nv'$ defining the basis $E_I$. Fortunately, the latter is delusive because it is easy to see that
for $\nv\cdot\mu=0$, the quantity $\mu\cdot E_I$ is independent of this choice.\footnote{This follows
from the observation that changing $\nv'$, one changes $E_I$ by a vector proportional to $\nv$. In particular,
it should now be clear that although the choice of $E_I$ described above seems to disagree
with the one implicitly done in \eqref{phi2Bm1} and \eqref{phi2Bm2} (our prescription implies
$E_I=n\nv - D_{I+2}$ with some $n\ne 0$, whereas both results suggest $E_I=-D_{I+2}$), the solution is actually the same. \label{footE}}
The former however remains. At this stage it is not clear to us what condition would allow to fix it.

\subsubsection{$N=3$}
\label{subsubsec-phiN3}

At the next order, Theorem \ref{th-main} tells us that
\be
g_{3,\mu}= \phi_{3,\mu}+\hf \sum_{\mu'\in \Lambda_S/2\Lambda_S}
\(\vartheta^{(2,1)}_{\mu,\mu'}(\Phi_2)+\vartheta^{(1,2)}_{\mu,\mu'}(\Phi_2)\)\phi_{2,\mu'}
+\vartheta^{(1,1,1)}_{\mu}(\Phi_3),
\label{g3expl}
\ee
where
\bea
&& \vartheta^{(2,1)}_{\mu,\mu'}(\Phi_2)= \!\!\!\sum_{k\in\Lambda_S+\frac{1}{3}\mu+\frac{1}{2}\mu'}\!\!\!
\Bigl(\sgn(c_1\cdot k)-\sgn(\nv\cdot(k+\beta c_1))\Bigr)
\,\q^{-3k^2}\, y^{6c_1\cdot k},
\\
&&\vartheta^{(1,1,1)}_{\mu}(\Phi_3)=
\frac14\sum_{k_1\in \Lambda_S+\frac13\mu}\sum_{k_2\in \Lambda_S+\frac23\mu}
\q^{-k_1^2-k_2^2+k_1k_2}\, y^{2c_1\cdot (k_1+k_2)}\, \biggl[
\frac13\, \delta_{c_1\cdot k_1=c_1\cdot k_2=0}
\biggr.
\label{theta21}\\
&&\biggl.\quad
+\Bigl(\sgn(c_1\cdot k_1)-\sgn(\nv\cdot(2k_1-k_2+2\beta c_1))\Bigr)
\Bigl(\sgn(c_1\cdot k_2)-\sgn(\nv\cdot(2k_2-k_1+2\beta c_1))\Bigr)\biggr],
\nn
\eea
and $\vartheta^{(1,2)}_{\mu,\mu'}$ differs from $\vartheta^{(2,1)}_{\mu,\mu'}$ only by the sign of $\mu$
appearing in the range of summation.
Using the sign identity \eqref{signident}, it is easy to show that the terms with non-vanishing
$\nv\cdot k_i$ in \eqref{theta21} can be recombined into a function having zero of second order at $z=0$,
in agreement with Conjecture \ref{conj-smooth}.\footnote{Note that the presence of the delta symbol term is essential for this property.}
Moreover, since the power of $y$ is $2c_1\cdot (k_1+k_2)\in 2\IZ$, at $z=1/2$ one also has zero of the same order.

The remaining zero mode terms are split into two classes depending on the order of the zero mode.
The first class consists of the terms for which one and only one of the following three scalar products
vanishes: $\nb_{12}=\nv\cdot(2k_1-k_2)$, $\nb_{23}=\nv\cdot(2k_2-k_1)=0$, or $\nb_{13}=\nv\cdot (k_1+k_2)$.
Combining these three contributions, for $\beta\ll 1$ one can obtain (see Appendix \ref{ap-detail3} for details)
\bea
&& \frac14\sum_{k_1\in \Lambda_S+\frac13\mu}\sum_{k_2\in \Lambda_S+\frac23\mu}
\q^{-k_1^2-k_2^2+k_1k_2}\,
\biggl[\delta_{\nb_{13}=0}\(y^{3c_1\cdot k_1}-y^{-3c_1\cdot k_1}\)\(y^{c_1\cdot(k_1-2k_2)}-y^{c_1\cdot(2k_2- k_1)}\)
\biggr.
\nn\\
&& \times \Bigl(\sgn(c_1\cdot k_1)-\sgn(c_1\cdot(k_1+k_2))\Bigr)
\Bigl(\sgn(c_1\cdot(k_1- k_2))-\sgn(\nv\cdot (k_1-k_2))\Bigr)
\label{theta21one}\\
&&
+\delta_{\nb_{12}=0}\,y^{2c_1\cdot (k_1+k_2)}\Bigl(\sgn(c_1\cdot(2k_1-k_2))-\sgn(\beta)\Bigr)
\Bigl(\sgn(c_1\cdot k_2)-\sgn(\nv\cdot(k_2+2\beta c_1))\Bigr)
\nn\\
&& \biggl.
+\delta_{\nb_{23}=0}\,y^{2c_1\cdot (k_1+k_2)}\Bigl(\sgn(c_1\cdot(2k_2-k_1))-\sgn(\beta)\Bigr)
\Bigl(\sgn(c_1\cdot k_1)-\sgn(\nv\cdot(k_1+2\beta c_1))\Bigr)\biggr].
\nn
\eea
The first contribution has zero of second order at $y=\pm 1$ and therefore does not play any role in our analysis.
To understand the other two contributions,
one should diagonalize the quadratic form by a change of the summation variables.
Then following the discussion in Appendix \ref{ap-ident} around \eqref{ident2lat}, one obtains
that these contributions are equal to
\be
\frac14\sum_{\mu'\in \Lambda_S/2\Lambda_S}
\(\tvtheta^{(2,1)}_{\mu,\mu'}+\tvtheta^{(1,2)}_{\mu,\mu'}\)
\sum_{\ell\in \Lambda_S+\hf\mu'}\delta_{\nv\cdot\ell=0}\,
\Bigl(\sgn(c_1\cdot\ell)-\sgn(\beta)\Bigr)
\,\q^{-\ell^2}\,y^{2c_1\cdot \ell},
\label{theta21one2}
\ee
where we used the notation introduced in \eqref{septheta} for the part of the theta series
which excludes the zero modes, i.e. terms with $\nv\cdot k=0$.
The second factor here is exactly the function $\cvtheta^{(1,1)}_{\mu'}$ analyzed in the previous subsection,
whose singularity is canceled by $\phi_{2,\mu'}$.
Therefore, the contribution \eqref{theta21one2} is combined with the second term in \eqref{g3expl},
more precisely, with its part obtained by replacing $\vartheta_{\mu,\mu'}$ with $\tvtheta_{\mu,\mu'}$.
Taking into account that $\tvtheta^{(2,1)}_{\mu,\mu'}+\tvtheta^{(1,2)}_{\mu,\mu'}$
has zero at $y=\pm 1$, what can be easily seen by replacing $k\to -k$ in the second theta function,
one concludes that the zero modes of first order, together with the corresponding part of the
$\phi_{2,\mu}$-dependent term, satisfy the required regularity property.

The second class of terms that we need to consider consists of those which satisfy $\nb_{12}=\nb_{23}=0$.
To evaluate them, one should use again the technique of lattice factorization. Applying the same factorization
as in $\S$\ref{subsec-lat} for the two lattices $\Lambda_S$ one sums over in \eqref{theta21},
one finds that the relevant terms are given by
\be
\begin{split}
&\,
\frac14\sum_{a_1,a_2=0}^{r_S-1}  \delta^{(3)}_{\vrh^1+3\vv_{a_1}^1}\delta^{(3)}_{2\vrh^1+3\vv_{a_2}^1}
\vartheta^{\, \perp,2}_{\frac13\rho+\vd_{a_1},\frac23\rho+\vd_{a_2}}\!\!(\tau)\,
\Biggl[\frac13\, \delta^{(3)}_{\vrh^0+3\vv_{a_1}^0}\delta^{(3)}_{2\vrh^0+3\vv_{a_2}^0}
\Biggr.
\\
+ & \Biggl.
\(\frac{2\,y^{2\cdv\{\frac13\vrh^0+\vv_{a_1}^0\}}}{1-y^{2\cdv}}-\delta^{(3)}_{\vrh^0+3\vv_{a_1}^0}\)
\(\frac{2\,y^{2\cdv\{\frac23\vrh^0+\vv_{a_2}^0\}}}{1-y^{2\cdv}}-\delta^{(3)}_{2\vrh^0+3\vv_{a_2}^0}\)
\Biggr] ,
\end{split}
\label{cvthetaN3}
\ee
where $\vartheta^{\, \perp,2}_{\rho_1,\rho_2}$ denotes the theta series \eqref{defthetaL3} for $\Lambda=\Lambda_S^\perp$.
The two factors in the second line come from the geometric progressions as in \eqref{cvthetaN2},
which are factorized due to vanishing of all cross-terms in the quadratic form.
Finally, one has to take into account the part of the $\phi_{2,\mu}$-dependent term which was missed so far.
It is given by
\bea
&& \hf \sum_{\mu'\in \Lambda_S/2\Lambda_S}
\(\cvtheta^{(2,1)}_{\mu,\mu'}+\cvtheta^{(1,2)}_{\mu,\mu'}\)\phi_{2,\mu'}
\label{N3termphi}
\\
&=&
\hf \sum_{\mu'\in \Lambda_S/2\Lambda_S}\phi_{2,\mu'}
\sum_{a=0}^{r_S-1}\sum_{\eps=\pm}
\delta^{(1)}_{\frac{\eps}{3}\vrh^1+\frac12\vrh'^1+\vv_a^1}
\(\frac{2\,y^{6\cdv\{\frac{\eps}{3}\vrh^0+\hf\vrh'^0+\vv_a^0\}}}{1-y^{6\cdv}}
-\delta^{(1)}_{\frac{\eps}{3}\vrh^0+\hf\vrh'^0+\vv_a^0}\)\vartheta^{\, \perp}_{\frac{\eps}{3}\rho+\hf\rho'+\vd_a}\!(3\tau),
\nn
\eea
where the theta series are evaluated in the same way as $\cvtheta^{(1,1)}_{\mu}$ in the previous subsection.

The next step is to analyze the expansion of \eqref{cvthetaN3} and \eqref{N3termphi} around $z=0$
keeping terms up to $O(z)$. This can be done exactly as in Appendices \ref{sec-Fm} and \ref{sec-Bm},
although the analysis for del Pezzo surfaces becomes very cumbersome and requires the extensive use of various identities
between theta functions derived in Appendix \ref{ap-ident}. We provide a few intermediate steps in Appendix \ref{ap-detail3}.
As a result, one arrives at the following condition on $\phi_{3,\mu}$,
ensuring the existence of zero of second order in \eqref{g3expl},
\be
\phi_{3,\mu}\sim
\delta^{(3)}_{\nv\cdot\mu}\,
\frac{1+2\(\,\sum\limits_{I=1}^{b_2-2}\kappa_I^2\,\p_z^2 \log\theta^{(3)}_{\mu\cdot E_I}(\tau,0)
-32\kappa^2\cD(\tau)\) z^2}{3(4\pi \cdv z)^2}
\,\prod_{I=1}^{b_2-2}\theta^{(3)}_{\mu\cdot E_I}(\tau,0)
+O(z^2),
\label{expan-phi3}
\ee
where $\theta^{(3)}_{\ell}$ is defined in \eqref{deftheta3z}
and we denoted by $\cD$ the subleading term in the expansion around $z=0$ of $\theta_1(\tau,z)$,
\be
\begin{split}
\theta_1(\tau,z)=&\, -2\pi z\,\eta(\tau)^3 \Bigl(1+z^2 \cD(\tau) + O(z^4)\Bigr).
\end{split}
\label{defderiv}
\ee
Remarkably, the solution to the condition \eqref{expan-phi3} can be found without explicit knowledge of this function.
Indeed, it is straightforward to see that the following function satisfies all the required properties
\be
\phi_{3,\mu}=
-\delta^{(3)}_{\nv\cdot\mu}\,
\frac{4\kappa^2\,\eta(\tau)^{6}\,\theta_1(\tau,2 \kappa z)}{\cdv^2\,\theta_1(\tau,6\kappa z)\,\theta_1(\tau,4\kappa z)^2}
\prod_{I=1}^{b_2-2}\theta^{(3)}_{\mu\cdot E_I}(\tau,\kappa_I z).
\label{genphi3}
\ee

\subsubsection{Arbitrary rank}
\label{subsubsec-anyN}

The results \eqref{genphi2} and \eqref{genphi3} are in fact very suggestive: the factor constructed from
the Dedekind eta function and the Jacobi theta function can be recognized as
the generating function of stack invariants, mentioned in the end of section \ref{subsec-VWinv}.
For Hirzebruch surfaces, evaluated in the chamber of the moduli space with $J=[f]$ corresponding to
the fiber of the bundle defining $\Fb_m$ (see Appendix \ref{sec-Fm}),
this function is given by \cite{Manschot:2011ym,Mozgovoy:2013zqx}
\be
\label{defHN}
H_{N}(\tau,z) = \frac{\I (-1)^{N-1} \eta(\tau)^{2N-3}}
{\theta_1(\tau,2Nz)\, \prod_{m=1}^{N-1} \theta_1(\tau,2mz)^2}\, .
\ee
After normalization by $(\eta(\tau)\theta_1(\tau,2z))^N=(\eta(\tau)^{b_2-2} h_{1,0}(\tau,z))^{-N}$,
for $N=2$ and 3 (and evaluated at $\kappa z$) it exactly coincides with
the corresponding factors.
Furthermore, the theta series $\theta^{(N)}_{\ell}$ appearing in the last factor
in \eqref{genphi2} and \eqref{genphi3} are nothing else but a rescaled version of
the ``blow-up functions" \cite{Yoshioka:1996,0961.14022,Li:1998nv}
\be
\label{defBNk}
B_{N,\ell}(\tau,z) = \frac{1}{\eta(\tau)^N} \sum_{\sum_{i=1}^N a_i=0 \atop a_i\in\IZ+\frac{\ell}{N}}
\q^{- \sum_{i<j} a_i a_j}\, y^{\sum_{i<j}(a_i-a_j)},
\ee
which relate the generating functions of stack invariants on manifolds connected by the blow-up of an exceptional divisor.
This observation suggests the following
\begin{conj}
\label{conj-phi}
Choosing the holomorphic modular ambiguity as
\be
\phi_{N,\mu}=
\delta^{(N)}_{\nv\cdot\mu}\(\frac{2\kappa}{\cdv}\)^{\! N-1}
\frac{H_{N}(\tau,\kappa z)}{h_{1,0}^{N}(\tau,\kappa z)}\,
\prod_{I=1}^{b_2-2}B_{N,\mu\cdot E_I}(\tau,\kappa_I z),
\label{genphiN}
\ee
where $\kappa,\kappa_I$ are restricted to satisfy \eqref{kappac},
ensures the normalized generating functions $g_{N,\mu}$ to have zeros of order $N-1$ at $z=0$ and $z=\hf$.
\end{conj}
The modular properties of \eqref{genphiN} follow from that of $H_{N}$ and $B_{N,\mu}$: the former is
a Jacobi form of weight $-1$ and index $-\frac23(2N^3+N)$, whereas the latter is a vector valued Jacobi form of weight $-\hf$
and index $\frac16(N^3-N)$. Then, as required, $\phi_{N,\mu}$ is a vector valued Jacobi form of weight $\frac12(N-1)b_2(S)$ and
index $-\frac16\, (N^3-N)K_S^2$ due to Proposition \ref{prop-part} of Appendix \ref{ap-theta} and the condition on $\kappa$'s.

\subsection{Comparison to the known results}
\label{subsec-true}

So far we kept various parameters of our solution, such as $\kappa,\kappa_I$ and the null vector $\nv$, arbitrary.
However, the generating functions constructed using different parameters are unlikely the same.
So what is the right choice of these parameters?

It can easily be established by comparing with the known results at $N=2$.
It is immediate to see that for Hirzebruch surfaces one should take the null vector to be
$\nv=D_1=[f]$, whereas the parameter $\kappa$ is automatically fixed by condition \eqref{kappac} to $\kappa=1$.
This implies the following normalized generating function
\be
g_{2,\mu}^{\Fb_m}=\delta^{(2)}_{\mu_1}\,
\frac{\I \,\eta(\tau)^3}{\theta_1(\tau,4 z)}
+\hf\!\!\sum_{k_\alpha\in \IZ+\hf\mu_\alpha}\!\!
\Bigl(\sgn((2-m)k_1+2k_2)-\sgn(k_1+2\beta)\Bigr)\,
\q^{ mk_1^2-2k_1k_2}\, \y^{ 2((2-m)k_1+2k_2)},
\label{resg2Fm}
\ee
which coincides, for instance, with the results presented in $\S$3.1 of \cite{Manschot:2011dj} after
their specification to the canonical chamber $J=-K_{\Fb_m}$.\footnote{Eq. \eqref{resg2Fm} is obtained by setting
$k=C^{\alpha\beta}k_\alpha D_\beta$ in \eqref{resg2}. To match \cite{Manschot:2011dj},
one should change the variables $m=\ell$, $(k_1,k_2)=(b-\hf\beta,-a+\hf\alpha)$ and set there the parameters of the K\"ahler form as
$(m,n)=(2,2-\ell)$.}
Similarly, the result for $g_{3,\mu}$, which can be read off from \eqref{g3expl}, matches the generating function
found in section 5.2 of \cite{Manschot:2011ym}.

For del Pezzo surfaces, the right null vector is $\nv=D_1-D_2$ (``choice I" analyzed in Appendix \ref{sec-Bm})
which allows to take $E_I=-D_{I+2}$ (see footnote \ref{footE}), whereas $\kappa$'s are trivial:
\be
\kappa=\kappa_I=1.
\label{choickap}
\ee
Then the holomorphic modular ambiguity becomes identical to the normalized generating function of stack invariants in
the chamber $J=\nv$, obtained from \eqref{defHN} by applying the ``blow-up formula" $m-1$ times \cite{Haghighat:2012bm}:
\be
\phi_{2,\mu}^{\Bb_m}(\tau,z)=h_{1,0}^{-N}(\tau, z)\, H_{N,\mu}^{\Bb_m}(\tau,z;J=\nv).
\ee
Again, one can check that our results are in perfect agreement with those given in $\S$4.2 of \cite{Haghighat:2012bm}.

Note that typically the generating functions in the cited references are given in a {\it different} form
than the one obtained in this paper. As follows from Theorem \ref{th-main}, we derive them as combinations
of theta functions $\vartheta^{(\vec N)}_{\mu,\vec\mu}(\Phi_n)$ and the modular functions $\phi_{N,\mu}$.
Instead, in \cite{Manschot:2011ym,Haghighat:2012bm} they are represented as combinations of the same objects plus
certain rational functions of $y$. These functions arise due to sheaves which are strictly semi-stable for $J=\nv$
and due to differences between stack and VW invariants. Essentially, their computation represents
the most non-trivial part of constructing $g_{N,\mu}$. In turns out that these rational functions
are nothing else but the zero modes discussed in $\S$\ref{sec-holom} and used to determine the holomorphic modular ambiguity.
Thus, in our case they are {\it automatically} included into the theta series.
This crucial simplification is achieved due to the presence of the parameter
$\beta=-\frac{\Im(y)}{2\pi\tau_2}$ in a half of the sign functions
building the kernel of the theta series (see \eqref{defbk}). We emphasize that we did not put it there by hand, but
it is a direct consequence of modularity and the form of the modular completion.\footnote{One may wonder
how the dependence on $\beta$ could arise given that the refined VW invariants are defined as holomorphic polynomials in $y$ \eqref{intOm}.
In fact, this piecewise dependence is an artefact of the representation of the generating functions through theta series
and appears due to their divergence at a discrete set of poles. Indeed, as discussed below, the $\beta$-dependent part
of theta series can be resumed into generalized Appell functions. They are meromorphic functions with poles corresponding to
the same values of $\beta$ where the $\beta$-dependent sign functions of theta series jump.
Viewed as poles in the $\tau$-plane, they give rise to an ambiguity in the choice of the contour
defining the Fourier coefficients, and the original theta series representation corresponds
to the prescription where the integration is performed keeping fixed $\alpha$ and $\beta$, but not $z$.
For instance, in \eqref{ApFm} this is manifested by the fact that
it reproduces the original theta series by expanding the denominator in the two opposite ways
depending on whether the value of $|\q^{2n}y^4|=e^{4\pi\tau_2(2\beta-n)}$ is greater or less than 1.
Clearly, this condition corresponds to the last sign function in \eqref{resg2Fm}.
Of course, for the purpose of extracting VW invariants, refined or unrefined, it is sufficient to restrict to
the region $|\beta|\ll 1$ where all these subtleties can be ignored.
}

Another representation used to express the generating functions for Hirzebruch surfaces
involves generalized Appell functions \cite{Manschot:2011dj,Manschot:2014cca}.
It can be obtained from our theta functions by splitting them into two parts:
one corresponds to the contributions generated by wall-crossing from $J=D_2$ to $J=-K_S$,
whereas the other can be resumed into an Appell function.
For instance, for $N=2$ and $m>0$ this representation is given by \cite{Manschot:2011dj}
\be
\begin{split}
g_{2,\mu}^{\Fb_m}=&\, \delta^{(2)}_{\mu_1}\,
\frac{\I \,\eta(\tau)^3}{\theta_1(\tau,4 z)}+A^{(m)}_{\mu_1,\mu_2}(\tau,z)
\\
&\,
+\hf\sum_{k_\alpha\in \IZ+\hf\mu_\alpha}
\Bigl(\sgn((2-m)k_1+2k_2)-\sgn(k_2)\Bigr)\,
\q^{ mk_1^2-2k_1k_2}\, \y^{ 2((2-m)k_1+2k_2)},
\end{split}
\label{resg2FmAp}
\ee
where
\be
A^{(m)}_{\mu_1,\mu_2}(\tau,z)=\sum_{n\in\IZ+\hf\mu_1}\frac{\q^{mn^2+\mu_2 n}y^{2(m-2)n+2\mu_2}}{1- \q^{2n}y^4}
-\hf\,\delta^{(2)}_{\mu_2}\, \theta^{(2)}_{\mu_1}(m\tau,(m-2)z)
\label{ApFm}
\ee
and the sum over $n$ can be expressed through the level-$m$ Appell function \eqref{Alevel}.
The zero modes are hidden in $A^{(m)}_{\mu_1,\mu_2}$ and can be extracted as its $\q$-independent contribution.
Again we observe that our representation is significantly simpler.
It is this simplification that allowed to provide a closed formula for all ranks and to make explicit the modular properties
of the generating functions.

\section{Null vectors and duality}
\label{sec-ident}

\subsection{Fiber-base duality}

Let us consider the Hirzebruch surface $\Fb_0$.
As noticed in Appendix \ref{sec-Fm}, the intersection matrix of $\Fb_0$ and its first Chern class are symmetric under the exchange
of the basis vectors $D_1$ and $D_2$. This implies that our construction based on the second null vector $\nv'$ in \eqref{nullF0}
gives rise to the same generating functions as the one based on $\nv$, up to the exchange of the components of the residue class $\mu$.
On one hand, these components have different geometric meaning since one corresponds to the fiber and the other to the base
of the bundle. On the other hand, $\Fb_0=\IP^1\times\IP^1$ and hence everything should be symmetric under the exchange
of the fiber and the base. This is the simplest example of the fiber-base duality \cite{Katz:1997eq,Mitev:2014jza}.
In our case, this duality requires that
\be
g^{\Fb_0}_{N,(\mu_1,\mu_2)}=g^{\Fb_0}_{N,(\mu_2,\mu_1)}.
\label{dualF0}
\ee
Let us check this relation for low ranks.

At $N=2$, using \eqref{resg2Fm}, one finds
\be
\begin{split}
g^{\Fb_0}_{N,(1,0)}=g^{\Fb_0}_{N,(0,1)}\ \Rightarrow\ &\,
\hf\sum_{k_1\in\IZ\atop k_2\in \IZ+\hf}\Bigl( \sgn(k_1+2\beta)-\sgn(k_2+2\beta)\Bigr) \q^{-2k_1 k_2} \, y^{4(k_1+k_2)}
\\
=&\, \sum_{n\in \IZ} \frac{\q^n y^{4n-2}}{1-\q^{2n} y^{-4}}
=\frac{\I \,\eta(\tau)^3}{\theta_1(\tau,4 z)}\, .
\end{split}
\ee
This relation is easily checked numerically, but it can also be verified analytically.
An instructive way to do this is to pass via the surface $\Fb_2$.
In Appendix \ref{sec-Fm} it is shown that all geometric data of $\Fb_2$ are mapped to those of $\Fb_0$ by
a simple change of basis \eqref{chbF0F2}. Since our construction is uniquely determined by these data
and by the choice of the null vector, which are also mapped to each other,
one immediately concludes that\footnote{Recall that
the components $\mu_\alpha$ in this section are related to $\mu^\alpha$ in Appendix \ref{sec-Fm} by $\mu_\alpha=C_{\alpha\beta}\mu^\beta$.}
\be
g^{\Fb_2}_{N,(\mu_1,\mu_2)}=g^{\Fb_0}_{N,(\mu_1,\mu_2-\mu_1)}.
\label{dualF2F0}
\ee
This proves the relation between VW invariants of $\Fb_2$ and $\Fb_0$ noticed in \cite{Beaujard:2020sgs}.
Combined with \eqref{dualF0}, this implies a fiber-base duality for $\Fb_2$ which gives
\be
g^{\Fb_2}_{N,(\mu_1,\mu_2)}=g^{\Fb_2}_{N,(\mu_2-\mu_1,\mu_2)}.
\label{dualF2}
\ee
Taking $\mu_1=\mu_2=1$ and using the representation \eqref{resg2FmAp}, which is convenient because the last term vanishes for $m=2$,
on can show that it boils down to the following identity for the Appell--Lerch function \eqref{clAp}
\be
A(\tau,4z,0)
=-\frac{\I\,\eta(\tau)^3}{\theta_1(\tau,4z)}\,.
\ee
This identity in turn follows directly from the periodicity relation \eqref{Appellperiod} evaluated at $u=4z$, $v=0$.
Note that it specifies an Appell--Lerch function which is a true modular form.

At $N=3$, the normalized generating function is obtained from \eqref{g3expl}.
A direct evaluation gives
\bea
g^{\Fb_2}_{3,\mu}&=& -\delta^{(3)}_{\mu_1}\,
\frac{\eta(\tau)^{6}\,\theta_1(\tau,2 z)}{\theta_1(\tau,6 z)\,\theta_1(\tau,4 z)^2}
+\frac{\I \,\eta(\tau)^3}{\theta_1(\tau,4 z)}
\sum_{\eps=\pm}\sum_{k\in\IZ+\frac{\eps\mu_1}{3}}\q^{6k^2}\(\frac{\q^{-2\eps \mu_2 k} \, y^{-4\eps\mu_2}}{1- \q^{3k}\, y^6}
-\hf\,\delta^{(3)}_{\mu_2}\)
\nn\\
&+&
\sum_{k_1\in \IZ+\frac13\, \mu_1\atop k_2\in \IZ+\frac23\, \mu_1}
\q^{2(k_1^2+k_2^2-k_1k_2)}
\(\frac{\q^{\frac{\mu_2}{3}(k_2-2k_1)} \, y^{\frac43 \mu_2}}{1- \q^{k_2-2k_1}\, y^4}-\hf\,\delta^{(3)}_{\mu_2}\)
\(\frac{\q^{\frac{2\mu_2}{3}(k_1-2k_2)} \, y^{\frac83 \mu_2}}{1- \q^{k_1-2k_2}\, y^4}-\hf\,\delta^{(3)}_{\mu_2}\).
\eea
At this order the only non-trivial relation in \eqref{dualF2} is again obtained for $\mu_1=\mu_2=1$. Expressing
the sums through the generalized Appell functions \eqref{Alevel} and \eqref{Ap2}, it can be represented as
\be
\begin{split}
&\, y^4 \Bigl( A^{(A_2)}_2\(\tau, 4z,4z,\tfrac{1}{3}\tau,\tfrac{2}{3}\tau\)-
\q\, A^{(A_2)}_2\(\tau, 4z-\tau,4z+\tau,-\tfrac{1}{3}\tau,\tfrac{4}{3}\tau\)
\Bigr)
\\
&\,+ \frac{\I \,\eta(\tau)^3}{\theta_1(\tau,4 z)}
\Bigl[ y^{4}\( A_4\(3\tau,6z, \tfrac{1}{2}\,\tau\)-A_4\(3\tau,6z-\tau, -\tfrac{1}{2}\,\tau\)\)
\Bigr.
\\
&\, \Bigl.\qquad
+ y^{-4}\(A_4\(3\tau,6z, -\tfrac{1}{2}\,\tau\)- A_4\(3\tau,6z+\tau, \tfrac{1}{2}\,\tau\)\)\Bigr]
= \frac{\eta(\tau)^{6}\,\theta_1(\tau,2 z)}{\theta_1(\tau,6 z)\,\theta_1(\tau,4 z)^2}.
\end{split}
\label{indentA}
\ee
We have checked this identity numerically which provides a strong evidence that it does hold.
In fact, there is a striking similarity between \eqref{indentA} and an identity proven in \cite[Theorem 1.1]{Bringmann:2015},
although they are not exactly the same differing in precise arguments of Appell and theta functions.
The identity in \cite{Bringmann:2015} is a consequence of the blow-up formula relating the generating functions for
$\IP^2$ and $\Fb_1$, whereas in our case the relevant surface is either $\Fb_0$ or $\Fb_2$.
This explains the differences between the two cases and shows that
the fiber-base duality can be taken as a source of non-trivial identities for the generalized Appell functions.

\subsection{Generalizations}

Remarkably, the duality relation \eqref{dualF0} generalizes to the del Pezzo surfaces as well.
We formulate this statement as
\begin{conj}
\label{conj-dual}
Let the two null vectors $\nv,\nv'\in \Lambda_S $ satisfy
\be
\nv\cdot\nv'=1,
\qquad
\cdv\equiv c_1\cdot \nv
=c_1\cdot \nv',
\label{conddual}
\ee
and the parameters $\kappa,\kappa_I$ are given by\footnote{They satisfy the condition \eqref{kappac} because
it becomes equivalent to the decomposition of $c_1^2$ in the basis $\{\nv,\nv', E_1,\dots,E_{b_2-2}\}$.}
\be
\kappa=\frac{\cdv}{2}\, ,
\qquad
\kappa_I=|c_1\cdot E_I|,
\label{choickapc}
\ee
where $\{E_I\}_{I=1}^{b_2-2}$ is the orthonormal basis of the orthogonal complement to $\nv$ and $\nv'$, as in \eqref{genphiN}.
Then the two sets of generating functions, constructed out of the null vectors $\nv$ and $\nv'$
using Theorem \ref{th-main} with $\phi_{N,\mu}$ given in \eqref{genphiN}, are the same
\be
g_{N,\mu}[\nv]=g_{N,\mu}[\nv'].
\label{dualDP}
\ee
\end{conj}
This conjecture follows from the observation that in the difference $g_{N,\mu}[\nv]-g_{N,\mu}[\nv']$ all
$\sgn(\Gamma_k)$ which appear in the kernel \eqref{kerg} of indefinite theta series cancel recursively, i.e.
assuming that the relation \eqref{dualDP} holds for $N'<N$. This is the only non-trivial step to prove.
It ensures that the only remaining sign functions are those which involve $\nv$ and $\nv'$, but not $c_1$.
Then the first condition in \eqref{conddual} allows to change the basis of the lattice to
$\{\nv,\nv', E_1,\dots,E_{b_2-2}\}$ so that $\Lambda_S$ appears as a direct sum of two unimodular orthogonal sublattices.
As a result, the part generated by $\{E_I\}_{I=1}^{b_2-2}$
decouples giving rise exactly to the same factor $\prod_{I=1}^{b_2-2}B_{N,\mu\cdot E_I}(\tau,\kappa_I z)$
(rescaled by a power of $\eta(\tau)$) which arises in $\phi_{N,\mu}$ \eqref{genphiN}. After this decoupling,
the second condition in \eqref{conddual} ensures that the relation \eqref{dualDP} reduces
to the one for $\Fb_0$ with $y$ replaced by $y^\kappa$ and hence must hold.

Of course, the pairs of null vectors \eqref{nullF0} for $\Fb_0$ and \eqref{nullF2} for $\Fb_2$ satisfy the conditions
\eqref{conddual} and the conjecture reduces to the duality relations \eqref{dualF0} and \eqref{dualF2}, respectively.
In contrast, the null vectors \eqref{nullF1} for $\Fb_1$ spoil both conditions and therefore one does not expect any
fiber-base duality for VW invariants on $\Fb_1$.

For del Pezzo surfaces there is an evident pair of vectors satisfying \eqref{conddual}: $\nv=D_1-D_2$ and $\nv'=D_1-D_{I+1}$
for any $I=1,\dots,m-1$. However, the corresponding duality is trivial since it follows from the symmetry between
the exceptional divisors. A non-trivial duality can be obtained by taking the second null vector $\nv'$
to be one of the following divisors
\be
2D_1-\sum_{\alpha=2}^5 D_\alpha,
\qquad
3D_1-2D_2-\sum_{\alpha=3}^7 D_\alpha,
\qquad
4D_1-3D_2-\sum_{\alpha=3}^9 D_\alpha,
\ee
which exist for $m\ge 4$, 6 and 8, respectively.
The resulting duality is a non-trivial prediction of our formalism and it can be traced back to the Weyl
reflection symmetry of the lattice $\Lambda_{\Bb_m}$ \cite{Iqbal:2001ye}.\footnote{The author is grateful to Boris Pioline
for this observation.}

\subsection{New invariants?}

If correct, Conjecture \ref{conj-phi} together with Theorem \ref{th-main}
provide holomorphic functions, as well as their completions,
satisfying the modular anomaly equation and regularity conditions,
for any null vector $\nv$ and parameters $\kappa,\kappa_I$ restricted only by a single constraint \eqref{kappac}.
For the vectors and parameters specified in $\S$\ref{subsec-true},
and for those which form with them the dual pairs satisfying the conditions of Conjecture \ref{conj-dual},
these holomorphic functions are the generating functions of refined VW invariants.
But what is the meaning of the holomorphic functions built out of other null vectors?

For instance, for the Hirzebruch surface $\Fb_1$ there is a null vector $\nv'=2D_2-D_1$ which has $\cdv'=4$
(see Appendix \ref{sec-Fm}) and gives rise to holomorphic functions which are {\it different} from those of VW invariants.
Nevertheless, our results indicate that they still possess the same modular transformation properties
and have a well defined unrefined limit.
Thus, one can ask a question about interpretation of the rational numbers which appear as Fourier coefficients
of these functions.

It might be that our construction still misses some important constraints which
exclude these additional solutions. For instance, there is no guarantee that the rational numbers they generate
can be related to integer numbers in some reasonable way similar to \eqref{defcref}.
However, it is also not clear why such condition should be enough to exclude them from consideration.
We leave the interpretation of these solutions as an interesting open problem.

\section{Conclusions}
\label{sec-concl}

The main result of this paper is the explicit expression, given in \eqref{mainres},
for the generating functions of refined VW invariants in the canonical chamber $J=-K_S$
for Hirzebruch and del Pezzo surfaces.
We emphasize that this formula is supposed to work for {\it any} rank $N$ of the gauge group of VW theory
(or, mathematically speaking, rank of the semi-stable sheaves).
This result not only unifies and extends
many explicit expressions for such generating functions scattered in the literature,
but also simplifies them. In particular, it represents the generating functions
as combinations of indefinite theta series and true Jacobi forms which makes
the modular properties of these functions transparent.
As a result, we also provided a similar explicit expression for the modular completion of the generating functions.

Although our result is valid only in the canonical chamber, it can be translated to
other chambers of the moduli space with help of wall-crossing formulas \cite{ks,Joyce:2008pc,Joyce:2009xv},
which typically generates additional theta series.
Thus, this is not a serious restriction, but, of course, it would be desirable to have a similar
explicit and generic expression in arbitrary chamber.

We also believe that the restriction to Hirzebruch and del Pezzo surfaces can be lifted and
similar results should hold more generally. First of all, although we restricted ourselves to the del Pezzo surfaces
which are Fano, i.e. with $m\le 8$, the comparison of our results with those for $\IB_9$
\cite{Klemm:2012sx}, known as the rational elliptic surface or half-K3, shows that they continue to hold in this case as well.
The reason why we preferred to exclude $\IB_9$ from our analysis is that for this surface the canonical class is null,
so that such case, strictly speaking, is not captured by the construction of the modular completion
in \cite{Alexandrov:2018lgp,Alexandrov:2019rth} which requires $K_S^2\ne 0$.
Nevertheless, we see that the resulting formula for the completion \eqref{exp-whhr}
seems to have much larger area of applicability that one could initially think.

In fact, the construction in this paper relies just on a few ingredients:
\begin{itemize}
\item
the existence of the unrefined limit;
\item
the formula for the completion \eqref{exp-whhr};
\item
the unimodularity of the lattice $\Lambda_S$;
\item
the existence of a null vector belonging to the lattice.
\end{itemize}
Taking into account that the first holds by definition and assuming the second, one may hope
that the same construction applies to all surfaces satisfying the last two conditions, which are rather mild.
In particular, it would be interesting to verify whether our formula for the generating functions
continues to hold for the toric almost Fano surfaces considered recently in \cite{Beaujard:2020sgs}.

At the same time, the case of $\IP^2$ appears to be special since the one-dimensional lattice does not have null vectors.
Of course, it is known how to deal with such case: one should multiply by a modular theta series
thereby extending the lattice to two-dimensional \cite{Zwegers-thesis}.
Thus, it should also be possible to include $\IP^2$ in the present framework.
However, the details of the analysis might be different and we leave this case for future research.

A particularly interesting problem is to see whether it is possible to extend the present construction
to the case of a compact CY threefold. There are however important obstacles on this way.
The main difference with respect to the non-compact case is that the divisor classes may not be
proportional to each other so that they are not classified by a single number $N$.
As a result, the spectral flow decomposition of charges \eqref{quant-q}
and their symplectic product \eqref{gam12} will not be so simple anymore.
This might have important consequences for various steps in the derivation.
Finally, the lattice of charges might also be more complicated, in particular, not unimodular.
Nevertheless, even if this can make impossible to obtain some general results,
one can still hope that some particular cases might be tractable.

Considering generalizations of the results presented in this paper to more general cases,
it is useful to keep in mind the possibility to have non-trivial $\kappa$-parameters appearing in \eqref{genphiN}.
Although for Hirzebruch and del Pezzo surfaces they must be set to one as in \eqref{choickap},
this is consistent with the constraint \eqref{kappac} only for $c_1^2(S)=10-b_2(S)$.
A more general prescription, which is also natural from the point of view of the fiber-base duality, is
the choice \eqref{choickapc}.

The duality which we observed is another interesting byproduct of our construction.
In this framework it has a very simple realization as an exchange of two null vectors satisfying
the conditions \eqref{conddual}. However, it is expected to have a geometric origin, as in the case of $\Fb_0$,
and important physical consequences, see e.g. \cite{Mitev:2014jza}.
To elucidate them, it is probably necessary to establish a precise relation with the analysis of the fiber-base duality
in the gauge theory context.

This duality has also an immediate mathematical application which we have already started to explore in this paper.
Namely, it can be taken as a generating technique of non-trivial identities between generalized Appell functions
and other special functions. We have obtained one such identity, Eq. \eqref{indentA}, but many more can be generated
by making the duality relations between generating functions explicit.

Finally, we indicated the existence of alternative solutions for the generating functions satisfying all modularity
and regularity constraints, which are constructed using different null vectors.
Are they artifacts of the formalism and should not be taken seriously?
Or they imply the existence of new topological invariants?
If yes, what is their physical and mathematical meaning?
These are open questions for future research.

\section*{Acknowledgements}

The author is grateful to Sibasish Banerjee, Jan Manschot and Boris Pioline for valuable discussions
and previous collaboration leading to this work.

\appendix

\section{Indefinite theta series and generalized error functions}
\label{ap-gener}

\subsection{Jacobi and modular forms}
\label{ap-Jacobi}

We define (in general, non-holomorphic) vector valued Jacobi form of weight $(w,\bw)$ and index $m$
as a finite set of functions $f_\bfmu(\tau, z)$ with $\tau\in \IH$, $z\in \IC$ labelled by $\bfmu$,
satisfying the following transformation properties \cite{MR781735}
\be
\begin{split}
f_{\bfmu}(\tau,z+k\tau+\ell)=&\, e^{- 2\pi\I m \( k^2\tau + 2 k z\)} \,f_{\bfmu}(\tau,z),
\\
f_{\bfmu}\(\frac{a\tau+b}{c\tau+d}, \frac{z}{c\tau+d}\)
=&\, (c\tau+d)^w(c\btau+d)^{\bw} \, e^{\frac{2\pi\I m c z^2}{c\tau+d}}\sum_\bfnu M_{\bfmu\bfnu}(\rho)\,f_{\bfnu}(\tau,z),
\end{split}
\label{Jacobi}
\ee
where $\rho=\scriptsize{\(\begin{array}{cc}
a & b \\ c & d
\end{array}\)}\in SL(2,\IZ)$ and $M_{\bfmu\bfnu}(\rho)$ is a multiplier system. Setting $z=m=0$, \eqref{Jacobi}
reduces to the definition of a vector valued modular form $f_\bfmu(\tau)$.
Since $M_{\bfmu\bfnu}(\rho)$ must furnish a representation of the group and $SL(2,\IZ)$ is generated by two transformations,
$T=\scriptsize{\(\begin{array}{cc}
1 & 1 \\ 0 & 1
\end{array}\)}$
and
$S=\scriptsize{\(\begin{array}{cc}
0 & -1 \\ 1 & 0
\end{array}\)}$,
to define the multiplier system, it is enough to specify it for $\rho=T$ and $S$.
Thus, to characterize the modular behaviour of a Jacobi form, it is sufficient to provide
its modular weight $(w,\bw)$, index $m$ and two matrices $M_{\bfmu\bfnu}(T)$ and $M_{\bfmu\bfnu}(S)$.

\subsection{Theta series}
\label{ap-theta}

In this work we consider theta series of the following type
\be
\label{Vignerasth}
\vartheta_{\bfmu}(\tau,z;\bfp,\bfptt,\Phi)=
\!\!\!\!
\sum_{{\bfk}\in \Lat+\bfmu+\hf\bfp}\!\!
(-1)^{\bfk\bpt\bfp}\,\Phi(\sqrt{2\tau_2}(\bfk+\beta\bfptt))\, \q^{-\tfrac{1}{2}\bfk^2} y^{\bfptt\bpt\bfk},
\ee
where as usual $\q=e^{2\pi\I\tau}$, $y=e^{2\pi\I z}$ and $z=\alpha-\tau\beta$.
Additionally, the definition involves the following data:
\begin{itemize}
\item
$\Lat$ is a $d$-dimensional lattice equipped with a bilinear form
$(\bfx,\bfy)\equiv \bfx\bpt\bfy$,
where $\bfx,\bfy\in \Lat \otimes \IR$, such that its associated quadratic form
has signature\footnote{It is convenient to choose mostly negative signature
(and hence the minus sign in the power of $\q$ in \eqref{Vignerasth}) because this is the case
for the signature of $H^2(S,\IZ)$ for all relevant surfaces and in agreement with conventions of
\cite{Alexandrov:2016tnf,Alexandrov:2016enp,Alexandrov:2018lgp,Alexandrov:2019rth}.}
$(n,d-n)$ and is integer valued, i.e. $\bfk^2\equiv \bfk\bpt \bfk\in\IZ$ for $\bfk\in\Lat$.

\item
$\bfp\in \Lat$ is a characteristic vector, i.e. such that $\bfk\bpt(\bfk+ \bfp)\in 2\IZ$, $\forall \,\bfk \in \Lat$.

\item
$\bfmu\in\Lat^*/\Lat$ is a glue vector or residue class.

\item
$\bfptt\in\Lat$ is an arbitrary vector in the lattice.

\item
$\Phi(\bfx)$ is the kernel of theta series satisfying suitable decay conditions on $\Phi(\bfx) e^{\pi \bfx^2}$.
\end{itemize}
The crucial result about these theta series proven in \cite{Vigneras:1977} is that,
provided the kernel satisfies the following differential equation (which we call the Vign\'eras equation)\footnote{This equation
ensures that the kernel defines an eigenfunction under a Fourier transform which appears in the Poisson resummation of the S-transformation
of the theta series explicitly computed in  \cite{Vigneras:1977}.}
\be
\label{Vigdif}
\Vop_\lambda\, \Phi(\bfx)=0,
\qquad
\Vop_\lambda= \partial_{\bfx}^2   + 2\pi \( \bfx\bpt \pa_{\bfx}  - \lambda\)
\ee
with a real parameter $\lambda$, under the standard $SL(2,\IZ)$ transformations of $\tau$ and $z$ \eqref{transf-tz},
$\vartheta_{\bfmu}$ transforms as a vector-valued Jacobi form of weight $(\hf(d+\lambda),-\hf\lambda)$, index $-\bfptt^2/2$,
and the multiplier system specified by
\be
M_{\bfmu\bfnu}(T)=e^{-\pi\I \(\bfmu+\tfrac12 \bfp\)^2}\,\delta_{\bfmu\bfnu},
\qquad
M_{\bfmu\bfnu}(S)=
\frac{(-\I)^{\lambda+\frac{n}{2}}}{\sqrt{|\Lat^*/\Lat|}}\,
e^{\frac{\pi\I}{2} \,\bfp^2} \,
e^{2\pi\I\bfmu\bpt\bfnu}.
\label{eq:thetatransforms}
\ee

There is an important particular case allowing to describe theta series which do not exactly fit the class \eqref{Vignerasth},
but still transform as in the same way. It is captured by the following

\begin{proposition}
\label{prop-part}
Let the glue vector has a decomposition
\be
\bfmu=\bfrho+a_A\bfxi^A+b^A\bfzeta_A,
\qquad
{ A=1, \dots, m_0,\atop a_A,b^A=0,\dots,n_0-1,}
\ee
where $\bfrho\in \bfDelta\subset \Lat^\star/\Lat$ and the vectors $\bfxi^A$, $\bfzeta_A$
satisfy $\bfxi^A\bpt\,\bfzeta_B=n_0^{-1}\delta^A_B$ with $n_0\in \IN$ as well as
$\bfxi^A\bpt\,\bfxi^B,\ \bfxi^A\bpt\,\bfp\in 2\IZ$ and $\bfxi^A\bpt\,\bfrho\in \IZ$.\footnote{Note that
we do not require the vectors to be linearly independent.}
Furthermore, we assume that
\be
\vartheta_{\bfrho+a_A\bfxi^A}=\vartheta^{(0)}_{\bfrho},
\qquad
\vartheta_{\bfrho+a_A\bfxi^A+b^A\bfzeta_A}= 0
\quad \mbox{\rm for } b^A\ne 0.
\ee
Then $\vartheta_{\bfmu}$ transforms according to \eqref{Jacobi} with multiplier system \eqref{eq:thetatransforms} if and only if
$\vartheta^{(0)}_{\bfrho}$ does so with $\Lat^\star/\Lat$ replaced by $\bfDelta$.
\end{proposition}
\begin{proof}
The agreement of the $T$-transformation of the two theta series
trivially follows from the properties of the vectors $\bfxi^A$ which ensure that for $b^A=0$
the phase does not depend on $a_A$.
For the $S$-transformation of $\vartheta_{\bfmu}$, the relevant factor is given by
\be
\sum_{\bfnu\in\Lambda^*/\Lambda}
e^{2\pi\I\bfmu\bpt\bfnu}\vartheta_{\bfnu}
=\sum_{\bfrho'\in\Delta}\sum_{a'_A=0}^{n_0-1}
e^{2\pi\I\((\bfrho+b^A\bfzeta_A)\bpt\bfrho'+\frac{b^A a'_A}{n_0}\)}\vartheta^{(0)}_{\bfrho'}
=n_0^{m_0}\[\prod_{A=1}^{m_0}\delta_{b^A=0}\]
\sum_{\bfrho'\in\Delta} e^{2\pi\I\bfrho\bpt\bfrho'}\vartheta^{(0)}_{\bfrho'}.
\ee
Taking into account that $|\Lat^\star/\Lat|=n_0^{2m_0}|\bfDelta|$, this
agrees with the transformation of $\vartheta^{(0)}_{\bfrho}$.
\end{proof}

Due to this proposition, some theta series which, comparing with \eqref{Vignerasth},
seem to be only a part of a modular vector, in fact transform themselves as (vector valued) modular forms.

\subsection{Unary theta series}
\label{ap-thJ}

The most simple example of the modular theta series is the Jacobi function
\be
\theta_1(\tau,z)=\sum_{k\in \IZ+\hf}\q^{k^2/2} (-y)^k,
\ee
which corresponds to the unimodular lattice $\Lat=\IZ$ with the characteristic vector $\bfp=1$, trivial kernel $\Phi=1$, and $\bfptt=1$.
Thus, it transforms as Jacobi form of weight 1/2 and index $1/2$.
This function satisfies an important property: while it vanishes at $z=0$, its first derivative gives
\be
\p_z\theta_1(\tau,0)=-2\pi \eta(\tau)^3,
\label{derth1}
\ee
where $\eta(\tau)$ is the Dedekind eta function, modular form of weight 1/2, which is given by
\be
\eta(\tau)=\q^{\frac{1}{24}}\prod_{n=1}^\infty(1-\q^n)=\q^{\frac{1}{24}}\sum_{n\in\IZ}(-1)^n \q^{(3n^2+n)/2} ,
\label{Dedfun}
\ee
where the second representation is due to Euler’s Pentagonal Number Theorem.

One can ask how the modularity of the eta function is consistent with the statement of the previous subsection.
Slightly rewriting \eqref{Dedfun}, one obtains
\be
\eta(\tau)=-\I \sum_{k\in 3\IZ+\hf}(-1)^k \q^{k^2/6}.
\ee
This corresponds to $\Lat=3\IZ$, $\bfp=3$ and $\bfmu=-1$.\footnote{We use conventions where the lattice $\Lat=N \IZ$
has bilinear form $k_1\bpt k_2=k_1k_2/N$ for $k_1,k_2\in \Lat$ which allows to be consistent with equations in VW theory.
Then the choice of $N$ is determined by the consistency with \eqref{Vignerasth}.}
However, since the determinant of the lattice is equal to 3, one could expect that $\eta$ is only one component of
a three-dimensional modular vector $\eta_\bfmu$. However, it is easy to see that $\eta_0=0$ and $\eta_1=-\eta_{-1}$,
and these relations are preserved under transformations \eqref{Jacobi} with
multiplier system \eqref{eq:thetatransforms} so that $\eta=\eta_{-1}$ is mapped to itself.
Furthermore, it is easy to check that they reproduce the correct multiplier system of the Dedekind function,
which can be found for instance in \cite{Siegel-book}.
This example can be seen as a generalization of the situation described by Proposition \ref{prop-part}.

\subsection{Theta series identities}
\label{ap-ident}

In this work we encounter several theta series defined on lattices with a negative definite quadratic form.
In this section, we provide their definitions as well as some useful identities which they satisfy.

Let $\Lambda$ be such a lattice. Then we define
\bea
\vartheta^{\,\Lambda}_\mu(\tau)&=&\sum_{k\in\Lambda+\mu}\q^{-k^2} ,
\label{defthetaL2}
\\
\vartheta^{\,\Lambda,2}_{\mu_1,\mu_2}(\tau)&=&\sum_{k_a\in\Lambda+\mu_a, \ a\scriptscriptstyle{=1,2}}\q^{-k_1^2-k_2^2+k_1\cdot k_2} .
\label{defthetaL3}
\eea
These are theta series fitting the definition \eqref{Vignerasth} for $\Lat=2\Lambda$ and $\Lat=A_2\otimes\Lambda$, respectively,
with $\bfp=\bfptt=0$ and $\Phi=1$. Since both quadratic forms are even, both theta series are vector valued modular forms
of weights 1/2 and 1.
Note, however, that we do not restrict the residue classes $\mu$ and $\mu_i$
to run only over values corresponding to the independent components of such vector.
Rather they can take any rational values.
These theta series satisfy the following relations, which can be established by simple
changes of the summation variables:
\be
\vartheta^{\,\Lambda}_\mu=\vartheta^{\,\Lambda}_{-\mu},
\qquad
\vartheta^{\,\Lambda,2}_{\mu_1,\mu_2}=\vartheta^{\,\Lambda,2}_{\mu_2,\mu_1}=\vartheta^{\,\Lambda,2}_{-\mu_1,-\mu_2}
=\vartheta^{\,\Lambda,2}_{\mu_2-\mu_1,\mu_2}.
\label{identities-th}
\ee
Furthermore, making the change of variable $(k_1,k_2)\to (k-\ell,k+\ell)$, one diagonalizes the quadratic form in \eqref{defthetaL3}
which becomes $2k^2+6\ell^2$. However, this changes the determinant of the lattice by the factor of $4^{{\rm dim}( \Lambda)}$
indicating the necessity to introduce $2^{{\rm dim}( \Lambda)}$ glue vectors (see section \ref{subsec-lat}).
In other words, one can identify
\be
A_2\otimes\Lambda=\bcup\limits_{\mu'\in \Lambda/2\Lambda}\Bigl[ \(2\Lambda+\mu'\)\oplus \(2\Lambda+\mu'\)\Bigr],
\label{ident2lat}
\ee
which implies the following identity for theta series
\be
\vartheta^{\,\Lambda,2}_{\mu_1,\mu_2}(\tau)
=\sum_{\mu'\in \Lambda/2\Lambda}\vartheta^{\,\Lambda}_{\mu_2-\mu_1+\hf\mu'}(3\tau)\,\vartheta^{\,\Lambda}_{\mu_1+\mu_2+\hf\mu'}(\tau)
\label{ident2vth}
\ee
and a similar identity with exchange of $\mu_1$ and $\mu_2$.

Finally, we introduce two Jacobi theta series
\bea
\theta^{(2)}_\ell(\tau,z)&=&\sum_{k\in\IZ+\hf\ell}\q^{k^2} y^{2k} ,
\label{deftheta2z}
\\
\theta^{(3)}_{\ell}(\tau,z)&=&
\sum_{k_a\in\IZ+\frac{a}{3}\ell_a\ a\scriptscriptstyle{=1,2}}\q^{k_1^2+k_2^2-k_1k_2}\, y^{2(k_1+k_2)}.
\label{deftheta3z}
\eea
It is clear that they satisfy an analogue of the relation \eqref{ident2vth} which explicitly reads
\be
\theta^{(3)}_{\ell}(\tau,z)
=\sum_{\ell'=0,1}\theta^{(2)}_{\frac23\ell+\ell'}(3\tau,0)\,\theta^{(2)}_{\ell'}(\tau,2z).
\label{ident2theta}
\ee
This also implies a relation between the derivatives with respect to $z$ of these theta series.
Whereas their first derivatives evaluated at $z=0$ vanish, the second derivatives satisfy
\be
\p_z^2\theta^{(3)}_{\ell}(\tau,0)=
4\sum_{\ell'=0,1}\theta^{(2)}_{\frac23\ell+\ell'}(3\tau,0)\,\p_z^2\theta^{(2)}_{\ell'}(\tau,0).
\label{ident2thetader}
\ee

\subsection{Indefinite theta series: convergence}
\label{ap-converge}

Let us now turn to the case of the theta series with the quadratic form of indefinite signature.
The first problem here is to ensure their convergence that must be achieved by an appropriate choice of the kernel $\Phi(\bfx)$
which cannot be trivial anymore. The following theorem, which is a generalization of the ones presented in
\cite{Alexandrov:2017qhn,funke2017theta}, provides the simplest solution to this problem

\begin{theorem}
\label{th-conv}
Let the signature of the quadratic form be $(n,d-n)$ and
\be
\Phi(\bfx)=\prod_{i=1}^n\(\sgn(\bfv_{1,i}\bpt \bfx)-\sgn(\bfv_{2,i}\bpt\bfx)\).
\label{kerconverge}
\ee
Then the theta series \eqref{Vignerasth} is convergent provided:
\begin{enumerate}
\item
for all $i\in \Zv_{n}=\{1,\dots,n\}$,
$\bfv_{1,i}^2,\bfv_{2,i}^2> 0$;
\item
for any subset $\cI\subseteq \Zv_{n}$ and any set of $s_i\in \{1,2\}$, $i\in\cI$,
\be
\Delta_{\cI}(\{s_i\})\equiv \mathop{\det}\limits_{i,j\in \cI}(\bfv_{s_i,i}\bpt \bfv_{s_j,j})> 0;
\label{condDel}
\ee
\item
for all $\ell\in\Zv_n$ and any set of $s_i\in \{1,2\}$, $i\in\Zv_n\setminus\{\ell\}$,
\be
\bfv_{1,\ell\perp\{s_i\}}\bpt \bfv_{2,\ell\perp\{s_i\}}>0,
\label{condscpr}
\ee
where $_{\perp\{s_i\}}$ denotes the orthogonal projection on the subspace orthogonal to the span of
$\{\bfv_{s_i,i}\}_{i\in \Zv_n\setminus\{\ell\}}$.
\end{enumerate}
\label{th-converg}
\end{theorem}

These conditions can be further relaxed allowing for null vectors, in which case the quantities on the l.h.s. of
\eqref{condDel} and \eqref{condscpr} also may vanish. However, the null vectors should satisfy an additional condition
that (up to an irrelevant overall factor) they belong to the lattice, namely $\bfv_{s,i}\in \Lat$
if $\bfv_{s,i}^2=0$ \cite{Alexandrov:2017qhn}.
And even in that case the theta series diverges at the points where $\bfv_{s,i}\bpt (\bfk+\beta\bfptt)=0$ for some
${\bfk}\in \Lat+\bfmu+\hf\bfp$.

\subsection{Generalized error functions}
\label{ap-generr}

The series constructed in the previous subsection are holomorphic since $\sqrt{2\tau_2}$
entering the argument of the kernel in \eqref{Vignerasth} drops out from the sign functions, but
they are not modular because the discontinuities of the signs spoil the Vign\'eras equation.
Nevertheless, there is a simple recipe to construct their completion \cite{Zwegers-thesis,Alexandrov:2016enp,Nazaroglu:2016lmr}.

To this end, we have to define the generalized error functions
introduced in \cite{Alexandrov:2016enp,Nazaroglu:2016lmr} (see also \cite{kudla2016theta}):
\bea
E_n(\cM;\vu)&=& \int_{\IR^n} \de \vu' \, e^{-\pi(\vu-\vu')^{\rm tr}(\vu-\vu')} \sign(\cM^{\rm tr} \vu'),
\label{generr-E}
\eea
where $\vu=(u_1,\dots,u_n)$ is $n$-dimensional vector, $\cM$ is $n\times n$ matrix of parameters,
and we used the shorthand notation $\sign(\vu)=\prod_{i=1}^n \sign(u_i)$.
The detailed properties of these functions can be found in \cite{Nazaroglu:2016lmr}.
This is however not enough since, to define a kernel of theta series, we need a function depending on a $d$-dimensional vector.
Such functions, called boosted generalized error functions, are defined by
\be
\Phi_n^E(\cV;\bfx)=E_n(\cB\bpt \cV;\cB\bpt\, \bfx).
\label{generrPhiME}
\ee
Here $\cV$ is $d\times n$ matrix which can be viewed as a collection of $n$ vectors,
$\cV=(\bfv_1,\dots,\bfv_n)$, and it is assumed that these vectors span a positive definite subspace,
i.e. $\cV^{\rm tr}\bpt\cV$ is positive definite, whereas
$\cB$ is $n\times d$ matrix whose rows define an orthonormal basis for this subspace.
It can be shown that $\Phi_n^E$ does not depend on $\cB$ and
solves the Vign\'eras equation \eqref{Vigdif} with $\lambda=0$.
Furthermore, $\Phi_n^E(\{\bfv_i\};\bfx)$ at large $\bfx$ reduces to
$\prod_{i=1}^n \sgn (\bfv_i\bpt\,\bfx)$.
Thus, to construct a completion of the theta series whose kernel is a combination of sign functions,
it is sufficient to replace each product of $n$ sign functions by $\Phi_n^E$ with matrix of parameters $\cV$
given by the corresponding vectors $\bfv_i$.

It is important that if one of vectors is null, it reduces the rank of the generalized error function.
Namely, for $\bfv_\ell^2=0$, one has
\be
\Phi_n^E(\{\bfv_i\};\bfx)=\sgn (\bfv_\ell\,\bpt\bfx)\,\Phi_{n-1}^E(\{\bfv_i\}_{i\in \Zv_{n}\setminus\{\ell\}};\bfx).
\label{Phinull}
\ee
In other words, for such vectors the completion is not required.

\subsection{Generalized Appell functions}
\label{ap-Appell}

The generalized Appell functions have been introduced in \cite{Manschot:2014cca} and shown to capture the generating functions
of stack invariants for Hirzebruch surfaces in a particular chamber of the moduli space corresponding to $J=D_1$
in the notations of Appendix \ref{sec-Fm}. These functions appear to be special cases of indefinite theta series and therefore
generically transform as (higher depth) mock modular forms with completions constructed following the recipe of the previous subsection.

In this paper we need only three instances of these functions:
\begin{itemize}
\item
the classical Appell--Lerch function \cite{Appell:1886,Lerch}
\be
A(\tau,u,v):=e^{\pi\I u}\sum_{n\in\IZ}\frac{(-1)^n\, \q^{\hf n(n+1)}\, e^{2\pi\I n v}}{1-e^{2\pi\I u}\, \q^n}\, ,
\label{clAp}
\ee
\item
the level-$m$ Appell function \cite{Semikhatov:2003uc}
\be
A_m(\tau,u,v):=\sum_{n\in\IZ}\frac{\q^{\frac{m}{2} n^2}\, e^{2\pi\I m n v}}{1-e^{2\pi\I u}\, \q^n}\, ,
\label{Alevel}
\ee
\item
the level-2 Appell function for lattice $A_2$
\be
A^{(A_2)}_2(\tau, u_1,u_2,v_1,v_2)=\sum_{n_1,n_2\in\IZ}\frac{\q^{2(n_1^2+n_2^2-n_1n_2)}\, e^{2\pi\I v_1(2n_1-n_2)}\, e^{2\pi\I v_2(2n_2-n_1)}}
{\(1-e^{2\pi\I u_1}\, \q^{2n_1-n_2}\)\(1-e^{2\pi\I u_2}\, \q^{2n_2-n_1}\)}\, .
\label{Ap2}
\ee
\end{itemize}
The classical Appell--Lerch function is known to satisfy the following periodicity relation
\be
\theta_1(\tau,v)\,A(\tau,u+z,v+z)-\theta_1(\tau,v+z)\,A(\tau,u,v)
=\frac{\I(\eta(\tau))^3\,\theta_1(\tau,u+v+z)\,\theta_1(\tau,z)}{\theta_1(\tau,u)\,\theta_1(\tau,u+z)}\, .
\label{Appellperiod}
\ee

\setcounter{theorem}{0}

\section{Proof of Theorem 1}
\label{ap_theorem}

In this appendix we prove Theorem \ref{th-main} from $\S$\ref{subsec-anyN}, which we copy here for convenience:
\begin{theorem}
The normalized generating functions and their modular completions are expressed through the combinations \eqref{combTh}
\be
g_{N,\mu}=\Theta_{N,\mu}(\tau,z;\{\Phi_n\}),
\qquad
\whg_{N,\mu}(\tau,z)=\Theta_{N,\mu}(\tau,z;\{\whPhi_n\}),
\label{whgNc}
\ee
where the kernels are given by
\bea
\Phi_n(\{\gama_i\})&=&\(\sum_{\cJ\subseteq \cI} e_{|\cI\cap\cJ|}\prod_{k\in \cI\setminus\cJ}\Bigl(-\sgn\(\nv\cdot b_k\)\Bigr)\)
\prod_{k\in \Zv_{n-1}\setminus \cI}\Bigl(\sgn(\Gamma_k)-\sgn\(\nv\cdot b_k\)\Bigr),
\label{kergc}
\\
\whPhi_n(\{\gama_i\})&=&\sum_{\cJ\subseteq \Zv_{n-1}} \Phi_{|J|}^E(\{ \bfv_{\ell}\}_{\ell\in \cJ};\bfx)
\prod_{k\in \Zv_{n-1}\setminus \cJ}\Bigl(-\sgn\(\nv\cdot b_k\)\Bigr).
\label{kerhgc}
\eea
\end{theorem}

To prove this theorem, one needs to show three facts:

i) the kernels \eqref{kergc} and \eqref{kerhgc} define convergent theta series \eqref{newtheta};

ii) the substitution of $g_{N,\mu}$ into the formula \eqref{exp-whhr} for the completion is consistent
with the result for $\whg_{N,\mu}$ given here;

iii) $\whg_{N,\mu}$ transforms as a vector valued Jacobi form.

\subsection*{Convergence}

The kernel $\Phi_n$ determines the theta series \eqref{newtheta} which is an indefinite theta series with lattice
$\Lat=\Bigl[\mathop{\oplus}\limits_{i=1}^n N_i\Lambda_{S}\Bigr]/(N\Lambda_S)$ and quadratic
form $-Q_n$ \eqref{defQlr} of signature $(n-1,(n-1)(b_2-1))$.
For generic charges for which all $\Gamma_k$ are non-vanishing, it takes the simple form \eqref{kerg-no0}.
It coincides with the kernel \eqref{kerconverge} where the two sets of vectors are $\bfv_k$ \eqref{vectors}
and $\bfw_{k,k+1}$ \eqref{nullvec}. Hence, to prove the convergence, one has to check whether
these vectors fulfill the three conditions of Theorem \ref{th-conv}.
Note that since $\sum_{i=1}^n N_i \bfv_{\ell,i}^\alpha=\sum_{i=1}^n N_i \bfw_{k\ell,i}^\alpha=0$, the first term in
\eqref{defQlr} does not contribute to the scalar products of these vectors, and therefore
one can use the bilinear form \eqref{biform} for their evaluation.

The first condition of Theorem \ref{th-conv} holds since $\bfv_k^2=N\sN_k(N-\sN_k)c_1^2>0$, whereas $\bfw_{k,k+1}$ are null
and belong to $\Lat/(N_k N_{k+1})$ by assumptions about $\nv$.

To check the second condition, we evaluate
\be
\begin{split}
& \bfv_k\bpt\bfv_\ell=N\sN_k(N-\sN_\ell)c_1^2 \qquad \mbox{for } k<\ell,
\\
\bfv_k\bpt \bfw_{\ell,\ell+1}&=\delta_{k\ell} N N_k N_{k+1} \nv\cdot c_1,
\qquad
\bfw_{k,k+1}\bpt\bfw_{\ell,\ell+1}=0.
\end{split}
\label{scpr}
\ee
Since $\bfw_{k,k+1}$ are null and orthogonal to all other relevant vectors except $\bfv_k$,
the determinants of the Gram matrices involving them vanish.
Thus, it remains to check the positivity of
\be
\Delta_{\cI}\equiv \mathop{\det}\limits_{k,\ell\in \cI}(\bfv_k\bpt \bfv_\ell)
\label{condDelI}
\ee
for any subset $\cI\subseteq \Zv_{n-1}$. Denoting $m=|\cI|$ and ordering $k_1<\cdots <k_m$ the elements of the subset,
one finds
\bea
\Delta_{\cI}&=&(Nc_1^2)^{m}\left| \begin{array}{cccc}
\sN_{k_1}(N-\sN_{k_1}) & \sN_{k_1}(N-\sN_{k_2}) & \cdots & \sN_{k_1}(N-\sN_{k_m})
\\
\sN_{k_1}(N-\sN_{k_2}) & \sN_{k_2}(N-\sN_{k_2}) &   & \vdots
\\
\vdots &   & \ddots & \sN_{k_{m-1}}(N-\sN_{k_m})
\\
\sN_{k_1}(N-\sN_{k_m}) & \cdots & \sN_{k_{m-1}}(N-\sN_{k_m}) & \sN_{k_m}(N-\sN_{k_m})
\end{array}\right|
\nn\\
&=&
(Nc_1^2)^{m}\sN_{k_1}(N-\sN_{k_m})\left| \begin{array}{cccc}
(N-\sN_{k_1}) & (N-\sN_{k_2}) & \cdots & 1
\\
\sN_{k_1}(N-\sN_{k_2}) & \sN_{k_2}(N-\sN_{k_2}) &   & \vdots
\\
\vdots &   & \ddots & \sN_{k_{m-1}}
\\
\sN_{k_1}(N-\sN_{k_m}) & \cdots & \sN_{k_{m-1}}(N-\sN_{k_m}) & \sN_{k_m}
\end{array}\right|
\nn\\
&=&
(c_1^2)^m N^{2m-1}\sN_{k_1}(N-\sN_{k_m})\left| \begin{array}{ccccc}
0 & 0 &  \cdots & 1
\\
\sN_{k_1}-\sN_{k_2} &\quad 0 \quad&  & \vdots
\\
\vdots &  &  \ddots & \sN_{k_{m-1}}
\\
\sN_{k_1}-\sN_{k_m} & \cdots &  \sN_{k_{m-1}}-\sN_{k_m} & \sN_{k_m}
\end{array}\right|
\nn
\eea
\bea \hspace{-6cm}
&=&(c_1^2)^m N^{2m-1}\sN_{k_1}(N-\sN_{k_m})\prod_{i=1}^{m-1}(M_{k_{i+1}}-M_{k_i}),
\eea
which is indeed positive since $M_{k_{i+1}}>M_{k_i}$.

Finally, the third condition crucially simplifies due to the orthogonality properties
of $\bfw_{k,k+1}$, which ensure that the orthogonal projections in \eqref{condscpr} do not affect the scalar product.
Hence, it reduces to the one evaluated in \eqref{scpr} and is equal to $N N_\ell N_{\ell+1} \nv\cdot c_1$.
This is positive due to the condition \eqref{c0c1} on the null vector.

Once the convergence is shown for generic charges, the configurations with vanishing $\Gamma_k$ can also be taken into account.
Indeed, since $c_1$ is timelike, i.e. $c_1^2>0$, each condition $\Gamma_k=0$ fixes one of the timelike components
of the lattice vector of charges. Hence, after imposing $m$ such conditions, one remains with a lattice of signature
$(n-m-1,(n-1)(b_2-1))$. On the other hand, the second factor in the kernel \eqref{kergc} has exactly
$n-m-1$ factors and the form suitable for Theorem \ref{th-conv}. Since it is constructed from the same vectors as above,
the conditions of the theorem are again satisfied, which proves the convergence of $g_{N,\mu}$.

Finally, the convergence of the completion $\whg_{N,\mu}$ follows from the convergence
of the theta series $\vartheta^{(\vec N)}_{\mu,\vec\mu}(\tau,z;\Phi_n)$ with the kernel \eqref{kerg-no0}
and the properties of the generalized error functions \cite{Alexandrov:2016enp,Nazaroglu:2016lmr}
ensuring the convergence of any completion constructed by the recipe of Appendix \ref{ap-generr}.

\subsection*{Functional form}

Substituting the expression \eqref{whgNc} for $g_{N,\mu}$ into the formula for the completion \eqref{exp-whhr},
combining the sums over partitions, and taking into account that
\be
\sum_{i<j} \gamma_{ij}=c_1\cdot\sum\limits_{i=1}^n \ptt_i\, q_i,
\label{ident-gammaij}
\ee
it is easy to see that $\whg_{N,\mu}$ can be written as in \eqref{whgNc}
with the following kernels
\be
\whPhi_n(\{\gama_i\})=\sum_{n_1+\cdots n_m=n}\Rv_m(\{\gama'_i\},\tau_2,\beta)\,
\prod_{i=1}^m  \Phi_{n_i}(\gama_{j_{k-1}+1},\dots,\gama_{j_k}),
\ee
where $j_k=\sum_{i=1}^k n_i$ and the first factor $\Rv_m$ depends on $m$ charges defined by $\gama'_i=\sum_{i=j_{k-1}+1}^{j_k}\gama_i$.
Furthermore, cutting the tree $T$, appearing in the expression \eqref{solRnr} for $\Rv_m$, along the edges attached to the root vertex,
these kernels can be rewritten as
\be
\whPhi_n(\{\gama_i\})=\Phi_{n}(\{\gama_i\})+\sum_{n_1+\cdots n_m=n\atop m\ge 2}
\sEp_m(\{\gama'_i\},\tau_2,\beta)\,
\prod_{i=1}^m \Phi_{n_i}^{(0)}(\gama_{j_{k-1}+1},\dots,\gama_{j_k}),
\label{whg-fN}
\ee
where we introduced
\be
\Phi_n^{(0)}(\{\gama_i\})=\sum_{n_1+\cdots n_m=n}
\[\sum_{T\in\IT_m^{\rm S}}(-1)^{n_T}\prod_{v\in V_T}\sEf_{v}(\{\gama'_i\})\]
\prod_{i=1}^m\Phi_{n_i}(\gama_{j_{k-1}+1},\dots,\gama_{j_k}).
\label{deffN}
\ee
Then we have
\begin{lemma} For $\Phi_n$ given by \eqref{kergc}, it holds
\be
\Phi_n^{(0)}=\prod_{k=1}^{n-1} \Bigl(-\sgn\(\nv\cdot b_k\)\Bigr).
\label{expr-fTh}
\ee
\end{lemma}
\begin{proof}
We prove this Lemma by induction. For $n=2$ the definition \eqref{deffN} gives
\be
\Phi_2^{(0)}=\Phi_{2}-\sgn(\gamma_{12})=-\sgn\(\nv\cdot b_k\),
\ee
which indeed agrees with \eqref{expr-fTh}.
Let us now assume that the statement holds for all ranks up to $n-1$.
The key for the proof is the observation that the definition \eqref{deffN} implies
\be
\Phi_n^{(0)}=\Phi_{n}-\sum_{n_1+\cdots n_m=n\atop m\ge 2}
\sEf_{m}\,
\prod_{i=1}^m \Phi_{n_i}^{(0)}.
\label{recurs-fN}
\ee
Using the induction hypothesis, one can replace all $\Phi_{n_i}^{(0)}$ by \eqref{expr-fTh}.
Substituting also $\Phi_n$ from \eqref{kergc} and $\sEf_{m}$ from \eqref{rel-gEf-zero},
and rewriting the sum over splittings $n=n_1+\cdots n_m$ as the sum over subsets $\cJ\subseteq \Zv_{n-1}$, one obtains
\be
\begin{split}
\Phi_n^{(0)}=&\, \(\sum_{\cJ\subseteq \cI}e_{|\cJ|}\prod_{k\in \cI\setminus\cJ}\Bigl(-\sgn\(\nv\cdot b_k\)\Bigr)\)
\prod_{k\in \Zv_{n-1}\setminus \cI}\Bigl(\sgn(\Gamma_k)-\sgn\(\nv\cdot b_k\)\Bigr)
\\
&\, -\sum_{\emptyset\ne \cJ\subseteq \Zv_{n-1}} e_{|\cI\cap\cJ|}
\prod_{k\in \cJ\setminus \cI}\sgn(\Gamma_k)
\prod_{k\in \Zv_{n-1}\setminus \cJ}\Bigl(-\sgn\(\nv\cdot b_k\)\Bigr).
\end{split}
\label{tempker1}
\ee
It is straightforward to see that the second term would cancel the first if the empty set
contributed to the sum over $\cJ$. This implies that the kernel $\Phi_n^{(0)}$ coincides with the one given in
the statement of the Lemma \eqref{expr-fTh}.
\end{proof}

Using \eqref{expr-fTh} in \eqref{whg-fN} and substituting there the explicit expression for $\sEp_n$ and
\eqref{kergc} for $\Phi_n$, one obtains for $\whPhi_n$ the following expression
\bea
\whPhi_n&=& \(\sum_{\cJ\subseteq \cI}e_{|\cJ|}\prod_{k\in \cI\setminus\cJ}\Bigl(-\sgn\(\nv\cdot b_k\)\Bigr)\)
\prod_{k\in \Zv_{n-1}\setminus \cI}\Bigl(\sgn(\Gamma_k)-\sgn\(\nv\cdot b_k\)\Bigr)
\label{tempker2}\\
&&
+\sum_{\emptyset\ne\cJ\subseteq \Zv_{n-1}}\( \Phi_{|J|}^E(\{ \bfv_{\ell}\}_{\ell\in \cJ};\bfx)
-e_{|\cI\cap \cJ| }\prod_{k\in \cJ\setminus\cI }\sgn(\Gamma_k)\)
\prod_{k\in \Zv_{n-1}\setminus \cJ}\Bigl(-\sgn\(\nv\cdot b_k\)\Bigr).
\nn
\eea
Here the origin of the terms is similar to \eqref{tempker1}:
the first term is equal to $\Phi_n$ and the sum over subsets $\cJ$ in the second corresponds to
the sum over splittings $n=n_1+\cdots n_m$ in the second term in \eqref{whg-fN}.
As was already noticed, the first term plus the part of the second proportional to $e_{|\cI\cap \cJ| }$
combine to give $\Phi_n^{(0)}$, which is then can be included into the remaining term as
the contribution of $\cJ=\emptyset$.
As a result, one reproduces the kernel \eqref{kerhgc} in the statement of the theorem.

\subsection*{Modularity}

The final step of the proof is to show that $\whg_{N,\mu}$ defined by \eqref{whgNc} transforms as a vector valued Jacobi form
of weight $\frac12(N-1)b_2(S)$ and index $ -\frac16\, (N^3-N)K_S^2$.
Given the modularity of all $\phi_{N_i,\mu_i}$, it is clear
that this is true if the theta series $\vartheta^{(\vec N)}_{\mu,\vec\mu}(\tau,z;\whPhi_n)$ transform as
vector valued Jacobi forms of the following weight and index
\be
w_\vartheta(\vec N)=\frac{1}2\,(n-1)b_2(S),
\qquad
m_\vartheta(\vec N)= -\( \sum_{i<j<k}N_iN_jN_k +\hf\sum_{i\ne j}N_i^2N_j\)K_S^2
\label{wm-theta}
\ee
for all splittings $N=N_1+\cdots N_n$.

The easiest way to see that this is indeed the case is to contract $\vartheta^{(\vec N)}_{\mu,\vec\mu}$
with the Seigel-Narain theta series
\be
\theta^{\rm SN}_{N,\mu} (\tau)=\sum_{q\in N\Lambda_S+\mu+\frac{N}{2} K_S}\!\!
(-1)^{K_S\cdot q}\, \q^{-\tfrac{1}{2N}\(q^2-\frac{(c_1\cdot q)^2}{c_1^2}\)}
\bar\q^{-\frac{(c_1\cdot q)^2}{2N c_1^2}},
\ee
which is a (non-holomorphic) modular form of weight $(\hf (b_2-1), \hf)$ since it is equal to the theta series
$\vartheta_{\mu}(\tau;N K_S,0,\Phi^{\rm SN})$ \eqref{Vignerasth} for the lattice $\Lat=N\Lambda_S$ with the kernel
\be
\Phi^{\rm SN}(x)=e^{-\pi \, \frac{(c_1\cdot x)^2}{2N c_1^2}}
\ee
satisfying the Vign\'eras equation \eqref{Vigdif} with $\lambda=-1$.
Thus, we consider
\be
\theta^{(\vec N)}_{\vec\mu}(\tau,z)=
\sum_{\mu \in \Lambda_S/N\Lambda_S}\theta^{\rm SN}_{N,\mu} (\tau)\, \vartheta^{(\vec N)}_{\mu,\vec\mu}(\tau,z;\whPhi_n).
\ee
This contraction removes the condition on charges in the definition of
$\vartheta^{(\vec N)}_{\mu,\vec\mu}$ \eqref{newtheta} which crucially simplifies the resulting theta series.
As a result, it also belongs to the class of theta series \eqref{Vignerasth} with the following
data\footnote{Alternatively, one can take $\Lat=\mathop{\oplus}\limits_{i=1}^n N_i\Lambda_{S}$
with the inverse bilinear form and
the vectors obtained from those in \eqref{datathetaN} by contraction with the matrix $\diag(N_i) C_{\alpha\beta}$.}
\be
\begin{split}
\Lat &=\mathop{\oplus}\limits_{i=1}^n \Lambda_{S}\quad \mbox{ with bilinear form \eqref{biform}},
\\
\bfmu_i^\alpha= \tfrac{1}{N_i}\,C^{\alpha\beta}\mu_{i,\beta}, &
\qquad
\bfp_i^\alpha=-c_1^\alpha,
\qquad
\bfptt_i^\alpha=\ptt_i c_1^\alpha,
\qquad
\Phi=\Phi^{\rm SN}\whPhi_n.
\end{split}
\label{datathetaN}
\ee
Since all vectors $\bfv_k$, $\bfw_{k\ell}$ entering the definition of $\whPhi_n$
are orthogonal to any vector with components independent of the index $i$,
which is the case for the one providing the embedding of $\Phi^{\rm SN}$ into $\Lat$,
the action of the Vign\'eras operator on $\Phi$ factorizes.
Because both factors satisfy the Vign\'eras equation with $\lambda=-1$ and $\lambda=0$, respectively, $\Phi$ does so with $\lambda=-1$.
Therefore, $\theta^{(\vec N)}_{\vec\mu}$ is a vector valued Jacobi form of weight $(\frac{1}{2} (nb_2-1),\hf)$ and index
$-\hf\bfptt^2=m_\vartheta(\vec N)$. Subtracting the weight of the Seigel-Narain theta series, one finds
for $\vartheta^{(\vec N)}_{\mu,\vec\mu}$ precisely the required weight and index \eqref{wm-theta}.

Of course, such analysis can be performed directly for $\vartheta^{(\vec N)}_{\mu,\vec\mu}$, but
it is more complicated due to the condition on charges in the sum in \eqref{newtheta},
which results in that the relevant lattice is obtained from \eqref{datathetaN} by factoring out the diagonal $\Lambda_S$.
Nevertheless, one question is worth clarification: how does the reduction of the lattice by the condition
$\sum_{i=1}^n q_i=\mu+\frac{N}{2} K_S$
lead to the increase of the vector valued index $\bfmu$ to $(\mu,\vec \mu)$?

Substituting the decomposition of charges \eqref{quant-q}, the condition can be rewritten as
\be
\sum_{i=1}^n   N_i \, \epsilon_i^{\alpha}
=C^{\alpha\beta}\(\mu_\beta-\sum_{i=1}^n\mu_{i,\beta}\)\equiv \hmu^\alpha.
\label{cond-q}
\ee
Let $n_0=\gcd (\vec N)$. Then the constraint \eqref{cond-q} has solutions only for
$\mu_\beta=\sum_{i=1}^n\mu_{i,\beta} \mod n_0$,
in which case it can be written as
\be
\eps_i^\alpha=\sum_{a=0}^{n-1}\cM_i^a\veps_a^\alpha+\frac{1}{n_0}\,\cM_i^n\hmu^\alpha,
\label{soleps}
\ee
where $\veps_a^\alpha\in \IZ$ span the lattice
defining our theta series and the matrix $\cM\in SL(n,\IZ)$ satisfies
\be
\sum_{i=1}^n N_i\cM_i^a=0,
\qquad
\sum_{i=1}^n N_i\cM_i^n=n_0.
\label{matM}
\ee
Substituting \eqref{quant-q} and \eqref{soleps} into the quadratic form \eqref{defQlr}, one obtains
\be
-Q_n(\{\gama_i\})
=\sum_{i<j}\frac{N_iN_j}{N}\, (\eps_i-\eps_j)^2 +\cdots=\sum_{i,a,b} N_i \cM_i^a\cM_i^b \,\veps_a\cdot \veps_b +\cdots,
\label{defQeps}
\ee
where the dots denote terms linear and constant in $\veps_a$. We are interested in the determinant of the lattice which
is supposed to give the dimension of the representation where the theta series takes values.
It is given by the determinant of the matrix
$\cQ^{ab}=\sum_{i=1}^n N_i \cM_i^a\cM_i^b$ in the power $b_2$.
To evaluate it, we note that, due to $\det \cM=1$, one has
\be
\mathop{\det}\limits_{jk}\[ \sum_{i=1}^n N_i \cM_i^j\cM_i^k\]=\prod_{i=1}^n N_i.
\label{detfullM}
\ee
The determinant on the l.h.s. can be rewritten as
\bea
&&
\frac{1}{N}\left|\begin{array}{cc}
\cQ^{ab} & \sum\limits_{i,j}N_iN_j\cM_i^a\cM_i^n
\\
\sum\limits_{i} N_i \cM_i^n\cM_i^b\ & \sum\limits_{i,j}N_iN_j(\cM_i^n)^2
\end{array}\right|
=
\frac{1}{N}\left|\begin{array}{cc}
\cQ^{ab} & -\sum\limits_{j,c}\cQ^{ac}(\cM^{-1})_c^j
\\
\sum\limits_{i} N_i \cM_i^n\cM_i^b\ & n_0^2-\sum\limits_{i,j,c}N_i\cM_i^n\cM_i^c(\cM^{-1})_c^j
\end{array}\right|
\nn\\
&=& \frac{1}{N}\left|\begin{array}{cc}
\cQ^{ab} & 0
\\
\sum\limits_{i} N_i \cM_i^n\cM_i^b\ & n_0^2
\end{array}\right|=\frac{n_0^2}{N}\, \det \cQ^{ab},
\eea
where we used $\cM_i^n N_j =n_0\delta_i^j-\sum_c \cM_i^c(\cM^{-1})_c^j$ and \eqref{matM}.
Comparing with \eqref{detfullM}, we conclude that
\be
\det \cQ^{ab}=n_0^{-2} N\prod_{i=1}^n N_i.
\ee
Taken to the power $b_2$, this result should be compared with the number of components of the vector $(\mu,\vec \mu)$
equal to $\[ N\prod_{i=1}^n N_i\]^{b_2}$.
Thus, for $n_0=1$ the two dimensions coincide confirming that $\vartheta^{(\vec N)}_{\mu,\vec\mu}$
has the right modular properties.

But how can one understand the mismatch for $n_0>1$? In this case the dimension of the modular representation
is less than the number of components of the theta series. But not all of these components are independent!
We have already seen below \eqref{cond-q} that they vanish if $\mu_\beta\ne \sum_{i=1}^n\mu_{i,\beta} \mod n_0$.
Furthermore, it is easy to check that $\vartheta^{(\vec N)}_{\mu,\vec\mu}(\tau,z;\whPhi_n)$ is invariant under the shifts
$(\mu,\vec \mu)\to (\mu+ND_\alpha/n_0,\vec\mu+\vec N D_\alpha/n_0)$
where $D_\alpha\in\Lambda_S$, $\alpha=1,\dots,b_2$, is a basis of the lattice.
This two facts account for the missing factor $n_0^{2b_2}$
so that the number of independent components of the theta series precisely matches the determinant of the lattice.

However, this is not the end of the story as the above reasoning does not explain
why $\vartheta^{(\vec N)}_{\mu,\vec\mu}$ transforms as a vector.
This follows from Proposition \ref{prop-part} in Appendix \ref{ap-theta} where we should take
\be
\bfmu=(\mu/N,\{\mu_{i}/N_i\}),
\qquad
\bfxi_\alpha=(D_\alpha/n_0, \{D_\alpha/n_0\}),
\qquad
\bfzeta_\alpha=(-D^\star_\alpha/N,\{0\}),
\ee
where $D^\star_\alpha$ is a basis dual to $D_\alpha$, i.e. $D_\alpha\cdot D^\star_\beta=\delta_{\alpha\beta}$,
and $\bfrho$ spans
integer linear combinations of $\bfrho_\alpha^j=(D^\star_\alpha/N,\{\delta_{i}^j D^\star_\alpha/N_i\})$
and $\bfrho_\alpha^0=(n_0 D^\star_\alpha/N,\{0\})$
except $a^\alpha\bfxi_\alpha$, $a^\alpha=1,\dots, n_0-1$.
The conditions of the proposition are satisfied because, with respect to the bilinear form
\be
\bfx\bpt\bfy=\sum_{i=1}^n N_i \, x_i\cdot y_i-N x\cdot y,
\label{biform-factor}
\ee
which can be seen as the inverse of \eqref{defQlr},
one has
\be
\bfxi_\alpha\bpt\, \bfxi_\beta=\bfxi_\alpha\bpt\,\bfrho_\beta^j=0,
\qquad
\bfxi_\alpha\bpt\,\bfrho^0_\beta=-\delta_{\alpha\beta},
\qquad
\bfxi_\alpha\bpt\,\bfzeta_\beta=n_0^{-1}\delta_{\alpha\beta}.
\ee
As a result, the theta series $\vartheta^{(\vec N)}_{\mu,\vec\mu}(\tau,z;\whPhi_n)$ and hence the modular completion
$\whg_{N,\mu}(\tau,z)$ transform according to \eqref{Jacobi} with multiplier system
\eqref{eq:thetatransforms} and with the proper weight and index, as already shown above.

\section{Hirzebruch surfaces}
\label{sec-Fm}

The Hirzebruch surface $S=\Fb_m$ (also known as ruled rational surface) is defined as the projectivization of the
$\cO(m)\oplus \cO(0)$ bundle over $\IP^1$. It has $b_2(S)=2$ and $\chi(S)=4$.
In the basis
\be
D_1=[f], \qquad  D_2=[s]+m[f],
\ee
where $[f]$ and $[s]$ are the curves corresponding to the fiber and the section of the bundle,
the intersection matrix and the first Chern class are the following  \cite{Alim:2010cf,Alexandrov:2017mgi}
\be
\begin{split}
C_{\alpha\beta} =&\, \begin{pmatrix} 0 & 1 \\ 1 & m \end{pmatrix},
\qquad
C^{\alpha\beta} = \begin{pmatrix} -m & 1 \\ 1 & 0 \end{pmatrix},
\\
c_1^\alpha =&\,  ( 2-m , 2 ),
\qquad
c_1^2 =  8.
\end{split}
\label{dataFk}
\ee

Each lattice $\Lambda_{\Fb_m}$ has two null vectors.
Since their properties are slightly different for different $m$, we consider them one by one.
Our purpose here is to find a vector valued Jacobi form $\phi_{2,\mu}$ of weight $\hf b_2=1$
and index $-c_1^2=-8$, with the leading behavior near $z=0$ given by \eqref{resN2}.

\subsection*{$\Fb_0$}

In this case the null vectors are
\be
\nv^\alpha=(1,0),
\qquad
\nv'^\alpha=(0,1).
\label{nullF0}
\ee
Note that these vectors as well as the intersection matrix and the first Chern class
are symmetric under the exchange of the basis vectors.
Therefore, the construction does not depend on which null vector is chosen, up to this change of the basis.

In the construction of $\S$\ref{subsec-lat}, the new basis generating $\Latci{\Fb_0}$ is $\{\nv,\cv=(1,1)\}$, since $\gcd(c_1^\alpha)=2$,
and hence $r=\nv\cdot\cv=1$. Thus, glue vectors are not required and $\Latci{\Fb_0}=\Lambda_{\Fb_0}$.
In the new basis the residue class is $\mu=(\mu^1+\mu^2)\nv+\mu^2\cv$, and the condition \eqref{resN2} takes
a very simple form
\be
\phi_{2,\mu}\sim \delta^{(2)}_{\mu^2}\,\frac{1}{8\pi\I  z}\, ,
\label{resN2F0}
\ee
where we took into account that the function $\Delta(x)$ \eqref{FunDx} vanishes for $x\in\IZ$.
Then there is a natural function satisfying all the required properties (see \eqref{derth1})\footnote{Note that
multiplication of $z$ by $r$ changes the index of a Jacobi form by the factor $r^2$.}
\be
\phi_{2,\mu}(\tau,z)=\delta^{(2)}_{\mu^2}\,
\frac{\I\,\eta(\tau)^{3}}{\theta_1(\tau,4 z)}\, .
\label{phi2F0}
\ee
The only non-trivial thing to check is that it transforms as a modular vector.
This immediately follows from Proposition \ref{prop-part} in Appendix \ref{ap-theta}
where one can take $\bfxi=v_0$ and $\bfzeta=D_2$,
whereas the relevant lattice is $\Lat=2\Lambda_S$ which implies $\bfxi\bpt\,\bfzeta=1/2$.

As indicated above, the choice of the second null vector $\nv'$ leads to the same result with
$\mu^2$ replaced by $\mu^1$ in \eqref{phi2F0}.

\subsection*{$\Fb_1$}

The null vectors are
\be
\nv^\alpha=(1,0),
\qquad
\nv'^\alpha=(-1,2).
\label{nullF1}
\ee
They lead to different constructions because
\be
r=\cdv=2,
\qquad
r'=\cdv'=4
\ee
and hence require 2 and 4 glue vectors, respectively, which can be chosen as
\be
\begin{array}{l}
\vg_0=0,
\\
\vg_1=D_2=\hf(\cv-\nv),
\end{array}
\qquad
\begin{array}{l}
\vg'_0=0,
\\
\vg'_1=D_1=\hf (\cv-\nv'),
\\
\vg'_2=D_2=\frac14(\cv+\nv'),
\\
\vg'_3=D_1+D_2=\frac14(3\cv-\nv').
\end{array}
\ee

For the first choice one has $\mu=\(\mu^1-\hf\mu^2\)\nv+\hf\mu^2\cv$,
so that only the glue vector $\vg_0$ contributes to the condition \eqref{resN2} for $\mu^\alpha\in \{0,1\}$.
As a result, the condition becomes identical to the previous case \eqref{resN2F0} and has the same solution as in \eqref{phi2F0},
\be
\phi_{2,\mu}(\tau,z)=\delta^{(2)}_{\mu^2}\,
\frac{\I\,\eta(\tau)^{3}}{\theta_1(\tau,4 z)}\, .
\label{phi2F1}
\ee

The second choice gives
$\mu=\(\frac{1}{4}\mu^2-\frac{1}{2}\mu^1\)\nv'+\(\frac{1}{2}\mu^1+\frac{1}{4}\mu^2\)\cv$.
Substituting the components of $\vrh=\mu$ in the new basis into \eqref{resN2},
it is easy to see that the vectors $\vg'_1$ and $\vg'_2$ do not contribute,
whereas the contributions of $\vg'_0$ and $\vg'_3$ can be written as
\be
\phi_{2,\mu}\sim \delta^{(2)}_{\mu^2}\,\frac{1}{16\pi\I  z}\, ,
\label{resN2F2p}
\ee
which differs only by factor 1/2 from \eqref{resN2F0}. As a result, in this case one must take
\be
\phi_{2,\mu}(\tau,z)=\delta^{(2)}_{\mu^2}\,
\frac{\I\,\eta(\tau)^{3}}{2\,\theta_1(\tau,4 z)}\, .
\label{phi2F1p}
\ee

\subsection*{$\Fb_2$}

The null vectors read
\be
\nv^\alpha=(1,0),
\qquad
\nv'^\alpha=(-1,1).
\label{nullF2}
\ee
Since $\gcd(c_1^\alpha)=2$ and hence $\cv^\alpha=(0,1)$, both of them give $r=\nv\cdot\cv=1$ so that
no glue vectors are required.
Furthermore, changing the basis to
\be
D'_1=D_1,
\qquad
D'_2=D_2-D_1,
\label{chbF0F2}
\ee
one maps all the geometric data of $\Fb_2$ to those of $\Fb_0$.
In particular, the coordinates of the null vectors in the new basis become as in \eqref{nullF0}.
Hence, one can use the result \eqref{phi2F0} found for $\Fb_0$. Expressing it back in terms of the coordinates $\mu^\alpha$
defined by the original basis, one gets
\be
\phi_{2,\mu}(\tau,z)=\delta^{(2)}_{\mu^2}\,
\frac{\I\,\eta(\tau)^{3}}{\theta_1(\tau,4 z)}\, ,
\qquad
\phi_{2,\mu}(\tau,z)=\delta^{(2)}_{\mu^1+\mu^2}\,
\frac{\I\,\eta(\tau)^{3}}{\theta_1(\tau,4 z)}\, .
\label{phi2F2}
\ee
for the two null vectors \eqref{nullF2}, respectively.

\section{Del Pezzo surfaces}
\label{sec-Bm}

For the del Pezzo surface $S=\Bb_m$, defined as the
blow-up of $\IP^2$ over $m$ generic points,
one has $b_2(S)=m+1$, $\chi(S)=m+3$.
The most convenient basis is given by $D_1$ (the hyperplane class of $\IP^2$) and $D_2,\dots,D_{m+1}$,
(the exceptional divisors of the blow-up) with the diagonal intersection matrix $C_{\alpha\beta}= \diag (1,-1,\dots,-1)$.
In this basis the first Chern class is
\be
c_1^\alpha = (3, -1,\dots, -1),
\qquad
c_1^2 =  9-m.
\label{dataBk}
\ee
Note that $\Bb_0=\IP^2$ and $\Bb_1=\Fb_1$ so that below we assume that $m\ge 2$.

The lattice $\Lambda_{\Bb_m}$ has infinitely many non-equivalent null vectors.
A convenient way to classify them is to order them according to the number $r=\cdv=\nv\cdot c_1$.
Then for each $r$ there is only finitely many null vectors.
Furthermore, the vectors that can be mapped to each other by a permutation of the exceptional divisors $D_\alpha$, $\alpha\ge 2$,
can be identified due to the symmetry of the intersection matrix and the first Chern class.
Still this leaves us with several choices, even for low values of the parameter $r$.
Below we consider only three possible choices of such null vectors.

\subsection*{\be\hspace{-10cm}\mbox{Choice I: }\nv^\alpha=(1,-1,0,\dots,0)\label{nullBm1}\ee}

This choice corresponds to $r=2$. The orthogonal sublattice $\Lambda_{\Bb_m}^\perp$ can be chosen as follows
\be
\begin{array}{l}
d_1=D_1-D_2-2D_3,
\\
d_I=D_{I+1}-D_{I+2},\quad I\ge 2,
\end{array}
\qquad\qquad
C_{IJ}^\perp =  \begin{pmatrix}
-4 & 2 &  & & \\
2 & -2 & 1 & &   \\
 & 1 & -2 &  &  \\
 &  &  & \ddots & 1\\
 &  &  & 1 & -2
\end{pmatrix},
\label{basisBm1}
\ee
where $I=1,\dots, m-1$.
Note that after the change of the first basis element to $d_1+d_2$, the intersection matrix,
becomes identical to the (minus) $D_{m-1}$ Cartan matrix.\footnote{We preferred the basis \eqref{basisBm1}
because it works also for $m=2$.}
In the new basis the residue class is given by
\bea
\mu&=&\textstyle
\frac14\( (m-3)\mu^1+(m-7)\mu^2+2\sum\limits_{\alpha=3}^{m+1}\mu^\alpha\)\nv
+\hf\(\mu^1+\mu^2\)\cv
\label{muBm1}\\
&&\textstyle
-\frac14\( (m-1)(\mu^1+\mu^2)+2\sum\limits_{\alpha=3}^{m+1}\mu^\alpha\)d_1
-\hf\sum\limits_{I=2}^{m-1} \( (m-I)(\mu^1+\mu^2)+2\sum\limits_{\alpha=I+2}^{m+1}\mu^\alpha\)d_I.
\nn\eea
In our notations, the first line is equal to $\vrh$ and the second to $\rho$.

The determinant of $C_{IJ}^\perp$ is equal (up to sign) to 4. It follows that one has to introduce
$r_S=r\sqrt{\det C_{IJ}^\perp}=4$ glue vectors.
A possible choice is
\be
\begin{array}{l}
\vg_0=0,
\\
\vg_1=D_1=\frac14\((m-3)\nv+2\cv-(m-1)d_1-2\sum\limits_{I=2}^{m-1}(m-I)d_I\),
\\
\vg_2=D_3=\hf(\nv-d_1),
\\
\vg_3=D_1+D_3=\frac14\((m-1)\nv+2\cv-(m+1)d_1-2\sum\limits_{I=2}^{m-1}(m-I)d_I\).
\end{array}
\ee
As a result, the condition \eqref{resN2} becomes
\bea
\phi_{2,\mu}&\sim & \delta_{\mu^1=\mu^2=0}\[
\(\frac{1}{8\pi\I z}+\Delta(\haf\mu_{\rm tot})\)\vartheta^{\, \perp}_{\hf\rho}
+\(\frac{1}{8\pi\I z}+\Delta(\haf\mu_{\rm tot}+1)\)\vartheta^{\, \perp}_{\hf\rho+\vd_2}\]
\label{condphiBm1}\\
&&
+\delta_{\mu^1=\mu^2=1}\[
\(\frac{1}{8\pi\I z}+\Delta(\haf\mu_{\rm tot}+m)\)\vartheta^{\, \perp}_{\hf\rho+\vd_1}
+\(\frac{1}{8\pi\I z}+\Delta(\haf\mu_{\rm tot}+m+1)\)\vartheta^{\, \perp}_{\hf\rho+\vd_3}\],
\nn
\eea
where $\mu_{\rm tot}=\sum_{\alpha=3}^{m+1}\mu^\alpha$ and $\vd_a$ are projections of $\vg_a$ on $\Lambda_{\Bb_m}^\perp$.
Note that for $\mu_{\rm tot}$ even, all functions $\Delta$ in \eqref{condphiBm1} vanish since they are evaluated for integer arguments.
For $\mu_{\rm tot}$ odd, instead one gets $\Delta(\haf\mu_{\rm tot})=-\Delta(\haf\mu_{\rm tot}+1)$, and similarly for the second pair
of functions appearing in the second line of \eqref{condphiBm1}. Furthermore, in this case one has
\be
\vartheta^{\, \perp}_{\hf\rho}=\vartheta^{\, \perp}_{\hf\rho+\vd_2},
\qquad
\vartheta^{\, \perp}_{\hf\rho+\vd_1}=\vartheta^{\, \perp}_{\hf\rho+\vd_3}.
\label{identthetaBm1}
\ee
Indeed, $\mu_{\rm tot}$ odd implies $\hf\rho^I+\vd_a^I\in \haf\IZ$ for $I\ge 2$, whereas $\hf\rho^1+\vd_a^1\in \haf\IZ+\frac14$.
Moreover, $\vd_0^1-\vd_2^1=\vd_1^1-\vd_3^1=\hf$. Therefore, the relations in \eqref{identthetaBm1} are
particular cases of the first identity in \eqref{identities-th}.
As a result, all terms with $\Delta$ in \eqref{condphiBm1} cancel leaving
\be
\phi_{2,\mu}\sim  \frac{1}{8\pi\I z}\[\delta_{\mu^1=\mu^2=0}
\(\vartheta^{\, \perp}_{\hf\rho}+\vartheta^{\, \perp}_{\hf\rho+\vd_2}\)
+\delta_{\mu^1=\mu^2=1}\(\vartheta^{\, \perp}_{\hf\rho+\vd_1}+\vartheta^{\, \perp}_{\hf\rho+\vd_3}\)\].
\label{condphiBm12}
\ee

Let us now change the basis of $\Lambda_{\Bb_m}^\perp$ to $\td_I$ where
\be
\td_I=D_{I+2}|_{\Lambda_{\Bb_m}^\perp}\
\Longleftrightarrow\
\begin{array}{l}
d_1=2\td_1,
\\
d_I=\td_{I-1}-\td_{I},\quad I\ge 2.
\end{array}
\label{changebackBm1}
\ee
Since this is the same change of basis as in \eqref{basisBm1}, except for the absence of the null divisor $D_1-D_2$,
the intersection matrix of $\td_I$ is the same as (the restriction of) the intersection matrix of the original basis of $\Lambda_{\Bb_m}$,
i.e. $\td_I\cap\td_J=-\delta_{IJ}$.
Therefore, we conclude that
\be
\IZ^{m-1}=\Lambda_{\Bb_m}^\perp\cup\(\Lambda_{\Bb_m}^\perp+\haf\,d_1\),
\label{latZm}
\ee
which implies
\be
\vartheta^{\, \perp}_{\hf\rho}+\vartheta^{\, \perp}_{\hf\rho+\vd_2}
\ \stackrel{\mu^0=\mu^1=0}{=}\ \prod_{I=1}^{m-1}\theta_{\mu^{I+2}}^{(2)}(\tau,0)
\ \stackrel{\mu^0=\mu^1=1}{=}\
\vartheta^{\, \perp}_{\hf\rho+\vd_1}+\vartheta^{\, \perp}_{\hf\rho+\vd_3},
\label{relthetaBm1}
\ee
where $\theta_\ell^{(2)}(\tau,z)$ is defined in \eqref{deftheta2z}.
and we took into account that $\vd_1=\hf\sum_{I=1}^{m-1}\td_I$,
Thus, the two terms in \eqref{condphiBm12} can be nicely combined into
\be
\phi_{2,\mu}\sim \frac{\delta^{(2)}_{\mu\cdot \nv}}{8\pi\I z}\prod_{\alpha=3}^{m+1}\theta_{\mu^{\alpha}}^{(2)}(\tau,0).
\label{condphiBm13}
\ee
Taking into account the modularity restrictions, a general solution to this condition is given by
\be
\phi_{2,\mu}(\tau,z)=\delta^{(2)}_{\mu\cdot \nv}\,
\frac{\I\kappa\,\eta(\tau)^{3}}{\theta_1(\tau,4\kappa z)}\, \prod_{\alpha=3}^{m+1}\theta^{(2)}_{\mu^{\alpha}}(\tau,\kappa_{\alpha-2} z),
\label{phi2Bm1}
\ee
where the parameters $\kappa,\kappa_I$ are restricted to satisfy
\be
8\kappa^2-\sum_{I=1}^{m-1}\kappa_I^2=9-m.
\label{kappa}
\ee
As in Appendix \ref{sec-Fm}, $\phi_{2,\mu}$ transforms as a modular vector due to
Proposition \ref{prop-part} where one should take $\bfxi=v_0$, $\bfzeta=D_2$ and $\bfrho=\sum_{\alpha=3}^{m+1} \mu^\alpha D_\alpha$.

\subsection*{\be\hspace{-10cm}\mbox{Choice II: }\nv'^\alpha=(1,1,0,\dots,0)\label{nullBm2}\ee}

Although for this choice one has $r=4$, indicating that one has to use more glue vectors,
it is very similar to the previous one. The orthogonal sublattice $\Lambda_{\Bb_m}^\perp$ differs from
\eqref{basisBm1} only by a few coefficients
\be
\begin{array}{l}
d_1=D_1+D_2-4D_3,
\\
d_I=D_{I+1}-D_{I+2},\quad I\ge 2,
\end{array}
\qquad\qquad
C_{IJ}^\perp =  \begin{pmatrix}
-16 & 4 &  & & \\
4 & -2 & 1 & &   \\
 & 1 & -2 &  &  \\
 &  &  & \ddots & 1\\
 &  &  & 1 & -2
\end{pmatrix},
\label{basisBm2}
\ee
and the residue class in the new basis is given by
\bea
\mu&=&\textstyle
\frac{1}{16}\( (m+3)\mu^1-(m-13)\mu^2+4\sum\limits_{\alpha=3}^{m+1}\mu^\alpha\)\nv
+\frac14\(\mu^1-\mu^2\)\cv
\label{muBm2}\\
&&\textstyle
-\frac{1}{16}\( (m-1)(\mu^1-\mu^2)+4\sum\limits_{\alpha=3}^{m+1}\mu^\alpha\)d_1
-\frac14\sum\limits_{I=2}^{m-1} \( (m-I)(\mu^1-\mu^2)+4\sum\limits_{\alpha=I+2}^{m+1}\mu^\alpha\)d_I.
\nn\eea

The matrix \eqref{basisBm2} has $|\det C_{IJ}^\perp|=16$, so that one has to introduce $16$ glue vectors,
which we choose to be
\be
\begin{split}
\vg_{k,\ell} =&\,k D_1+\ell D_3
\\
=&\,\textstyle
\frac{k}{16}\((m+3)\nv+4\cv-(m-1)d_1-4\sum\limits_{I=2}^{m-1}(m-I)d_I\)
+\frac{\ell}{4}\,(\nv-d_1),
\end{split}
\ee
where $k,\ell=0,\dots, 3$.
With this choice, only the vectors $\vg_{0,\ell}$ contribute to the condition \eqref{resN2} which takes
the following form
\bea
\phi_{2,\mu}&\sim &\sum_{\ell=0}^3\[\delta_{\mu^1=\mu^2=0}
\(\frac{1}{16\pi\I z}+\Delta(-\rho^1+\haf\ell)\)\vartheta^{\, \perp}_{\hf\rho-\frac{\ell}{4}d_1}
\right.
\label{condphiBm2}\\
&&\left.
+\delta_{\mu^1=\mu^2=1}
\(\frac{1}{16\pi\I z}+\Delta(-\rho^1+\haf\ell+1)\)\vartheta^{\, \perp}_{\hf\rho-\frac{\ell}{4}d_1}\],
\nn
\eea
where we took into account that for the allowed values of $\mu^1$ and $\mu^2$, one has $ \tfrac14\mu_{\rm tot}=-\rho^1$.
Next, as in the previous case, one can show that all terms with $\Delta$ cancel. To this end, one
should consider 3 cases. If $\rho^1\in \IZ$, then $\Delta$ functions are non-vanishing only for $\ell=1$ and 3,
in which case they differ by sign. But the theta functions they multiply are equal,
$\vartheta^{\, \perp}_{\rho,-\frac{1}{4}d_1}=\vartheta^{\, \perp}_{\rho,-\frac{3}{4}d_1}$,
by the same reason as in \eqref{identthetaBm1}.
Similarly, for $\rho^1\in \IZ+\hf$, only $\ell=0$ and 2 give rise to non-vanishing $\Delta$ functions which cancel each other
due to equality of the corresponding theta functions.
Finally, for $\rho^1\in \IZ+\frac14$ or $\rho^1\in \IZ+\frac34$, all $\ell$ contribute, but they can be split in pairs
so that $-\rho^1+\haf\ell=\pm\frac14\mbox{\ \small mod 2}$ or $\pm\frac34\mbox{\ \small mod 2}$.\footnote{The shift by 1
in the second line in \eqref{condphiBm2} simply exchanges these pairs.}
Then for each pair the $\Delta$ functions differ by sign and the theta functions coincide so that all terms cancel.
As a result, one remains with
\be
\phi_{2,\mu}\sim  \frac{\delta^{(2)}_{\mu\cdot \nv}}{16\pi\I z}\,\sum_{\ell=0}^3\vartheta^{\, \perp}_{\hf\rho-\frac{\ell}{4}d_1}.
\label{condphiBm22}
\ee

Similarly to \eqref{changebackBm1}, we now change the basis of $\Lambda_{\Bb_m}^\perp$ to $\td_I=D_{I+2}|_{\Lambda_{\Bb_m}^\perp}$
where the intersection matrix is again diagonal, $\td_I\cap\td_J=-\delta_{IJ}$.
This shows that
\be
\IZ^{m-1}=\bcup\limits_{\ell=0}^3\(\Lambda_{\Bb_m}^\perp+\ell d_1\),
\ee
which implies that the sum of theta functions in \eqref{condphiBm22} produces a theta function over $\IZ^{m-1}$.
As a result, one obtains
\be
\phi_{2,\mu}\sim \frac{\delta^{(2)}_{\mu\cdot \nv}}{16\pi\I z}\prod_{\alpha=3}^{m+1}\theta_{\mu^{\alpha}}^{(2)}(\tau,0),
\label{condphiBm23}
\ee
with the solution given by
\be
\phi_{2,\mu}(\tau,z)=\delta^{(2)}_{\mu\cdot \nv}\,
\frac{\I\kappa\,\eta(\tau)^{3}}{2\,\theta_1(\tau,4 \kappa z)}\, \prod_{\alpha=3}^{m+1}\theta^{(2)}_{\mu^{\alpha}}(\tau,\kappa_{\alpha-2}z).
\label{phi2Bm2}
\ee
Similarly to the two choices of the null vectors for $\Fb_1$, the two results \eqref{phi2Bm1} and \eqref{phi2Bm2}
differ only by factor 1/2. This is not surprising as $\Fb_1=\Bb_1$.

\subsection*{\be\hspace{-7.3cm}\mbox{Choice III: }\nv^\alpha=(2,-1,-1,-1,-1,0,\dots,0)\label{nullBm3}\ee}

This choice exists only for $m\ge4$ and corresponds to $r=2$. For simplicity we consider it only for $m=4$.
The analysis for $m>4$ is very similar to the one presented below and leads to the same results.

The basis for the orthogonal sublattice $\Lambda_{\Bb_4}^\perp$ and the associated intersection matrix are given by
\be
d_I=D_{I+1}-D_{I+2},
\qquad
C^\perp_{IJ} =  \begin{pmatrix}
-2 & 1 &   \\
1 & -2 & 1 \\
  & 1 & -2
\end{pmatrix},
\ee
where $I=1,2,3$. The residue class in the new basis is found to be
\bea
\mu&=&\textstyle
-\( \mu^1+\tfrac34\sum\limits_{\alpha=2}^5 \mu^\alpha\)\nv
+\(\mu^1+\haf\sum\limits_{\alpha=2}^5 \mu^\alpha\)\cv
\label{muBm3}\\
&&\textstyle
+\frac{1}{4}\bigl( 3\mu^2-\mu^3-\mu^4-\mu^5\bigr)d_1
+\frac{1}{2}\bigl( \mu^2+\mu^3-\mu^4-\mu^5\bigr)d_2
+\frac{1}{4}\bigl( \mu^2+\mu^3+\mu^4-3\mu^5\bigr)d_1.
\nn\eea
One has $|\det C_{IJ}^\perp|=4$, so that we have to provide 4 glue vectors which we choose as
\be
\vg_k=kD_2=\textstyle
\frac{k}{4}\bigl(-3\nv+2\cv+3d_1+2d_2+d_3\bigr),
\ee
where $k=0,\dots, 3$.
As a result, the condition \eqref{resN2} can be rewritten in the following explicit form
\bea
\phi_{2,\mu}&\sim &
\frac{\delta_{\mu^\alpha=0}}{8\pi\I z}\(\vartheta^{\, \perp}_{\hf\rho}
+\vartheta^{\, \perp}_{\hf\rho+2\vd_1}\)
+\delta_{\mu^\alpha=1}\[\frac{1}{8\pi\I z}\(\vartheta^{\, \perp}_{\hf\rho+\vd_1}+\vartheta^{\, \perp}_{\hf\rho+3\vd_1}\)
-\frac14\(\vartheta^{\, \perp}_{\hf\rho+\vd_1}-\vartheta^{\, \perp}_{\hf\rho+3\vd_1}\)
\]
\nn\\
&&
+\delta_{\mu^1+{\haf}\sum\limits_{\alpha=2}^5 \mu^\alpha=1}
\[\(\frac{1}{8\pi\I z}+\Delta(\haf\mu^1-1)\)\vartheta^{\, \perp}_{\hf\rho+\vd_1}
+\(\frac{1}{8\pi\I z}+\Delta(\haf\mu^1)\)\vartheta^{\, \perp}_{\hf\rho+3\vd_1}\]
\\
&&
+\delta_{\mu^1+{\haf}\sum\limits_{\alpha=2}^5 \mu^\alpha=2}
\[\(\frac{1}{8\pi\I z}+\Delta(\haf\mu^1-1)\)\vartheta^{\, \perp}_{\hf\rho}
+\(\frac{1}{8\pi\I z}+\Delta(\haf\mu^1)\)\vartheta^{\, \perp}_{\hf\rho+2\vd_1}\].
\nn
\label{condphiBm3}
\eea
As in all previous cases, one can show that all terms with $\Delta$ function cancel, which follows from the identities
$\vartheta^{\, \perp}_{\hf\rho}=\vartheta^{\, \perp}_{\hf\rho+2\vd_1}$ valid for $\sum\limits_{\alpha=2}^5 \mu^\alpha=2$
and $\vartheta^{\, \perp}_{\hf\rho+\vd_1}=\vartheta^{\, \perp}_{\hf\rho+3\vd_1}$
valid for all $\mu^\alpha$ vanishing or all equal to 1.
All these identities are, as usual, particular cases of \eqref{identities-th}.
Thus, one remains with
\be
\phi_{2,\mu}\sim  \frac{1}{8\pi\I z}\[\delta^{(2)}_{\mu^1+{\haf}\sum\limits_{\alpha=2}^5 \mu^\alpha}
\(\vartheta^{\, \perp}_{\hf\rho}+\vartheta^{\, \perp}_{\hf\rho+2\vd_1}\)
+\delta^{(2)}_{\mu^1+{\haf}\sum\limits_{\alpha=2}^5 \mu^\alpha+1}
\(\vartheta^{\, \perp}_{\hf\rho+\vd_1}+\vartheta^{\, \perp}_{\hf\rho+3\vd_1}\)\].
\label{condphiBm32}
\ee

Next, we change the basis of $\Lambda_{\Bb_4}^\perp$ to
\be
\td_1=\hf\, (d_1+d_3),
\qquad
\td_2=\hf\, (d_1-d_3),
\qquad
\td_3=d_2+\hf\, (d_1+d_3)
\ee
with the diagonal intersection matrix, $\td_I\cap\td_J=-\delta_{IJ}$.
Taking into account that $\vd_1=\hf\sum_{I=1}^{m-1}\td_I$, one has
\be
\IZ^{m-1}=\Lambda_{\Bb_m}^\perp\cup\(\Lambda_{\Bb_m}^\perp+2\vd_1\),
\label{latZm3}
\ee
which implies that the sums of theta functions appearing in \eqref{condphiBm32} can be expressed through
a theta function over $\IZ^{m-1}$. More specifically, using
\be
\rho=\hf\bigl(\mu^2-\mu^3+\mu^4-\mu^5\bigr)\td_1+
\hf\bigl(\mu^2-\mu^3-\mu^4+\mu^5\bigr)\td_2
+\hf\bigl(\mu^2+\mu^3-\mu^4-\mu^5\bigr)\td_3,
\ee
one finds
\be
\phi_{2,\mu}\sim \frac{\delta^{(2)}_{\mu\cdot \nv}}{8\pi\I z}
\prod_{\alpha=3}^{5}\theta_{\mu^1+\mu^2+\mu^\alpha}^{(2)}(\tau,0).
\label{condphiBm233}
\ee
The corresponding solution can be written as
\be
\phi_{2,\mu}(\tau,z)=\delta^{(2)}_{\mu\cdot \nv}\,
\frac{\I\kappa\,\eta(\tau)^{3}}{\theta_1(\tau,4\kappa z)}\, \prod_{I=1}^{3}\theta^{(2)}_{\mu\cdot E_I }(\tau,\kappa_I z),
\label{phi2Bm3}
\ee
where we introduced a set of lattice vectors
$E_I=D_1-D_2-D_{I+2}$.
They can be seen as an orthonormal basis in the negative definite sublattice orthogonal to the two null vectors,
$\nv$ and $D_1-D_2$.
To verify that $\phi_{2,\mu}$ transforms as a modular vector, one should again apply
Proposition \ref{prop-part} where one takes $\bfxi=v_0$, $\bfzeta=D_1-D_2$ and $\bfrho$ lies in the lattice
generated by $E_I$.

\section{Details of calculations for $N=3$}
\label{ap-detail3}

In this appendix we provide details omitted in $\S$\ref{subsubsec-phiN3}.

\subsection*{Derivation of \eqref{theta21one}}
\label{ap-combineterms}

The three contributions coming from \eqref{theta21} with one of $\nb_{k\ell}$ vanishing read
\bea
&&
\frac14\sum_{k_1\in \Lambda_S+\frac13\mu}\sum_{k_2\in \Lambda_S+\frac23\mu}
\biggl[
\delta_{\nb_{12}=0}\Bigl(\sgn(c_1\cdot k_1)-\sgn(\beta)\Bigr)
\Bigl(\sgn(c_1\cdot k_2)-\sgn(\nv\cdot k_2)\Bigr)
\biggr.
\nn\\
&&\
+\delta_{\nb_{23}=0}\Bigl(\sgn(c_1\cdot k_1)-\sgn(\nv\cdot k_1)\Bigr)
\Bigl(\sgn(c_1\cdot k_2)-\sgn(\beta)\Bigr)
\label{thetazeroN31}\\
&&\biggl.\
+\delta_{\nb_{13}=0}\Bigl(\sgn(c_1\cdot k_1)-\sgn(\nv\cdot k_1+\beta)\Bigr)
\Bigl(\sgn(c_1\cdot k_2)-\sgn(\nv\cdot k_2+\beta)\Bigr)\biggr]\, \q^{-k_1^2-k_2^2+k_1k_2}\, y^{2c_1\cdot (k_1+k_2)},
\nn
\eea
where we took into account that for $\beta\ll 0$ the coefficient of $\beta$ in the sign
functions with $\nb_{k\ell}\ne 0$ does not play any role.
Let us rewrite the first term in this expression as a sum of two contributions
\bea
&&\frac14\sum_{k_1\in \Lambda_S+\frac13\mu}\sum_{k_2\in \Lambda_S+\frac23\mu}
\q^{-k_1^2-k_2^2+k_1k_2}
\nn\\
&&\quad \times
\biggl[
\delta_{\nb_{12}=0}\Bigl(\sgn(c_1\cdot (2k_1-k_2))-\sgn(\beta)\Bigr)
\Bigl(\sgn(c_1\cdot k_2)-\sgn(\nv\cdot k_2)\Bigr) \, y^{2c_1\cdot (k_1+k_2)}
\biggr.
\label{term1}\\
&&\biggl. \quad
+\delta_{\nb_{13}=0}\Bigl(\sgn(c_1\cdot k_1)-\sgn(c_1\cdot (k_1+k_2))\Bigr)
\Bigl(\sgn(c_1\cdot (k_1-k_2))+\sgn(\nv\cdot k_2)\Bigr)
\, y^{2c_1\cdot (2k_1-k_2)} \biggr],
\nn
\eea
where the second contribution is obtained by changing the summation variables: $(k_1,k_2)\to(k_1,k_1-k_2)$.
Clearly, the second term in \eqref{thetazeroN31}
can be rewritten in the same way with $k_1$ and $k_2$ exchanged comparing to \eqref{term1},
\bea
&&\frac14\sum_{k_1\in \Lambda_S+\frac13\mu}\sum_{k_2\in \Lambda_S+\frac23\mu}
\q^{-k_1^2-k_2^2+k_1k_2}
\nn\\
&&\quad \times
\biggl[
\delta_{\nb_{23}=0}\Bigl(\sgn(c_1\cdot (2k_2-k_1))-\sgn(\beta)\Bigr)
\Bigl(\sgn(c_1\cdot k_1)-\sgn(\nv\cdot k_1)\Bigr) \, y^{2c_1\cdot (k_1+k_2)}
\biggr.
\label{term2}\\
&&\biggl. \quad
-\delta_{\nb_{13}}\Bigl(\sgn(c_1\cdot k_2)-\sgn(c_1\cdot (k_1+k_2))\Bigr)
\Bigl(\sgn(c_1\cdot (k_1-k_2))-\sgn(\nv\cdot k_1)\Bigr)
\, y^{2c_1\cdot (2k_2-k_1)} \biggr].
\nn
\eea
Finally, the third term in \eqref{thetazeroN31} is dealt with the help of the following sign identity
\be
\sgn(x_1)\,\sgn(x_2)=1-\sgn(x_1-x_2)\bigl( \sgn(x_1)-\sgn(x_2)\bigr).
\label{signident}
\ee
Applying it, one finds
\be
\begin{split}
&\,
-\frac14\sum_{k_1\in \Lambda_S+\frac13\mu}\sum_{k_2\in \Lambda_S+\frac23\mu}
\delta_{\nb_{13}=0}\,\Bigl[\sgn(c_1\cdot(k_1-k_2))\bigl(\sgn(c_1\cdot k_1)-\sgn(c_1\cdot k_2)\bigr)
\\
&\,\qquad
+\sgn(c_1\cdot k_1)\,\sgn(\nv\cdot k_2)+\sgn(c_1\cdot k_2)\,\sgn(\nv\cdot k_1)
\Bigr]\,\q^{-k_1^2-k_2^2+k_1k_2}\, y^{2c_1\cdot (k_1+k_2)}.
\end{split}
\ee
The first contributions in \eqref{term1} and \eqref{term2} are exactly the last two in \eqref{theta21one},
whereas the remaining contributions, replacing in half of them $(k_1,k_2)\to(-k_2,-k_1)$, are combined into
the first contribution in \eqref{theta21one}.

\subsection*{Condition on $\phi_{3,\mu}$}

Our goal is to expand the two contributions, \eqref{cvthetaN3} and \eqref{N3termphi}, around $z=0$ up to terms $O(z)$.
For Hirzebruch surfaces it is easy to check that these two expansions are given by
\be
\begin{split}
\eqref{cvthetaN3}\sim &\,  -\delta^{(3)}_{\nv\cdot\mu}\[\frac{1}{(4\pi\cdv z)^2}+\frac{1}{12}-\frac13\, \delta^{(3)}_{c_1\cdot\mu}\]
+O(z^2),
\\
\eqref{N3termphi}\sim &\,  \delta^{(3)}_{\nv\cdot\mu}\[\frac{1-(4\kappa z)^2\cD(\tau)}{3(2\pi\cdv z)^2}+\frac{1}{12}
-\frac13\, \delta^{(3)}_{c_1\cdot\mu}\]
+O(z^2).
\end{split}
\ee
They are nicely combined into
\be
\eqref{cvthetaN3}+\eqref{N3termphi}\sim  \delta^{(3)}_{\nv\cdot\mu}\frac{1-(8\kappa z)^2\cD(\tau)}{3(4\pi\cdv z)^2}
\ee
reproducing (minus) the r.h.s. of \eqref{expan-phi3}.

For del Pezzo surfaces the presence of theta series significantly complicates the analysis.
Nevertheless, using the identities \eqref{identities-th}, it is not difficult to show that all terms of order
$\sim \!z^{-1}$ and $\sim\! z$ cancel,
analogously to the $\Delta$-dependent terms in Appendix \ref{sec-Bm}.
The leading terms can also be easily evaluated. The point is that the sum over glue vectors, from which only $\cdv$
terms survive, allows to pass from the lattice $\Lambda_{\Bb_m}^\perp$ to $\IZ^{m-1}$ as in \eqref{latZm}.
As a result, the leading terms in the two expansions read
\be
\begin{split}
\eqref{cvthetaN3}\sim &\,  -\frac{\delta^{(3)}_{\nv\cdot\mu}}{(4\pi\cdv z)^2}\prod_{I=1}^{m-1}\theta^{(3)}_{\mu\cdot E_I }(\tau,0)
+O(1),
\\
\eqref{N3termphi}\sim &\,  \frac{\delta^{(3)}_{\nv\cdot\mu}}{3(2\pi\cdv z)^2}\prod_{I=1}^{m-1}\[\,\sum_{\ell_I=0,1}
\theta^{(2)}_{\ell_I }(\tau,0)\,\theta^{(2)}_{\frac23 \mu\cdot E_I+\ell_I }(3\tau,0)\]
+O(1).
\end{split}
\ee
Applying the identity \eqref{ident2theta}, one expresses the second expansion in terms of $\theta^{(3)}_{\mu\cdot E_I}$
so that the two contributions can be combined and reproduce (minus) the leading term in \eqref{expan-phi3}.

The most complicated is to show that the only constant terms which survive are those that are obtained from the expansion of
$\theta_1(\tau,\kappa z)$ and $\theta^{(2)}_{\mu\cdot E_I}(\tau, \kappa_I z)$ in $\phi_{2,\mu}$ \eqref{genphi2}.
Unfortunately, we did not find a general pattern for the constant terms appearing in the expansion of
\eqref{cvthetaN3} and \eqref{N3termphi} before they are combined together. Therefore, one has to analyze these terms case by case.
To illustrate the procedure, we present here the case of the null vector \eqref{nullBm1} for which one finds
\bea
O(1)[\eqref{cvthetaN3}]&= &  -\delta^{(3)}_{\nv\cdot\mu}\(\frac{1}{12}
-\frac13\, \delta^{(3)}_{c_1\cdot\mu}\)\vartheta^{\, \perp,2}_{\frac13\trho,\frac23\trho}(\tau),
\label{O1term1}\\
O(1)[\eqref{N3termphi}]&=&   \delta^{(3)}_{\nv\cdot\mu}\sum_{\ell_I=0,1\atop I=1,\dots m-1}
\Biggl[
\frac{1}{3(2\pi\cdv z)^2} \(\sum\limits_{I=1}^{b_2-2}\frac{\kappa_I^2}{2}\,\p_z^2 \log\theta^{(2)}_{\ell_I}(\tau,0)
-16\kappa^2\cD(\tau)\)
\Biggr.
\nn\\
&+& \Biggl.
\frac{1+3(-1)^{\ell_{\rm tot}}}{4}\(\frac{1}{12}-\frac13\, \delta^{(3)}_{c_1\cdot\mu}\)\Biggr]
\prod_{I=1}^{m-1}\(\theta^{(2)}_{\ell_I }(\tau,0)\,\theta^{(2)}_{\frac23\mu\cdot E_I+\ell_I }(3\tau,0) \) ,
\label{O1term2}
\eea
where $\ell_{\rm tot}=\sum_{I=1}^{m-1}\! \ell_I$ and $\trho=\rho|_{\mu^1+\mu^2=0}$.
Applying the identities \eqref{ident2theta} and \eqref{ident2thetader}, it is immediate to see that
the first term in the expansion \eqref{O1term2} reproduces (minus) the constant term in \eqref{expan-phi3}.
To show the cancelation of the remaining terms, let us decompose the lattice $\IZ^{m-1}$ as in \eqref{latZm}
which allows to write
\be
\begin{split}
\prod_{I=1}^{m-1}\theta^{(2)}_{\ell_I }(\tau,0)
=&\, \vartheta^{\, \perp}_{\hf\trho(\ell)}(\tau)
+\vartheta^{\, \perp}_{\hf\trho(\ell)+\hf d_1}(\tau),
\\
\prod_{I=1}^{m-1}\theta^{(2)}_{\frac23\mu\cdot E_I+\ell_I }(3\tau,0)
=&\, \vartheta^{\, \perp}_{\frac13\trho+\hf\trho(\ell)}(3\tau)
+\vartheta^{\, \perp}_{\frac13\trho+\hf\trho(\ell)+\hf d_1}(3\tau),
\end{split}
\ee
where (c.f. \eqref{muBm1})
\be
\trho(\ell)=
-\frac12\(\sum\limits_{I=1}^{m-1}\ell_I\)d_1
-\sum\limits_{I=2}^{m-1} \( \sum\limits_{J=I}^{m-1}\ell_J\)d_I.
\ee
Note that for $\ell_{\rm tot}$ even, all components of $\trho(\ell)$ are integer, whereas
if $\ell_{\rm tot}$ is odd, the component $\trho^1(\ell)$ is half-integer. Due to this, one has
\be
\begin{split}
&
\sum_{\ell_I=0,1\atop I=1,\dots m-1}\frac{1+3(-1)^{\ell_{\rm tot}}}{4}
\,\prod_{I=1}^{m-1}\,\(\theta^{(2)}_{\ell_I }(\tau,0)\,\theta^{(2)}_{\frac23\mu\cdot E_I+\ell_I }(3\tau,0)\)
\\
=&\, \sum_{\ell_I=0,1, \ \ell_{\rm tot}\in 2\IZ\atop I=1,\dots m-1}
\(\vartheta^{\, \perp}_{\hf\trho(\ell)}(\tau)
+\vartheta^{\, \perp}_{\hf\trho(\ell)+\hf d_1}(\tau)\)
\( \vartheta^{\, \perp}_{\frac13\trho+\hf\trho(\ell)}(3\tau)
+\vartheta^{\, \perp}_{\frac13\trho+\hf\trho(\ell)+\hf d_1}(3\tau)\)
\\
&\,
-\hf\sum_{\ell_I=0,1, \ \ell_{\rm tot}\in 2\IZ+1\atop I=1,\dots m-1}
\(\vartheta^{\, \perp}_{\hf\trho(\ell)}(\tau)
+\vartheta^{\, \perp}_{\hf\trho(\ell)+\hf d_1}(\tau)\)
\( \vartheta^{\, \perp}_{\frac13\trho+\hf\trho(\ell)}(3\tau)
+\vartheta^{\, \perp}_{\frac13\trho+\hf\trho(\ell)+\hf d_1}(3\tau)\)
\\
=&\,
\sum_{k^I=0,1 \atop I=1,\dots m-1}
\(\vartheta^{\, \perp}_{\hf\sum_I k^I d_I}(\tau)\, \vartheta^{\, \perp}_{\frac13\trho+\hf\sum_I k^I d_I}(3\tau)
+\vartheta^{\, \perp}_{\hf\sum_I k^I d_I+\hf d_1}(\tau)\, \vartheta^{\, \perp}_{\frac13\trho+\hf\sum_I k^I d_I}(3\tau)
\right.
\\
&\, \left. \qquad\qquad
-\vartheta^{\, \perp}_{\frac14+\hf\sum_I k^I d_I}(\tau)\vartheta^{\, \perp}_{\frac13\trho+\frac14+\hf\sum_I k^I d_I}(3\tau)
\).
\end{split}
\label{O1term22}
\ee
The sum over all $k^I$ is equivalent to the sum over $\Lambda_{\Bb_m}^\perp/2\Lambda_{\Bb_m}^\perp$,
which allows to apply the identity \eqref{ident2vth}. Thus, the combination of theta series \eqref{O1term22}
reduces to\footnote{One should distinguish the two case whether $\mu_{\rm tot}=\sum_{\alpha=3}^{m+1}\mu^\alpha$ is even or odd.
In the former case the first (resp. second) term in \eqref{O1term22} gives rise to the first (resp. second) term in
\eqref{O1term23}, and in the latter case the identification is vice versa.}
\be
\vartheta^{\, \perp,2}_{\frac13\trho,\frac23\trho}(\tau)
+\vartheta^{\, \perp,2}_{\frac13\trho+\frac14 d_1,\frac23\trho+\frac14 d_1}(\tau)
-\vartheta^{\, \perp,2}_{\frac13\trho,\frac23\trho+\frac14 d_1}(\tau)
=\vartheta^{\, \perp,2}_{\frac13\trho,\frac23\trho}(\tau),
\label{O1term23}
\ee
where at the second step we applied the identity \eqref{identities-th}.
As a result, one obtains the contribution cancelling the one in \eqref{O1term1}.
This finishes the proof of \eqref{expan-phi3} for $S=\Bb_m$ and the null vector $\nv$ chosen as in \eqref{nullBm1}.

\providecommand{\href}[2]{#2}\begingroup\raggedright\endgroup


\end{document}